\documentclass[12pt,a4paper,final]{iopart}

\usepackage{iopams}  
\usepackage[dvips]{graphicx}
\usepackage[breaklinks=true,colorlinks=true,linkcolor=blue,urlcolor=blue,citecolor=blue]{hyperref}
\expandafter\let\csname equation*\endcsname\relax
\expandafter\let\csname endequation*\endcsname\relax
\usepackage{amsmath}
\usepackage{ulem}
\usepackage{cancel}
\usepackage{amsthm}
\usepackage{mathtools}
\usepackage{tikz}
\usepackage{graphicx}
\usepackage{mathrsfs}
\usepackage{multirow}
\usepackage[mathscr]{euscript}
\graphicspath{ {./images/} }
\newtheorem{theorem}{Theorem}[section]
\newtheorem{corollary}{Corollary}[theorem]
\newtheorem{lemma}[theorem]{Lemma}
\newcommand{\vct}{\textbf}

\newcommand{\npg}{\vspace{\baselineskip}}

\newcommand{\npgni}{\npg\noindent}

\begin{document}
\title{NQG III - Two-Centre Problems, Whirlpool Galaxy and Toy Neutron Stars}

\author{Richard Durran and Aubrey Truman}
\address{Department of Mathematics, Computational Foundry, Swansea University Bay Campus, Fabian Way, Swansea, SA1 8EN, UK}
\ead{a.truman@swansea.ac.uk}

\begin{center}
"Beauty and Complexity as seen teetering on the Shoulders of Giants"\vspace{-3mm}
\end{center}

\begin{abstract}

In the hunt for WIMPish dark matter and testing our new theory, we extend the results obtained for the Kepler problem in NQG I and NQG II to the Euler two-centre problem and to other classical Hamiltonian systems with planar periodic orbits. In the first case our results lead to quantum elliptical spirals converging to elliptical orbits where stars and other celestial bodies can form as the corresponding WIMP/molecular clouds condense. The examples inevitably involve elliptic integrals as was the case in our earlier work on equatorial orbits of toy neutron stars (see Ref. [27]). Hence this is the example on which we focus in this work on quantisation. The main part of our analysis which leans heavily on Hamilton-Jacobi theory is applicable to any KLMN integrable planar periodic orbits for Hamiltonian systems. The most useful results on Weierstrass elliptic functions needed in these two works we have summarised with complete proofs in the appendix. This has been one of the most enjoyable parts of this research understanding in more detail the genius of Weierstrass and Jacobi. However we have to say that the beautiful simplicity of the Euler two-centre results herein transcend even this as far as we are concerned. At the end of the paper we see how the Burgers-Zeldovich fluid model relates to our set-up through Nelson's stochastic mechanics.
    
\end{abstract}

\section{Introduction}

\npgni Having seen how the asymptotics of Schrödinger wavefunctions for the Kepler problem can explain the formation of planets, planetary ring systems and the evolution of galaxies from spiral to elliptical, we now turn our attention to other astronomical problems. 

\npgni Firstly in the next section we discuss the quantisation of the Euler two-centre problem finding the corresponding gaussian wavefunction for an elliptical spiral orbit converging to Euler’s solution in the infinite time limit. When one combines this elliptical two-centre state with our astronomical elliptic states for the Kepler problem for each of the two centres, using our linearisation principle for the semi-classical mechanics, we shed new light on galaxies such as the Whirlpool suggesting new places to look for dark matter. Here we were very lucky to find a quantum Liouville condition (the analogue of the classical one) making our semi-classical equations easy to solve in terms of elliptic integrals. 

\npgni In Section 3 of this work we emphasise the role of “effective potentials” in simplifying the analysis of the roots of quartics originally due to Cardano et al. Although this analysis is somewhat intractable algebraically, when combined with ideas on “potential wells” for fixed energies, it does point up the importance of a convenient time-change, z. This is the “potential well-time”.

\npgni We first encountered this in our work on the KLMN problem for toy neutron stars and equatorial orbitals. In this set-up for P a unit mass, unit charged particle in the gravitational field of a $\mu$ unit point mass and electro-magnetic field due to a constant magnetic dipole moment $\vct{m}$ representing a neutron star centred at the origin $O$, rotating about $\hat{\vct{m}}$ as axis. For $\overrightarrow{OP}=\vct{r}$ and $\dfrac{d\vct{r}}{dt}=\dot{\vct{r}}$ etc., $t$ being the physical time,

$$\ddot{\vct{r}}=-\mu r^{-3}\vct{r}+\dot{\vct{r}}\times\vct{H},\;\;\;\;|\vct{r}|=r,\;\;\;\;\;(\textrm{KLMN})$$

\npgni with $\vct{H}=3(\vct{m}.\vct{r})r^{-5}\vct{r}\pm\vct{m}r^{-3}$, $\pm$ sign depending upon charge which we assume is positive. We assume that $\vct{m}.\vct{r}=0$, $\vct{m}=(0,0,B)$, giving the equation of the equatorial plane.

\npgni Taking the dot product with $\dot{\vct{r}}$ gives immediately,

$$\dfrac{d}{dt}\left(\dfrac{\dot{\vct{r}}^2}{2}-\dfrac{\mu}{r}\right)=0.$$

\npgni So for constant energy, $E=\dfrac{\dot{\vct{r}}^2}{2}-\dfrac{\mu}{r}$, and in the plane $\vct{m}.\vct{r}=0$, $\dfrac{\dot{\vct{r}}^2}{2}=\dfrac{\dot{r}^2}{2}+\dfrac{h^2}{2r^2}$, $\vct{h}=\vct{r}\times\dot{\vct{r}}$, so

$$\dfrac{\dot{r}^2}{2}+\dfrac{h^2}{2r^2}-\dfrac{\mu}{r}=E,\;\;\;\textrm{if}\;\;\vct{m}.\vct{r}=0.$$

\npgni Further, taking the vector product with $\vct{r}$ gives, for $\vct{m}.\vct{r}=0$,

$$\vct{r}\times\ddot{\vct{r}}=\pm r^{-3}(\vct{r}\times(\dot{\vct{r}}\times\vct{m}))=\pm r\dot{r}r^{-3}\vct{m},$$

\npgni $\hat{\vct{m}}=(0,0,1)$ in cartesians. So for + sign,

$$\dfrac{d\vct{h}}{dt}=\dfrac{d}{dt}(\vct{r}\times\dot{\vct{r}})=-\dfrac{d}{dt}\left(\dfrac{B}{r}\right)\hat{\vct{m}}$$

\npgni i.e. for a constant $C$, $h=|C-Bu|$, where $u=\dfrac{1}{r}$, $h=|\vct{h}|$.

\npgni Hence, we obtain

$$\dot{r}^2=2\left(E+\mu u-\dfrac{u^2}{2}(C-Bu)^2\right)=f(u),$$\vspace{-8mm}

\npgni $f$ being a quartic.

\npgni It only remains to establish the link with the Weierstrass elliptic function $\wp(z;g_{2},g_{3})$, $z=\displaystyle\int\limits^u\dfrac{du}{\sqrt{f(u)}}$, but first we remark:-

\npgni \textbf{Remark}\vspace{-3mm}

\npgni If $B$ is not a constant, but $B=B(r)$, the identity of the new constant of the motion, $C$, becomes

$$C=h-\displaystyle\int\limits^{u^{-1}}\dfrac{B(r)}{r^2}dr,\;\;\;\textrm{if}\;\;h>0,\;\;r=u^{-1},$$\vspace{-5mm}

\npgni and in this case\vspace{-5mm}

$$\dfrac{\dot{r}^2}{2}=E-V(u)-\dfrac{u^2}{2}\left(C+\displaystyle\int\limits^{u^{-1}}\dfrac{B(r)}{r^2}dr\right)^2,$$\vspace{-5mm}

\npgni where $V(u)$ is the potential energy function.

\npgni When the right hand side is still a quartic polynomial we call this KLMN integrable. Inevitably this involves the possibility of magnetic monopole density being non-zero. (See Ref. [27]). Alternatively the vector product term in the (KLMN) equation could be a Coriolis type force such as envisaged in Burgers-Zeldovich fluids.

\npgni The related elliptic integral is $\displaystyle\int\limits_{u_{0}}^u\dfrac{du}{\sqrt{f(u)}}=z$, $z$ being one of the “well-times” for the quartic $f$, with $f(u_{0})=0$, satisfying mild regularity conditions e.g. $\bigtriangleup\ne0$, $\bigtriangleup=\bigtriangleup(f)$, the corresponding discriminant, giving for $f_{0}'=f'(u_{0})$ etc.

$$(u-u_{0})=-\dfrac{f_{0}'}{4\left(\wp(z;g_{2},g_{3})-\dfrac{f_{0}''}{24}\right)},$$

\npgni $\wp$ being Weierstrass' elliptic function, with $g_{2}$, $g_{3}$ the quartic invariants of $f$.

\npgni This result is the vital ingredient in establishing the uniformisation of curves of genus unity, with algebraic equations $y^2=f(x)$, so useful in Ref. [27].

\npgni The analogue of this result, when $u_{0}$ is known but not known to be a root, is the most powerful tool in making progress with solving our equations and simplifying the algebraic complexities. This was also the case for the KLMN problem, Ref. [27]. This more advanced counter part, stated below, we have proved in the appendix. The complete proof is new to us at least, even after trawling through the standard original references. Needless to say it dwarfs our example above. (See Whittaker and Watson, Ref. [26] and Biermann, Ref. [4]).

\begin{theorem}

If $z=z(u)=\int\limits_{u_{0}}^{u}\{f(u)\}^{-\frac{1}{2}}du$, $u>u_{0}$, where $f(u)$ is a quartic polynomial with no repeated factors, then 

$$u=u_{0}+\frac{\{f(u_{0})\}^{\frac{1}{2}}\wp'(z)+\frac{1}{2}\left(\wp(z)-\frac{1}{24}f''(u_{0})\right)f'(u_{0})+\frac{1}{24}f(u_{0})f'''(u_{0})}{2\left(\wp(z)-\frac{1}{24}f''(u_{0})\right)^2-\frac{1}{48}f(u_{0})f^{''''}(u_{0})},$$

\npgni where $\wp(z)=\wp(z;g_{2},g_{3})$ is the Weierstrass function formed with the invariants $g_{2}$ and $g_{3}$ of the quartic $f$.
 
\end{theorem}

\npgni The main motivation in our last reference was to help lay the foundations for a mathematical understanding of classical particle motions in the force field of a neutron star. In Section 4 we carry this programme forward to calculate the exact polar equations of possible periodic equatorial orbitals, revealing some of the difficulties in trying to quantise this system as an example of our earlier work. Here Hamilton-Jacobi methods come to our rescue as does the early work of Einstein on Brownian motion. Our results are easily generalised to any KLMN system having planar periodic orbits which is KLMN integrable. 

\npgni In Section 5 we start with a Schrödinger formulation and using the Hopf-Cole transformation finish with a Burgers-Zeldovich fluid model, mirroring the behaviour of our Schrödinger set-up. The resulting solution of the Burgers-Zeldovich equation reveals it as a “sum over paths” in a Feynman-Kac formula. This is an extension of the elementary formula of Elworthy and Truman, Ref. [11], incorporating Born's probabilistic interpretation. It is intriguing and gratifying that the paths are essentially the sample paths of the Nelson diffusion process corresponding to the original Schrödinger equation satisfying a Nelson-Newton law of force. The singularities of these equations link our results to the beautiful work on caustics of Berry and Arnol'd. (See Refs. [2], [3], [1], [8] and [9]). This result is of enormous generality extending to include Hamiltonian systems with vector potentials as well as the usual ones and, more importantly, independent noise terms.

\npgni We hope that this paper leads to greater cooperation and collaboration with astronomers not least in explaining some of the most striking photographs taken by Hubble. Firstly, we see this and related works as a tribute to the founders of Natural Philosophy but especially to Newton nearly 300 years after his death in 1727.

\section{On the Semi-Classical Mechanics of Euler's Two-Centre Problem and a Quantum Elliptical Spiral Connecting Two Focal Centres}

\subsection{Semi-Classical Analysis of Two Centres}

\npgni We use elliptical coordinates $(\xi,\eta)$, where for our unit mass particle P, with\vspace{-5mm}

\npgni $\overrightarrow{OP}=(X,Y)$,

$$X+iY=c\;\textrm{cosh}(\xi+i\eta),$$

\npgni $(X,Y)$ being cartesian coordinates with origin $O$ at the mid-point of the two centres (at $\textrm{A}_{12}$, $(\pm c,0)$). Our putative elliptical spiral will converge to an ellipse $\mathcal{E}_{2c}$, centred at O with $A$ the semi-major axis, $Ae=c$, $e$ the eccentricity, as we shall see.

\npgni From the above, $\dfrac{z}{c}=\dfrac{\textrm{e}^w+\textrm{e}^{-w}}{2}$, $z=X+iY$, $w=\xi+i\eta$, giving

$$\textrm{e}^w=\dfrac{z}{c}\pm\sqrt{\dfrac{z^2}{c^2}-1}=\dfrac{z}{c}\pm\sqrt{\left(\dfrac{z}{c}-1\right)\left(\dfrac{z}{c}+1\right)},$$

\npgni so the r.h.s. has a square root branch point requiring a cut $z$ plane, with cut going from $(-c+0i)$ to $(+c+0i)$. On this cut,

$$\textrm{e}^{\xi+i\eta}=\dfrac{X}{c}\pm i\sqrt{1-\dfrac{X^2}{c^2}}\;,\;\;\;|X|<c,$$

\npgni so $\xi=0$ on the cut and $\eta=\textrm{arctan}\left(\dfrac{\sqrt{c^2-X^2}}{X}\right)$.

\npgni Our square root can now be defined by analytic continuation in the $z$ plane. There are no other singularities, because $\dfrac{z}{c}\pm\sqrt{\dfrac{z^2}{c^2}-1}\ne0$. We just have to work with the multiple-valued $\eta$ and the cut.

\npgni Defining $\sqrt{a+ib}=\alpha+i\beta\implies a+ib=\alpha^2-\beta^2+2i\alpha\beta$, where in our case

$$a=\dfrac{X^2-Y^2-c^2}{c^2},\;\;\;b=\dfrac{2XY}{c^2},\;\;\;a=\alpha^2-\beta^2,\;\;\;b=2\alpha\beta,$$

\npgni giving $\alpha^4-a\alpha^2-\dfrac{b^2}{4}=0$  i.e.  $\alpha^2=\dfrac{a\pm\sqrt{a^2+b^2}}{2}$, $\beta=\dfrac{b}{2\alpha}$, so there is just this ambiguity, from the cut, in sign for $\alpha$ and $\beta$, giving $(\xi,\eta)$ explicitly as

$$\xi+i\eta=\ln{\left(\dfrac{(X+iY)\pm(\alpha+i\beta)}{c}\right)},\;\;\;\textrm{for the above}\;\alpha\;\textrm{and}\;\beta.$$

\npgni Obviously the coordinates $\xi$ and $\eta$ are orthogonal with scale factors,

$$h_{\xi}=h_{\eta}=c\sqrt{\textrm{cosh}^2\xi-\cos^2\eta}\;,$$

$$\bigtriangleup_{w}=\dfrac{1}{c^2(\textrm{cosh}^2\xi-\cos^2\eta)}\left(\dfrac{\partial^2}{\partial\xi^2}+\dfrac{\partial^2}{\partial\eta^2}\right).$$

\npgni Level surfaces of $\xi$ and $\eta$ are ellipses and hyperbolae:-

$$\dfrac{X^2}{c^2\textrm{cosh}^2\xi}+\dfrac{Y^2}{c^2\textrm{sinh}^2\xi}=1\;\;\;\textrm{and}\;\;\;\dfrac{X^2}{c^2\cos^2\eta}-\dfrac{Y^2}{c^2\sin^2\eta}=1,$$

\npgni intersecting orthogonally. Here we concentrate on finding the asymptotics of our Schrödinger wavefunction corresponding to classical motion on the ellipse, $\xi=\xi_{0}$, a constant, the ellipse having $A_{1}$ and $A_{2}$ as fixed foci.

\npgni We have to find the stationary state wavefunction for the Schrödinger equation, $\epsilon^2=\dfrac{\hbar}{m}$,

$$-\dfrac{\epsilon^4}{2}\bigtriangleup_{w}\psi+V\psi=E\psi,\;\;\;w=\xi+i\eta,$$

\npgni where $V$ is the Euler two-centre potential for centres $\textrm{A}_{1}$ and $\textrm{A}_{2}$, for particle P,

$$V=-\dfrac{\mu_{1}}{r_{1}}-\dfrac{\mu_{2}}{r_{2}}.$$

\npgni $\overrightarrow{\textrm{OP}}=(X,Y)$, $\overrightarrow{\textrm{OA}}_{1}=(c,0)$, $\overrightarrow{\textrm{OA}}_{2}=(-c,0)$, $r_{1}=d(\textrm{A}_{1},\textrm{P})$, $r_{2}=d(\textrm{A}_{2},\textrm{P})$, where $X=X(\xi,\eta)$ and $Y=Y(\xi,\eta)$. It is easy to prove that

$$r_{1}=c(\textrm{cosh}\xi-\cos\eta),\;\;\;r_{2}=c(\textrm{cosh}\xi+\cos\eta).$$

\npgni So we have to solve

$$-\dfrac{\epsilon^4}{2}\bigtriangleup_{w}\psi-\dfrac{1}{c}\left(\dfrac{\mu_{1}}{\textrm{cosh}\xi-\cos\eta}+\dfrac{\mu_{2}}{\textrm{cosh}\xi+\cos\eta}\right)\psi=E\psi$$

\npgni in the semi-classical limit. Writing $\psi\sim\exp\left(\dfrac{R+iS}{\epsilon^2}\right)$ gives

$$2^{-1}(|\boldsymbol{\nabla}S|^2-|\boldsymbol{\nabla}R|^2)-\dfrac{\epsilon^2}{2}\bigtriangleup R+V=E,$$

\npgni and

$$\boldsymbol{\nabla}R.\boldsymbol{\nabla}S+\dfrac{\epsilon^2}{2}\bigtriangleup S=0,$$

\npgni with familiar zeroth order approximations:-

$$2^{-1}(|\boldsymbol{\nabla}S|^2-|\boldsymbol{\nabla}R|^2)+V=E,$$

\npgni and

$$\boldsymbol{\nabla}R.\boldsymbol{\nabla}S=0.$$

\npgni (See NQG I and NQG II, Refs. [28] and [29]).

\npgni Assume just for now that the zeroth order approximations are known $(R_{0},S_{0})$ and write: $R=R_{0}+\epsilon^2R_{1}$, $S=S_{0}+\epsilon^2S_{0}$, where $R_{1}$ and $S_{1}$ satisfy

$$\boldsymbol{\nabla}S_{0}.\boldsymbol{\nabla}S_{1}-\boldsymbol{\nabla}R_{0}.\boldsymbol{\nabla}R_{1}=\dfrac{1}{2}\bigtriangleup R_{0},$$

$$\boldsymbol{\nabla}S_{0}.\boldsymbol{\nabla}R_{1}+\boldsymbol{\nabla}R_{0}.\boldsymbol{\nabla}S_{1}=-\dfrac{1}{2}\bigtriangleup S_{0}.$$

\npgni In our case we shall show that for the specific energy $E=-\dfrac{(\mu_{1}+\mu_{2})^2}{4\alpha^2}$, where Euler's constant $\gamma=\dfrac{\alpha^2}{c^2}$ (see Whittaker, Ref. [25]), $R_{0}$ is given by

$$R_{0}=R_{0}(\xi)=\int\limits^\xi v_{0}(\xi)d\xi,\;\;v_{0}^2(\xi)=\dfrac{c^2(\mu_{1}+\mu_{2})}{2\alpha^2}\left(\textrm{cosh}\xi-\dfrac{2\alpha^2}{c(\mu_{1}+\mu_{2})}\right)^2.$$

\npgni and $S_{0}$ by

$$S_{0}=S_{0}(\eta)=\int\limits^\eta u_{0}(\eta)d\eta,\;\;u_{0}^2(\eta)=\dfrac{c^2(\mu_{1}+\mu_{2})}{2\alpha^2}\left(\cos\eta+\dfrac{2\alpha^2(\mu_{1}-\mu_{2})}{(\mu_{1}+\mu_{2})^2}\right)^2+\dfrac{8\alpha^2\mu_{1}\mu_{2}}{(\mu_{1}+\mu_{2})^2}.$$

\npgni Now $R_{1}=R_{1}(\xi,\eta)$ and $S_{1}=S_{1}(\xi,\eta)$ have to satisfy

$$u_{0}\dfrac{\partial S_{1}}{\partial\eta}-v_{0}\dfrac{\partial R_{1}}{\partial\xi}=\dfrac{1}{2}\dfrac{\partial^2R_{0}}{\partial\xi^2}=\dfrac{1}{2}\dfrac{dv_{0}}{d\xi},$$

$$u_{0}\dfrac{\partial R_{1}}{\partial\eta}+v_{0}\dfrac{\partial S_{1}}{\partial\xi}=-\dfrac{1}{2}\dfrac{\partial^2S_{0}}{\partial\eta^2}=-\dfrac{1}{2}\dfrac{du_{0}}{d\eta},$$

\npgni with relevant solutions,

$$\dfrac{\partial R_{1}}{\partial\xi}=v_{1},\;\;\;\dfrac{\partial R_{1}}{\partial\eta}=u_{1},\;\;\;\dfrac{\partial S_{1}}{\partial\xi}=0,\;\;\;\dfrac{\partial S_{1}}{\partial\eta}=0,$$

\npgni where $v_{1}=-\dfrac{1}{2v_{0}}\dfrac{dv_{0}}{d\xi}$ and $u_{1}=-\dfrac{1}{2u_{0}}\dfrac{du_{0}}{d\eta}$, $v_{0}$, $u_{0}$ defined above.

\npgni In this case $R$ and $S$ are approximated by

$$R=\int\limits^\xi v_{0}d\xi+\epsilon^2\Big{(}\int\limits^\xi v_{1}d\xi+\int\limits^\eta u_{1}d\eta\Big{)}\;\;;\;\;S=\int\limits^\eta u_{0}d\eta$$

\npgni Our underlying dynamical system,

$$\dfrac{d\vct{X}_{t}}{dt}=\boldsymbol{\nabla}(R+S)(\vct{X}_{t}),$$

\npgni reduces to

$$\dot{\xi}=\dfrac{v_{0}+\epsilon^2v_{1}}{c(\textrm{cosh}^2\xi-\cos^2\eta)},$$

$$\dot{\eta}=\dfrac{u_{0}+\epsilon^2u_{1}}{c(\textrm{cosh}^2\xi-\cos^2\eta)},$$

\npgni where $v_{0}$, $v_{1}$, $u_{0}$ and $u_{1}$ are given above.

\npgni It only remains to find $v_{0}=v_{0}(\xi)$ and $u_{0}=u_{0}(\eta)$, which is easy because miraculously in the Euler two-centre problem for our coordinates, $z=f(w)$, $X+iY=f(\xi+i\eta)$, we obtain the separation of variables, $R=R(\xi)$, $S=S(\eta)$, and more importantly, we have a quantum Liouville property,

$$2|f'(w)|^2(V-E)=v_{0}^2(\xi)-u_{0}^2(\eta)=\left(\dfrac{\partial R}{\partial\xi}\right)^2-\left(\dfrac{\partial S}{\partial\eta}\right)^2\;\;\;\;\;(\textrm{DQLC}),$$

\npgni $\dfrac{\partial R}{\partial\xi}=0$, having a repeated root at $\xi=\xi_{0}$ for our special value of $E$. Here $\boldsymbol{\nabla}R.\boldsymbol{\nabla}S=0$ automatically from the Cauchy-Riemann equations.

\npgni We refer to this identity as DQLC, "Durran Quantum Liouville Condition".

\subsection{Underlying Classical Mechanics of Euler's Two Centre Problem}

\npgni We have seen that the potential energy for two centres is

$$V=-\dfrac{1}{c}\left(\dfrac{\mu_{1}}{\textrm{cosh}\xi-\cos\eta}+\dfrac{\mu_{2}}{\textrm{cosh}\xi+\cos\eta}\right).$$

\npgni A simple calculation yields the kinetic energy $T$ as

$$T=2^{-1}(\dot{\vct{X}}^2+\dot{\vct{Y}}^2)=\dfrac{c^2}{2}(\textrm{cosh}^2\xi-\cos^2\eta)(\dot{\xi}^2+\dot{\eta}^2).$$

\npgni So the problem is of classical Liouville type. Here we give a brief account of Euler's solution following Whittaker, Ref. [25]. Starting with the Lagrange equation,

$$\dfrac{d}{dt}\left(\dfrac{\partial\mathcal{L}}{\partial\dot{\xi}}\right)=\dfrac{\partial\mathcal{L}}{\partial\xi},\;\;\;\mathcal{L}=T-V,$$

\npgni we arrive at

$$c^2\dfrac{d}{dt}\left((\textrm{cosh}^2\xi-\cos^2\eta)\dot{\xi}^2\right)=-2c^2\textrm{cosh}\xi\textrm{sinh}\xi(\textrm{cosh}^2\xi-\cos^2\eta)\dot{\xi}(\dot{\xi}^2+\dot{\eta}^2)$$

\npgni \hspace{61mm} $=-2(\textrm{cosh}^2\xi-\cos^2\eta)\dot{\xi}\dfrac{\partial V}{\partial\xi}.$

\npgni Since $T+V=E$,

$$c^2\dfrac{d}{dt}\left((\textrm{cosh}^2\xi-\cos^2\eta)\dot{\xi}^2\right)=-2(\textrm{cosh}^2\xi-\cos^2\eta)\dot{\xi}\dfrac{\partial V}{\partial\xi}+2(E-V)\dot{\xi}\dfrac{\partial}{\partial\xi}(\textrm{cosh}^2\xi-\cos^2\eta)$$

\npgni \hspace{47mm} $=2\dot{\xi}\dfrac{\partial}{\partial\xi}\Big{(}(E-V)(\textrm{cosh}^2\xi-\cos^2\eta)\Big{)}$

\npgni \hspace{47mm} $=2\dfrac{d}{dt}\Big{(}E\textrm{cosh}^2\xi-\dfrac{(\mu_{1}+\mu_{2})}{c}\textrm{cosh}\xi\Big{)}$.

\npgni Integration gives

$$\dfrac{c^2}{2}(\textrm{cosh}^2\xi-\cos^2\eta)^2\dot{\xi}^2=E\textrm{cosh}^2\xi-\dfrac{(\mu_{1}+\mu_{2})}{c}\textrm{cosh}\xi-\gamma,$$

\npgni $\gamma$ being Euler's constant of integration. Subtracting from the energy equation gives

$$\dfrac{c^2}{2}(\textrm{cosh}^2\xi-\cos^2\eta)^2\dot{\eta}^2=-E\cos^2\eta-\dfrac{(\mu_{2}-\mu_{1})}{c}\cos\eta+\gamma,$$

\npgni yielding

$$(d\zeta)^2=\dfrac{(d\xi)^2}{\left(E\textrm{cosh}^2\xi-\dfrac{(\mu_{1}+\mu_{2})}{c}\textrm{cosh}\xi-\gamma\right)}=\dfrac{(d\eta)^2}{\left(-E\cos^2\eta-\dfrac{(\mu_{2}-\mu_{1})}{c}\cos\eta+\gamma\right)},$$

\npgni where $d\zeta$ is our time change,

$$d\zeta=\dfrac{c}{\sqrt{2}(\textrm{cosh}^2\xi-\cos^2\eta)}dt=\dfrac{c}{\sqrt{2}|f'(\xi+i\eta)|^2}dt.$$

\npgni Rewriting in terms of the new time-changed variable $\zeta$ the last equation for $\xi$ gives

$$d\zeta=\dfrac{d\xi}{\sqrt{E\textrm{cosh}^2\xi-\dfrac{(\mu_{1}+\mu_{2})}{c}\textrm{cosh}\xi-\gamma}}=\dfrac{d(\textrm{cosh})\xi}{\textrm{sinh}\xi\sqrt{E\textrm{cosh}^2\xi-\dfrac{(\mu_{1}+\mu_{2})}{c}\textrm{cosh}\xi-\gamma}},$$

\npgni i.e. setting $t=\textrm{cosh}\xi$,

$$\zeta=\int\limits_{t_{0}}^{\textrm{cosh}\xi}\dfrac{dt}{\sqrt{Q_{\xi}(t)}}\implies\textrm{cosh}\xi(t)-\textrm{cosh}\xi(0)=\dfrac{Q'_{\xi}(t_{0})}{4\left(\wp(\zeta;g_{2},g_{3})-\frac{1}{24}Q''_{\xi}(t_{0})\right)},$$

\npgni where $t_{0}=\textrm{cosh}\xi(0)$ is a root of $Q_{\xi}(.)=0$ and $Q_{\xi}(t)=(t^2-1)\left(Et^2-\dfrac{(\mu_{1}+\mu_{2})}{2}t-\gamma\right)$ with quartic invariants $g_{2}$, $g_{3}$. Repeating this same argument for $\eta$ gives us the equation of the orbit for the Euler two-centre problem in terms of Weierstrass elliptic functions and quartic invariants $g_{2}$, $g_{3}$ of $Q_{\xi}$ and $Q_{\eta}$.

\begin{lemma}
Let $f(x)$ be the quartic with real coefficients $a_{i}$, $i=0,1,2,3,4$,

$$f(x)=a_{0}x^4+4a_{1}x^3+6a_{2}x^2+4a_{3}x+a_{4},$$

\npgni with simple roots, one of which is $x_{0}\in\mathbb{R}$. Then, if $g_{2}$, $g_{3}$ are the quartic invariants,

$$g_{2}=a_{0}a_{4}-4a_{1}a_{3}+3a_{2}^2,\;\;\;\;
g_{3}=\begin{vmatrix}
a_{0} & a_{1} & a_{2}\\
a_{1} & a_{2} & a_{3}\\
a_{2} & a_{3} & a_{4}
\end{vmatrix}\;\;\;\;
\textrm{and}\;\;\;\;
z=\int\limits_{x_{0}}^x\dfrac{dx}{\sqrt{f(x)}},$$

\npgni it follows that

$$(x-x_{0})=\dfrac{f'(x_{0})}{4\left(\wp(z;g_{2},g_{3})-\dfrac{f''(x_{0})}{24}\right)}.$$

\end{lemma}

\begin{proof}
Taylor's theorem gives for $f_{0}=f(x_{0})$ and $f'_{0}=f'(x_{0})$ etc.,

$$f(x)=f'_{0}(x-x_{0})+\dfrac{1}{2}f''_{0}(x-x_{0})^2+\dfrac{1}{6}f'''_{0}(x-x_{0})^3+\dfrac{1}{24}f''''_{0}(x-x_{0})^4.$$

\npgni After making the substitution

$$(x-x_{0})=\dfrac{f'_{0}}{4\left(y-\dfrac{f''_{0}}{24}\right)},$$

\npgni we obtain

$$\int\limits_{x_{0}}^x\dfrac{dx}{\sqrt{f(x)}}=-\dfrac{f'_{0}}{4}\int\limits_{\infty}^{y(x)}\dfrac{dy}{\sqrt{\left(y-\frac{f''_{0}}{24}\right)^4f(x)}}$$

\npgni i.e.

$$\int\limits_{x_{0}}^x\dfrac{dx}{\sqrt{f(x)}}=-\dfrac{f'_{0}}{4}\int\limits_{\infty}^{y(x)}\dfrac{dy}{\sqrt{\frac{f_{0}'^{2}}{4}\left(y-\frac{f''_{0}}{24}\right)^3+\frac{f_{0}''f_{0}'^{2}}{32}\left(y-\frac{f''_{0}}{24}\right)^2+\dotsc}}.$$

\npgni Expanding inside the square root on the r.h.s. reveals:

$$\textrm{coeff. of}\;\;y^3=\dfrac{f_{0}'^2}{4};\;\;\;\textrm{coeff. of}\;\;y^2=0;\;\;\;\textrm{coeff. of}\;\;y=\dfrac{f_{0}'''f_{0}'^3}{384}-\dfrac{f_{0}'^2f_{0}''^2}{768}$$

\npgni i.e. $\textrm{coeff. of}\;\;y=-\left(\dfrac{f_{0}'}{4}\right)^2g_{2}(f)$, $g_{2}(f)$ being the Taylor series version of $g_{2}$. Similarly

$$\textrm{coeff. of}\;\;y^0=-\left(\dfrac{f_{0}'}{4}\right)^2\left(-\dfrac{f_{0}''^3}{1728}+\dfrac{f_{0}'''f_{0}''f_{0}'}{576}-\dfrac{f_{0}''''f_{0}'^2}{384}\right)=-\left(\dfrac{f_{0}'}{4}\right)^2g_{3}(f).$$

\npgni But invariance now gives us that $g_{2}(f)=g_{2}(a)$ and $g_{3}(f)=g_{3}(a)$ as can be seen below.

\npgni Following Copson, Ref. [7], consider how

$$f=a_{0}x^4+4a_{1}x^3y+6a_{2}x^2y^2+4a_{3}xy^3+a_{4}y^4,$$

\npgni behaves under the change of variables, $x=lX+mY$, $y=l'X+m'Y$, for

\npgni $\bigtriangleup=(lm'-l'm)\ne0$, then as is well known the quartic invariants,

$$g_{2}=a_{0}a_{4}-4a_{1}a_{3}+3a_{2}^2=(A_{0}A_{4}-4A_{1}A_{3}+3A_{2}^2)\bigtriangleup^{-4},$$

$$g_{3}=\begin{vmatrix}
a_{0} & a_{1} & a_{2}\\
a_{1} & a_{2} & a_{3}\\
a_{2} & a_{3} & a_{4}
\end{vmatrix}=\begin{vmatrix}
A_{0} & A_{1} & A_{2}\\
A_{1} & A_{2} & A_{3}\\
A_{2} & A_{3} & A_{4}
\end{vmatrix}\bigtriangleup^{-6},$$

\npgni where

$$f=A_{0}X^4+4A_{1}X^3Y+6A_{2}X^2Y^2+4A_{3}XY^3+A_{4}Y^4.$$

\end{proof}

\npgni This result leads to the next lemma which we call Copson's Lemma. (See Copson, Ref. [7], $\S$13.7).

\begin{lemma}
Let $\phi(t)=a_{0}t^4+4a_{1}t^3+6a_{2}t^2+4a_{3}t+a_{4}$, so if $t=\dfrac{x}{y}$ (and $t_{0}$ satisfies $\phi(t_{0})=0)$

$$y^4\phi(t)=f(x,y)=a_{0}x^4+4a_{1}x^3y+6a_{2}x^2y^2+4a_{3}xy^3+a_{4}y^4.$$\vspace{-8mm}

\npgni Then

$$f(x,y)=\dfrac{Y^4(4w^3-g_{2}w-g_{3})}{A_{1}^2}=4A_{1}X^3Y+4A_{3}XY^3+A_{4}Y^4,$$

\npgni where $\dfrac{X}{Y}=\dfrac{w}{A_{1}}$, $A_{1}=-\dfrac{1}{4}\phi_{0}'=-\dfrac{1}{4}\phi'(t_{0})$, $A_{3}=-\dfrac{g_{2}}{4A_{1}}$ and $A_{4}=-\dfrac{g_{3}}{A_{1}^2}$, for

$$x=t_{0}(X+\lambda Y)-Y,\;\;\;\;y=X+\lambda Y$$\vspace{-8mm}

\npgni i.e.

$$t-t_{0}=-\dfrac{A_{1}}{\left(\dfrac{A_{1}X}{Y}+\lambda A_{1}\right)}=-\dfrac{1}{\left(\dfrac{X}{Y}+\lambda\right)},\;\;\;6\lambda=\dfrac{\phi_{0}''}{\phi_{0}'},\;\;\phi_{0}''=\phi''(t_{0}).$$

\npgni The all important point is that

$$\dfrac{x}{y}=\dfrac{t_{0}X+(t_{0}\lambda-1)Y}{X+\lambda Y},$$

\npgni so the corresponding $(lm'-l'm)=t_{0}\lambda-(t_{0}\lambda-1)=1$.

\end{lemma}

\npgni (Evidently finding the roots of the quartic is an important problem to be solved in this context. Hence, our inclusion of what we call $\lambda$-analysis).

\npgni Let us review the last result on the asymptotics of Schrödinger wavefunctions for gravitational two-centre problems in the context of known classical results for this set-up. Firstly, what we have done is to find a simple solution of the semi-classical mechanics for this problem embodied in the putative wavefunction, $\psi\sim\textrm{exp}\left(\frac{R+iS}{\epsilon^2}\right)$ as $\epsilon\sim0$, for above $R$ and $S$. Namely, for the two-centre problem, for $\vct{X}=\overrightarrow{OP}$, $O$ the origin as mid-point of the two centres $A_{1}$ and $A_{2}$, $P$ our unit mass particle subject to potential energy forces for potential $V$,

$$V(\vct{X})=-\dfrac{\mu_{1}}{r_{1}}-\dfrac{\mu_{2}}{r_{2}},$$

\npgni where $r_{i}=d(A_{i},P)$ for $i=1,2$,  $\overrightarrow{OA_{i}}=(\pm c,0)$, for $i=1,2$,

$$2^{-1}(|\boldsymbol{\nabla}S|^2-|\boldsymbol{\nabla}R|^2)+V=E,\;\;\;\;\boldsymbol{\nabla}R.\boldsymbol{\nabla}S=0,$$

\npgni $E$ being the energy, so that if $\dot{\vct{X}}_{t}=\boldsymbol{\nabla}(R+S)(\vct{X}_{t}),\;\;t\ge0$, $\dot{\vct{X}}_{t}=\dfrac{d\vct{X}_{t}}{dt}$,

$$\dfrac{d^2\vct{X}_{t}}{dt^2}=-\boldsymbol{\nabla}(V-|\boldsymbol{\nabla}R|^2)(\vct{X}_{t}),\;\;t\ge0,$$

\npgni and

$$\dfrac{d}{dt}R(\vct{X}_{t})=|\boldsymbol{\nabla}R|^2(\vct{X}_{t}),\;\;t\ge0.$$

\npgni Therefore, $R$ is monotonic increasing in time $t$, $R(\vct{X}_{t})\nearrow R_{\textrm{max}}$, the maximum being attained on our ellipse $\xi=\xi_{0}$ as our elliptical spiral, $\vct{X}_{t}$, converges to where $|\boldsymbol{\nabla}R|^2=0$. The only singularity here being on the line of join of the two foci $A_{1}$ and $A_{2}$ and only emerging from our elliptical solutions.

\npgni Here to within an additive constant

$$R(\xi)=\pm c\sqrt{\dfrac{\mu_{1}+\mu_{2}}{2\alpha^2}}\left(\sinh{\xi}-\dfrac{2\alpha^2\xi}{c(\mu_{1}+\mu_{2})}\right),$$

\npgni where we choose the sign so that $R$ achieves its maximum at $\xi=\xi_{0}$,

$$\cosh\xi_{0}=\dfrac{2\alpha^2}{c(\mu_{1}+\mu_{2})},$$

\npgni $\gamma=\dfrac{\alpha^2}{c^2}$, $\gamma$ being Euler/Whittaker constant.

\npgni The corresponding $S$ function is given by

$$S(\eta)=\pm\int c\sqrt{\dfrac{\mu_{1}+\mu_{2}}{2\alpha^2}}\sqrt{\left(\cosh\eta+2\alpha^2\left(\dfrac{\mu_{1}-\mu_{2}}{\mu_{1}+\mu_{2}}\right)\right)^2+\dfrac{16\alpha^4\mu_{1}\mu_{2}}{c^2(\mu_{1}+\mu_{2})^3}}\;d\eta.$$

\npgni $S$ is clearly dependent on elliptic integrals (see Byrd and Friedman, Ref. [5]). We shall elaborate on this later.

\npgni For now notice that generally the roots, $\cosh\xi$, of the quartic equation,

$$E\cosh^2\xi+\dfrac{(\mu_{1}+\mu_{2})}{c}\cosh\xi-\gamma=0,$$

\npgni determine the elliptical boundaries of the annular region to which our particle motion in the classical potential $V$ is confined. The condition for this classical orbit to be elliptical is that the above quadratic has equal roots, the ellipse having the two centres as foci. Setting $\gamma=\dfrac{\alpha^2}{c^2}$, we obtain

$$\gamma=\dfrac{(\mu_{1}+\mu_{2})}{2c}\cosh\xi_{0},$$

\npgni $\cosh\xi_{0}$ our repeated root and the energy $E$ is given by

$$E=-\dfrac{(\mu_{1}+\mu_{2})}{2A},$$

\npgni where $A$ is the semi-major axis of our ellipse and $c=Ae$, $e$ the eccentricity of our ellipse,

$$e=\dfrac{1}{\cosh\xi_{0}}.$$

\npgni It follows that our result is consistent with Bonnet's theorem in that the particle velocity $\vct{v}$ in the two centre forces is given by

$$v^2=v_{1}^2+v_{2}^2,$$

\npgni $\vct{v}_{1}$, $\vct{v}_{2}$ being the particle velocities when describing the same ellipse but subject to a force due to just one of the two centres. Namely,

$$v_{i}^2=2\left(E_{i}+\dfrac{\mu_{i}}{r_{i}}\right),\;\;i=1,2,\;\;\;E_{i}=-\dfrac{\mu_{i}}{2A},\;\;i=1,2.$$

\npgni For the existence of our quantum spiral converging to our two-centre ellipse the crucial condition is the quantum Liouville condition alluded to earlier, taking into account the last few remarks.

\subsection{Semi-Classical Analysis of Two Gravitational Centres with a Central Linear Restoring Force}

\npgni The above approach can be used to examine the motion of a particle P subject to the force of two fixed gravitational centres and a central linear restoring force. Needless to say such a force results from an all-enveloping spherically symmetric gravitational cloud. In the $(\xi,\eta)$ coordinates the potential field for this system is given by

$$V=-\dfrac{1}{c}\left(\dfrac{\mu_{1}}{\textrm{cosh}\xi-\cos\eta}+\dfrac{\mu_{2}}{\textrm{cosh}\xi+\cos\eta}\right)+\dfrac{1}{2}\omega^2c^2(\textrm{cosh}^2\xi+\cos^2\eta-1),$$.

\npgni where $\omega$ is a measure of the linear restoring force. To establish our semi-classical orbits we need to solve

$$-\dfrac{\epsilon^4}{2}\bigtriangleup_{w}\psi+\left(\dfrac{1}{2}\omega^2c^2(\textrm{cosh}^2\xi+\cos^2\eta-1)-\dfrac{\mu_{1}}{c(\textrm{cosh}\xi-\cos\eta)}-\dfrac{\mu_{2}}{c(\textrm{cosh}\xi+\cos\eta)}\right)\psi=E\psi.$$

\npgni Using the same methods as for the previous case we see that the solutions to our zeroth order semi-classical equations

$$2^{-1}(|\boldsymbol{\nabla}S|^2-|\boldsymbol{\nabla}R|^2)+V=E\;\;\;\textrm{and}\;\;\;\boldsymbol{\nabla}R.\boldsymbol{\nabla}S=0,$$

\npgni are given by

$$R_{0}=R_{0}(\xi)=\int\limits^\xi v_{0}(\xi)d\xi\;\;\;\textrm{and}\;\;\;S_{0}=S_{0}(\eta)=\int\limits^\eta u_{0}(\eta)d\eta,$$

\npgni where

$$v_{0}^2=\omega^2c^4\textrm{cosh}^4\xi-c^2(2E+\omega^2c^2)\textrm{cosh}^2\xi-2c(\mu_{1}+\mu_{2})\textrm{cosh}\xi+\gamma^2,$$\vspace{-8mm}

\npgni and

$$u_{0}^2=\omega^2c^4\cos^4\eta-c^2(2E+\omega^2c^2)\cos^2\eta+2c(\mu_{1}-\mu_{2})\cos\eta+\gamma^2.$$

\npgni As before our underlying semi-classical dynamical system,

$$\dfrac{d\vct{X}_{t}}{dt}=\boldsymbol{\nabla}(R+S)(\vct{X}_{t}),$$\vspace{-8mm}

\npgni reduces to

$$\dot{\xi}=\dfrac{v_{0}}{c(\textrm{cosh}^2\xi-\cos^2\eta)},\;\;\;\dot{\eta}=\dfrac{u_{0}}{c(\textrm{cosh}^2\xi-\cos^2\eta)}.$$

\npgni To illustrate this system we look at the special case when $E=-\dfrac{\omega^2c^2}{2}$ and $\gamma^2=3a^4\omega^2c^4$. For these values $v_{0}$ has the particularly simple form

$$v_{0}^2=\omega^2c^4(\textrm{cosh}\xi-a)^2\left((\textrm{cosh}\xi+a)^2+2a^2\right),$$

\npgni where $a=\dfrac{1}{c}\left(\dfrac{\mu_{1}+\mu_{2}}{2\omega^2}\right)^{\frac{1}{3}}$. The equation for $u_{0}$ reads\vspace{5mm}

$$u_{0}^2=\omega^2c^4\cos^4\eta+2c(\mu_{1}-\mu_{2})\cos\eta+3a^4\omega^2c^4.$$

\npgni The discriminant of this quartic in $\cos\eta$ is positive and since $3a^4\omega^2c^4>0$ we easily deduce that $u_{0}^2>0$. Moreover if $a>1$ we see that the dynamical system defined by $(\dot{\xi},\dot{\eta})$ above, converges to the ellipse $\textrm{cosh}\xi=a$. A computer simulation of such a process is shown below.

\begin{center}
\includegraphics[scale=0.65]{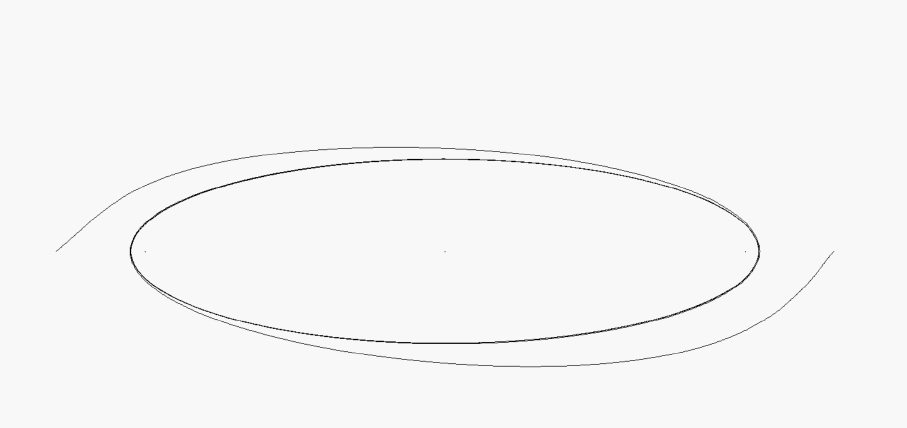}\\
Semi-Classical Process of $(\xi,\eta)$ Spiral Orbit\\
(The line joining the two foci is a cut singularity).
\end{center}

\npgni This picture reminds us of the shape of a barred galaxy such as NGC1300 as captured by the Hubble telescope shown in the picture below.

\begin{center}
\includegraphics[scale=0.5]{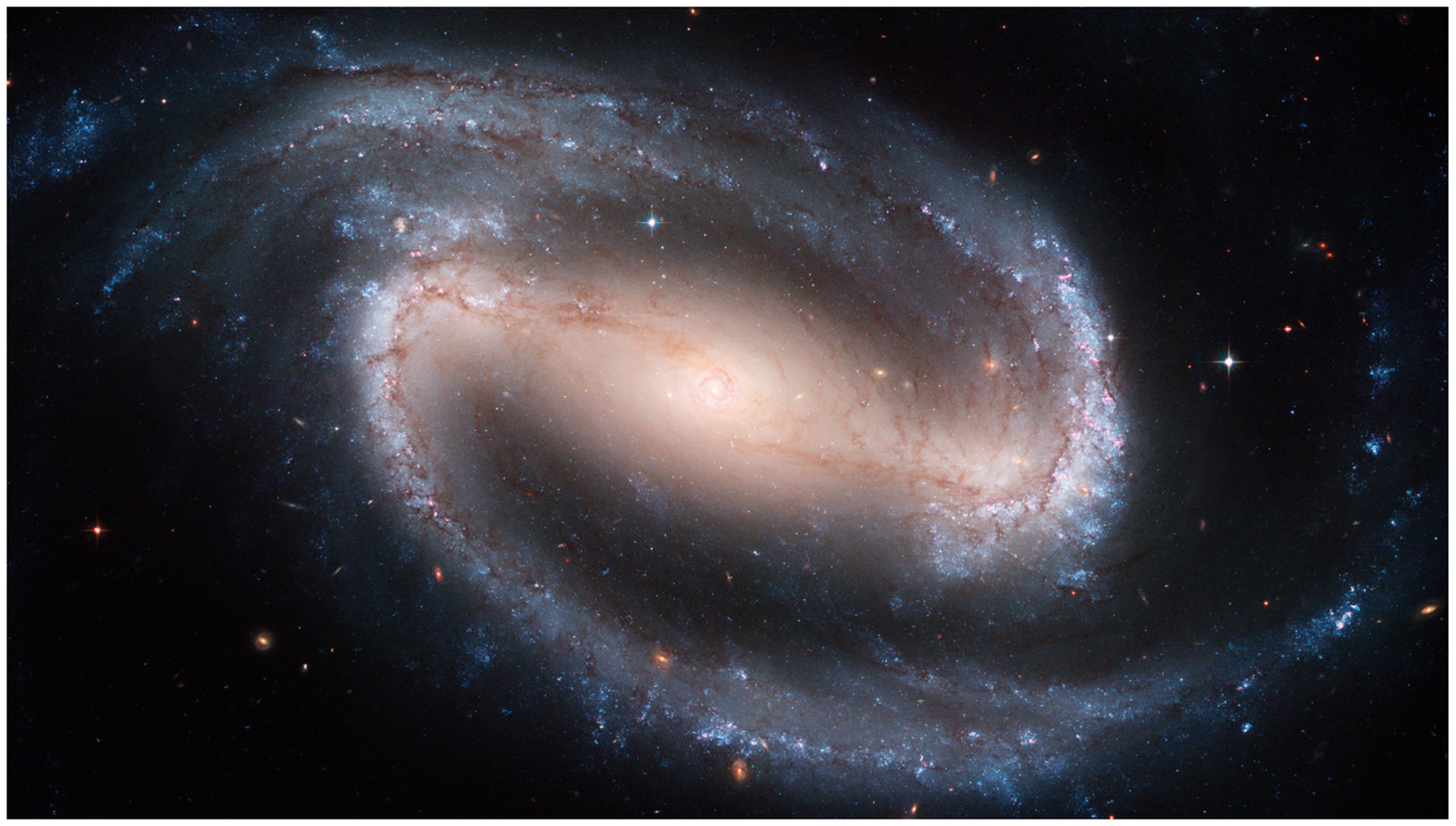}\\
Barred Galaxy NGC1300 Credit: Hubble.
\end{center}

\npgni We ask could this semi-classical system shed some light on the formation of galaxies of this type as well as galaxies such as the Whirlpool Galaxy (M51a) and the elliptical ring galaxy AM 0644-741? (See the images below).

\begin{center}
\includegraphics[scale=0.4]{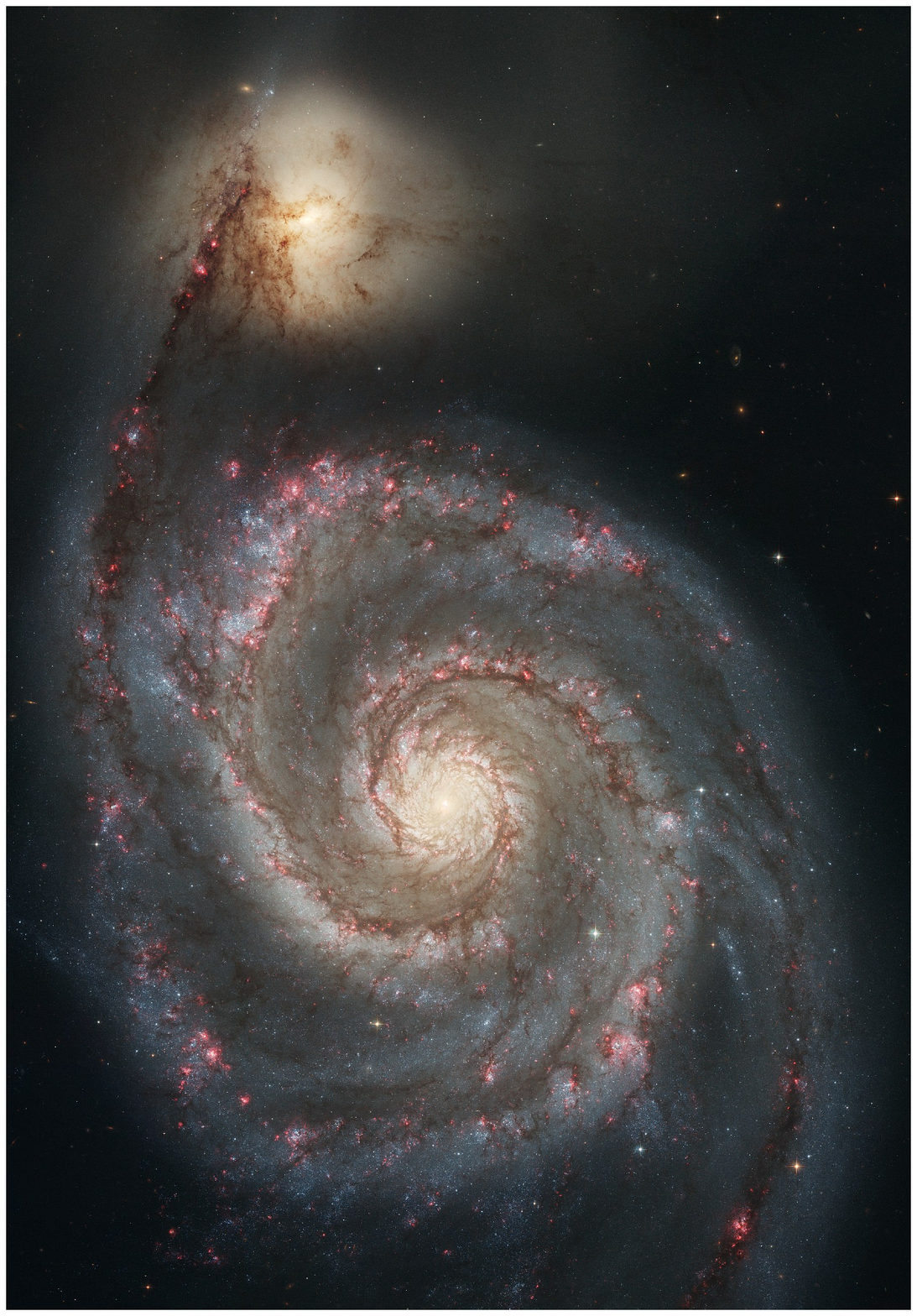}\\
Whirlpool Galaxy (M51a) and (M51b) Credit: ESA/Hubble
\end{center}

\begin{center}
\includegraphics[scale=0.65]{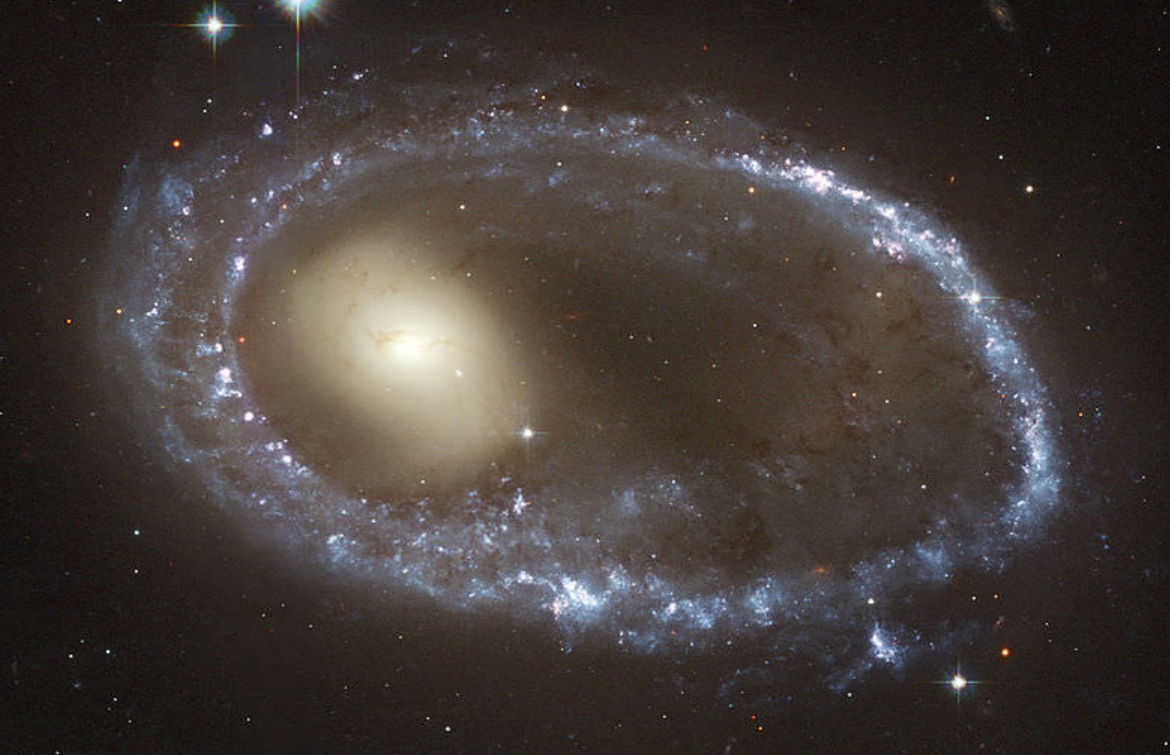}\\
Ring Galaxy AM 0644-741 Credit: Hubble.
\end{center}

\npgni In the last case if this is a two-centre ellipse the other focus could potentially be an area in space to look for dark matter?

\npgni \textbf{Exercise 1}

\npgni Prove that, if the classical Euler two-centre solution has the polar equation,\\ $\dfrac{l}{r}=1+e\cos{\theta}$, $r=r_{0}^{2C}(\theta)$, relative to centre 1, the corresponding entropy of the gaussian wavefunction of the cloud spiral is $\mathcal{E}_{2C}$, where

$$\mathcal{E}_{2C}=-\dfrac{1}{2\epsilon^2}\sqrt{\dfrac{\mu_{1}+\mu_{2}}{A^3}}\dfrac{(1+e\cos{\theta})^4}{\sqrt{1-e^2}(1+e^2+2e\cos{\theta})^2}\left(r-r_{0}^{2C}(\theta)\right)^2.$$

\npgni Compare this with the Keplerian case,

$$\mathcal{E}_{K}=-\dfrac{1}{2\epsilon^2}\sqrt{\dfrac{\mu}{a^3}}\dfrac{(1+ecos{\theta})^3}{(1-e^2)(1+3e^2+e(3+e^2)\cos{\theta})}\left(r-r_{0}^{K}(\theta)\right)^2.$$

\npgni These different entropies could control the behaviour of different parts of the cloud where stars etc. could form on different spirals as they converge to elliptical orbits for two-centre problems.

\npgni \textbf{Exercise 2}

\npgni For entropy $\mathcal{E}$ and gaussian wavefunction $\psi$,

$$|\psi|=N\exp\left(\dfrac{R}{\epsilon^2}\right)=N\exp{(\mathcal{E})},$$

\npgni $N$ the normalisation constant: $\int|\psi|^2d^2\vct{x}=1$, where in the limit as $\epsilon\rightarrow0$, taking into account contributions from inside and outside the cloud,

$$N^2\int\exp{\left(\dfrac{2R}{\epsilon^2}\right)}d^2\vct{x}=N^2\int\exp{(\mathcal{E})}d^2\vct{x}\sim1,$$

\npgni giving, for the Kepler ellipse, the integral,

$$\left(N_{K}\right)^{-2}=2\sqrt{2\pi}a\left(\dfrac{a^3}{\mu}\right)^{1/4}\epsilon\int\limits_0^{\pi}(1-e\cos{v})\sqrt{(1+e^2+2e\cos{v})}dv,$$

\npgni where $v$ is the eccentric anomaly. This can be evaluated in terms of elliptic integrals of the first and second kind as

$$\left(N_{K}\right)^{-2}=\dfrac{2\sqrt{2\pi}}{3}a\left(\dfrac{a^3}{\mu}\right)^{1/4}(1+e)\epsilon\bigg{\{}(1-e)^2F\left(\dfrac{\pi}{2},\dfrac{2\sqrt{e}}{1+e}\right)+(5-e^2)E\left(\dfrac{\pi}{2},\dfrac{2\sqrt{e}}{1+e}\right)\bigg{\}}.$$

\npgni Similarly for the two-centre ellipse

$$\left(N_{2C}\right)^{-2}=(2\pi)^{3/2}A\left(\dfrac{A^3}{\mu_{1}+\mu_{2}}\right)^{1/4}(1-e^2)^{3/4}\epsilon.$$

\npgni It is tempting to examine Hubble photographs to look for the Euler two-centre solutions and the corresponding elliptical spiral, taking into account the vastly different length and time scales of observer and the observed astronomical scene. Needless to say if the Schrödinger state of the cloud includes $\psi_{K}$ contributions, these states will dominate the dynamics in a neighbourhood of the corresponding orbit. In both cases there are cut singularities on the major axis of the elliptical orbit consistent with barred galaxies.

\section{The Effective Potential, Stability of Circular Orbits and Potential Wells}\vspace{-5mm}

\npgni We start by considering KLMN circular orbits for toy neutron stars and their stability, then turn to more general periodic equatorial orbitals and their equations including a Kepler equation, working our way towards the Feynman-Kac formula needed in this format. Simple physical insights have proven invaluable as you will see, but in our opinion the mathematical insights afforded us by the giants of the subject are key. In particular they reveal how important Brownian motion is in relation to the semi-classical orbits here and the ideas of some of the giants of our subject. The methods pioneered here we hope will help in more general astronomical settings in modelling the physical behaviours observed e.g. by Hubble.

\subsection{Stability of Circular Orbits}

\npgni Here the effective potential, $V_{\textrm{eff}}$, is defined by

$$\frac{\dot{r}^2}{2}+V_{\textrm{eff}}=E,$$

\npgni $E$ being the total energy, giving for $u=\dfrac{1}{r}$, $V_{\textrm{eff}}=E-\dfrac{f(u)}{2}$,

$$V'_{\textrm{eff}}(r)=-\frac{f'(u)}{2}\frac{du}{dr}=\frac{f'(\frac{1}{r})}{2r^2},\;\;\;\frac{f(u)}{2}=E+\mu u-\frac{u^2}{2}(C-Bu)^2.$$

\npgni (See Ref. [27]). In seeking stable circular orbits, with radius $a_{0}$ ($>0$), $V_{\textrm{eff}}$ needs to have a local minimum at $r=a_{0}$, so it is necessary first that in the simplest of cases

$$V'_{\textrm{eff}}(r)\Big{|}_{r=a_{0}}=0$$

\npgni and, if $V''(a_{0})\ne0$, that $V''_{\textrm{eff}}(a_{0})>0$ i.e.

$$\left(\frac{f''(\frac{1}{r})}{r^2}+\frac{2f'(\frac{1}{r})}{r}\right)\Bigg{|}_{r=a_{0}}<0.$$\vspace{-5mm}

\npgni Since $f'(u)=\mu-C^2u+3BCu^2-2B^2u^3$, this yields for $B\ne0$,

$$u^3-\frac{3C}{2B}u^2+\frac{C^2}{2B^2}u-\frac{\mu}{2B^2}=0.$$\vspace{-8mm}

\npgni Setting $u=\left(v+\dfrac{C}{2B}\right)$, gives

$$v^3-\frac{C^2}{4B^2}v-\frac{\mu}{2B^2}=0,\;\;\;\textrm{i.e.}\;\;v^3+pv+q=0,$$\vspace{-8mm}

\npgni where $p=-\dfrac{C^2}{4B^2}$ and $q=-\dfrac{\mu}{2B^2}$.\vspace{-3mm}

\npgni So\vspace{-5mm}

$$\left(q^2+\dfrac{4}{27}p^3\right)=\dfrac{1}{4B^4}\left(\mu^2-\dfrac{C^6}{108B^2}\right)$$

\npgni and our equation $V'_{\textrm{eff}}(r)=0$, will have one real root and a complex pair if

\npgni $\left(1-\dfrac{C^6}{108\mu^2B^2}\right)>0$, three real distinct roots if $\left(1-\dfrac{C^6}{108\mu^2B^2}\right)<0$ and repeated\vspace{-3mm}

\npgni roots if $\dfrac{C^3}{\mu B}=\pm\sqrt{108}=\pm6\sqrt{3}$. We call this the critical value of $\dfrac{C^3}{\mu B}$.

\npgni Working in terms of the dimensionless variables, $Z=-E^3B^2$, $W=E^2BD$, where $\mu=1$, it is easy to plot the curve in the $(Z,W)$ plane corresponding to $\bigtriangleup(V'_{\textrm{eff}})=0$, $\bigtriangleup$ being the discriminant. (see below).

$$0=\bigtriangleup(V'_{\textrm{eff}})\iff(W+Z)^3=\pm6\sqrt{3}Z^2.$$

\npgni (See energy bifurcation in Ref. [27]). As you will see an equally important curve here is $(W+Z)^3=\pm\dfrac{27}{2}Z^3$.

\npgni We now have enough information to graph $V_{\textrm{eff}}(r)$ in the distinct cases which can arise. Bearing in mind that

$$\dot{r}^2=2(E-V_{\textrm{eff}}(r))=f(u),$$

\npgni one can plot the curves $y_{1}=E$ and $y_{2}=V_{\textrm{eff}}(r)$ and see easily the nature of the roots of $f(u)=0$ depending upon when $\bigtriangleup_{4}(f)<0$ for which we have a pair of real roots and a complex conjugate pair and when $\bigtriangleup_{4}(f)>0$, when either all four roots of $f(u)=0$ are real or none is. This is the key to finding the polar equation of our orbits in terms of elliptic integrals of Legendre. Recall that

$$\bigtriangleup_{4}=(g_{2}^3-27g_{3}^2)$$

\npgni gives a horrendous expression impossible to analyse directly. However, using the above graphical approach involving the curves $y_{1}$ and $y_{2}$, we can establish the roots of our quartic $f(u)=0$ in terms of a parameter $\lambda$, a positive root of a cubic equation. We call this our $\lambda$-analysis.

\subsection{$\lambda$-Analysis}

\npgni For $B\ne0$ the quartic equation $f(u)=0$ reads

$$-u^4+\dfrac{2C}{B}u^3-\dfrac{C^2}{B^2}u^2+\dfrac{2\mu}{B^2}u+\dfrac{2E}{B^2}=0.$$

\npgni Reducing this to the depressed quartic and then completing the square in quartic and quadratic terms allows us to write down the four roots as

$$u=\dfrac{1}{2B}\Bigg\{C+B\sqrt{\lambda}\pm\sqrt{C^2-B^2\lambda+\dfrac{4\mu}{\sqrt{\lambda}}}\Bigg\}\;;\;u=\dfrac{1}{2B}\Bigg\{C-B\sqrt{\lambda}\pm\sqrt{C^2-B^2\lambda-\dfrac{4\mu}{\sqrt{\lambda}}}\Bigg\},$$

\npgni where $\lambda$ is a positive real solution of the cubic equation

$$\lambda^3-\dfrac{C^2}{B^2}\lambda^2+\dfrac{4}{B^3}(\mu C+2BE)\lambda-\dfrac{4\mu^2}{B^4}=0.$$

\npgni We note that for $\mu B^{-2}$ real, $\lambda>0$ always exists.

\npgni For the values of $E$ which correspond to the extrema of $V_{\textrm{eff}}(r)$ it is possible to calculate the value of $\lambda$ and hence determine the range(s) of $\lambda$ corresponding to all four roots of our original quartic as well as determine the nature of these roots.

\subsection{Graphical Analysis for small $|\vct{B}|$}

\npgni Firstly, in general, we note that $E<0$ guarantees that the real roots of, $f(u)=0$, have to be positive because $\mu>0$. Now considering the cubic equation, $V_{\textrm{eff}}'=0$; two cases arise if for now we ignore the possibility of equal roots. Case 1 when there are 3 real positive roots for $E<0$ for $\left(1-\frac{C^6}{108\mu^2B^2}\right)<0$ and case 2 when for $\left(1-\frac{C^6}{108\mu^2B^2}\right)>0$, $E<0$, we have one real root and a complex conjugate pair. For $V_{\textrm{eff}}$ in case 1 we have 2 local minima and one local maximum at $r_{0},r_{2},r_{1}$, respectively, $r=\frac{1}{u}$, and in case 2 the one real root corresponds to a local minimum at $r_{0}$. So in case 1 we have 2 potential wells and in case 2 we have just one potential well for appropriate energies $E<0$.

\npgni Also Vieta's formula gives:-

\npgni Case 1

$$r_{0}=\frac{\sqrt{3}B}{|C|}\left[\frac{\sqrt{3}}{2}\frac{C}{|C|}-\cos\left(\frac{1}{3}\cos^{-1}\left(\frac{6\sqrt{3}\mu B}{|C|^3}\right)\right)\right],\;\;\;(\textrm{min})$$

$$r_{1}=\frac{\sqrt{3}B}{|C|}\left[\frac{\sqrt{3}}{2}\frac{C}{|C|}-\cos\left(\frac{1}{3}\cos^{-1}\left(\frac{6\sqrt{3}\mu B}{|C|^3}\right)+\frac{\pi}{3}\right)\right],\;\;\;(\textrm{max})$$

$$r_{2}=\frac{\sqrt{3}B}{|C|}\left[\frac{\sqrt{3}}{2}\frac{C}{|C|}-\cos\left(\frac{1}{3}\cos^{-1}\left(\frac{6\sqrt{3}\mu B}{|C|^3}\right)-\frac{\pi}{3}\right)\right].\;\;\;(\textrm{min})$$

\npgni Case 2

$$r_{0}=\frac{\sqrt{3}B}{(\cos\alpha-\frac{\sqrt{3}}{2})|C|},\;\;\;\alpha=\frac{1}{3}\cos^{-1}\left(\frac{6\sqrt{3}\mu B}{|C|^3}\right),$$

\npgni $0<\alpha<\dfrac{\pi}{6}$, $\dfrac{\sqrt{3}}{2}<\cos\alpha<1$, $r_{0}$ a local minimum.

\npgni Case 2 is the simplest to analyse. L'Hopital gives as $B\rightarrow0$, $r_{0}\rightarrow\frac{C^2}{\mu}$, the radius of the circular orbit for the Kepler problem. Since there is only one potential well, this case must correspond to near circular orbits, $E\sim V_{\textrm{eff}}(r_{0})$ as $B\sim0$. We give the relevant equation below $r=r_{0}(\theta)$ in polar coordinates.

\npgni Always assuming $E<0$, case 1 is much more interesting. Set $V_{\textrm{eff}}^{+}(\textrm{min})$ the bigger of the 2 values of $V_{\textrm{eff}}(r_{0})$ and $V_{\textrm{eff}}(r_{2})$ and $V_{\textrm{eff}}^{-}(\textrm{min})$ the smaller of these 2 values. Then a simple calculation yields for $C^6=\left(\frac{27}{2}\right)^2\mu^2B^2$, $V_{\textrm{eff}}(r_{1})\sim0$, and for $C^6>\left(\frac{27}{2}\right)^2\mu^2B^2>108\mu^2B^2$, for  $V_{\textrm{eff}}^{-}(\textrm{min})<E<V_{\textrm{eff}}^{+}(\textrm{min})<0$, we have just one well accessible to energy $E$ i.e. 2 real roots of our quartic and a complex conjugate pair and for $V_{\textrm{eff}}^{+}(\textrm{min})<E<0$, 4 real positive roots of $E=V_{\textrm{eff}}(r)$, so in the first instance $\bigtriangleup(f)<0$ and in the second instance $\bigtriangleup(f)>0$, so avoiding the algebraic complications involved in calculating $\bigtriangleup(f)$.

\npgni It remains to indicate how to determine the above roots. We have given a full account of this in the $\lambda$-analysis. Here we consider the case of the escape orbit where we have an unstable circular orbit at $r=r_{1}$ for $E\sim0$.

\npgni From Taylor's theorem

$$Q_{B}(u)=(u-u_{1})^2q(u),$$

\npgni where $q$ is the quadratic

$$q(u)=\frac{1}{2}Q_{B}''(u_{1})+\frac{(u-u_{1})}{6}Q_{B}'''(u_{1})+\frac{(u-u_{1})^2}{24}Q_{B}''''(u_{1}).$$

\npgni So the remaining roots are real or complex depending on the sign of 

$$\left(\frac{(Q_{B}'''(u_{1}))^2}{36}-\frac{1}{6}Q_{B}''''(u_{1})Q_{B}''(u_{1})\right).$$

\npgni For small negative $E$

$$E=-(u-u_{1})^2q(u)\sim(u-u_{1})^2\frac{Q_{B}''(u_{1})}{2}$$

\npgni for small $|u-u_{1}|$, so

$$u=u_{1}\pm\sqrt{\frac{-2E}{Q_{B}''(u_{1})}}$$

\npgni as one expects from the Puiseux series. Knowing the sum and product of the roots from $Q_{B}(u)$ gives us the first order corrections for the remaining roots as Puiseux demands.

\subsection{On the $\mathbf{r=r_{0}(\theta)}$ equation for KLMN problem for small $\mathbf{|B|}$, $\mathbf{C^6>\frac{729}{4}\mu^2B^2}$ and $\mathbf{\Big{|}\frac{B}{C}\Big{|}\sim0}$}

\npgni We note that, when $C^6\sim\frac{729}{4}\mu^2B^2$, $V_{\textrm{eff}}(r_{1})\sim0$, in which case the energy $E\sim0$ corresponds to an unstable circular orbit with radius $r_{1}$. The condition $C^6>\frac{729}{4}\mu^2B^2$ can give one or two wells depending upon the energy $E$ which we assume here is negative. This is explained in our graphical analysis. The point is that the graphs, $y_{1}=V_{\textrm{eff}}(r)$ and $y_{2}=E$ have 4 points of intersection $u=a,b,c,d$, $u=\frac{1}{r}$, where we assume $a>b>c>d>0$ for $V_{\textrm{eff}}^+(\textrm{min})<E<0$ and only $a,b$ giving 2 points of intersection if $V_{\textrm{eff}}^-(\textrm{min})<E<V_{\textrm{eff}}^-(\textrm{min})<0$,

$$V_{\textrm{eff}}^+(\textrm{min})=\textrm{max}\left(V_{\textrm{eff}}(r_{0}),V_{\textrm{eff}}(r_{2})\right),\;\;\;V_{\textrm{eff}}^-(\textrm{min})=\textrm{min}\left(V_{\textrm{eff}}(r_{0}),V_{\textrm{eff}}(r_{2})\right).$$

\npgni We now calculate approximate solutions in the two sub-cases assuming $\Big{|}\dfrac{B}{C}\Big{|}\sim0$.

\npgni Recall that $u'=\dfrac{du}{d\theta}$ satisfies,

$$\frac{u'^2}{2}=\frac{(E+\mu u)}{(C-Bu)^2}-\frac{u^2}{2}=\frac{(E+\mu u)}{C^2}\left(1-\frac{Bu}{C}\right)^{-2}-\frac{u^2}{2},\;\;\;u(\theta_{0})=b\;\textrm{or}\;a.$$

\npgni Evidently the r.h.s. is an analytic function of $\dfrac{B}{C}$ for sufficiently small $\Big{|}\dfrac{B}{C}\Big{|}$, if $u$ is bounded as it will be in our cases. Hence we have a power series expansion of the r.h.s. which generates the asymptotic formal series:

$$Q_{B}(u)=\frac{2}{C^2}(E+a_{1}u+a_{2}u^2+a_{3}u^3+a_{4}u^4),$$

\npgni where

\npgni $a_{1}=\left(\dfrac{2BE}{C}+\mu\right)$, $a_{2}=\left(\dfrac{3B^2E}{C^2}+\dfrac{2B\mu}{C}-\dfrac{C^2}{2}\right)$, $a_{3}=\left(\dfrac{3B^2\mu}{C^2}+\dfrac{4B^3E}{C^3}\right)$, $a_{4}=\dfrac{4B^3\mu}{C^3}$

\npgni so that,

$$u'^2=Q_{B}(u)+O\left(\Big{|}\frac{B}{C}\Big{|}^4\right).$$

\npgni So the two sub-cases correspond to $\bigtriangleup(Q_{B})>0$ and $\bigtriangleup(Q_{B})<0$, $\bigtriangleup(Q_{B})$ the discriminant of the quartic above. Our approximate solution reads,

$$\left(\frac{du}{d\theta}\right)^2=Q_{B}(u).$$

\npgni Before tackling this equation we note that the product of the roots of $Q_{B}=0$ is given by $\dfrac{E}{4\mu}\left(\dfrac{C}{B}\right)^3$, so when $E<0$ for all 4 roots to be positive we must have $\dfrac{B}{C}<0$, an essential condition for this approximation to be physically relevant here.

\npgni Anyway, if $V_{\textrm{eff}}^+(\textrm{min})<E<0$, our approximate solution, $u$, satisfies

$$\sqrt{\Big{|}\frac{8B^3}{C^5}\Big{|}}(\theta-\theta_{0})=-\int\limits_{u}^a\frac{dt}{\sqrt{(a-t)(t-b)(t-c)(t-d)}},\;\;\;u(\theta_{0})=a,$$

\npgni for $u\in(b,a)$ as long as $C^6>\dfrac{729}{4}\mu^2B^2$, where

$$\int\limits_{u}^a\frac{dt}{\sqrt{(a-t)(t-b)(t-c)(t-d)}}=\frac{2F(\lambda,k)}{\sqrt{(a-c)(b-d)}},$$

\npgni $F$ an elliptic integral of the first kind, with

$$\lambda=\sin^{-1}\sqrt{\frac{(a-c)(u-b)}{(a-b)(a-c)}},\;\;\;k=\sqrt{\frac{(a-b)(c-d)}{(a-c)(b-d)}}.$$

\npgni And, when $V_{\textrm{eff}}^-(\textrm{min})<E<V_{\textrm{eff}}^-(\textrm{min})<0$, we have just 2 points of intersection of our graphs $y_{1}$ and $y_{2}$, so the roots of $Q_{B}(u)=0$ are a pair of real roots $a$ and $b$ and a complex conjugate pair, $c$ and $\bar{c}$, and assuming $a>b$ and $c=m+in$, $m,n\in\mathbb{R}$,

$$\sqrt{\Big{|}\frac{8B^3}{C^5}\Big{|}}(\theta-\theta_{0})=\int\limits_{b}^u\frac{dt}{\sqrt{(a-t)(t-b)((t-m)^2+n^2)}},\;\;\;u(\theta_{0})=b,$$

\npgni $u\in(b,a)$, again as long as $C^6>\dfrac{729}{4}\mu^2B^2$, and in this case

$$\int\limits_{b}^u\frac{dt}{\sqrt{(a-t)(t-b)((t-m)^2+n^2)}}=gF(\phi,h),$$

$$\phi=\sin^{-1}\left(\frac{(a-u)\tilde{B}-(u-b)A}{(a-u)\tilde{B}+(u-b)A}\right),\;\;\;h^2=\frac{(a-b)^2-(A-\tilde{B})^2}{4A\tilde{B}},$$

$$A^2=(a-m)^2+n^2,\;\;\;\tilde{B}^2=(b-m)^2+n^2,\;\;\;g=\frac{1}{\sqrt{A\tilde{B}}}.$$

\npgni This gives detailed information as to how our approximate solution behaves for different values of the energy $E$ but it is no substitute for the complete solution. The point is that as $|B|\sim0$ the equation,

$$\int\limits^u\frac{du}{\sqrt{Q_{B}(u)}}=\int\limits^\theta d\theta,$$

\npgni has to give the classical Keplerian ellipses with apses $r_{min}$ and $r_{max}$ and formally, $\sqrt{Q_{B}(u)}\sim\sqrt{Q_{0}(u)}\left(1+O\left(\Big{|}\dfrac{B}{C}\Big{|}^4\right)\right)$, but, taking the reciprocal, for $u=u_{0}=\dfrac{1}{r_{0}}$, $r_{0}\sim r_{min}$ or $r_{max}$, $\sqrt{Q_{0}(u)}=O(|B|)^{\frac{1}{2}}$ so our approximation can at best be accurate to $O\left(\Big{|}\dfrac{B}{C}\Big{|}^4|B|^{-\frac{1}{2}}\right)$.

\npgni For purposes of comparison we include the general classical solution here. The details and more powerful results with complete proofs can be found in Ref. [27].

\npgni Firstly one has to ascertain the value of $E<0$ to decide if there is one or two potential wells as explained above for this value of the energy. Once one has found the roots of the cubic $f'(u)=0$ and our corresponding effective potential $V_{\textrm{eff}}$, $u=\frac{1}{r}$, this is a simple task. Then, if the well in question is $u\in(u_{0},u_{1})$ for $u_{0}$, $u_{1}$ roots of $f(u)=0$, we define the well-time $z$ by

$$z=\int\limits_{u_{0}}^u\dfrac{du}{\sqrt{f(u)}},\;\;\;u\in(u_{0},u_{1}).$$

\npgni Then we have the simple result:

$$u(z)-u_{0}=\frac{f'(u_{0})}{4(\wp(z_{t};g_{2},g_{3})-\frac{1}{24}f''(u_{0}))},\;\;\;\;u=\dfrac{1}{r},$$

\npgni $\wp(z_{t};g_{2},g_{3})$ being the Weierstrass elliptic function with quartic invariants

$$g_{2}=a_{0}a_{4}-4a_{1}a_{3}+3a_{2}^2,\;\;\;g_{3}=a_{0}a_{2}a_{4}+2a_{1}a_{2}a_{3}-a_{2}^3-a_{0}a_{3}^2-a_{1}^2a_{4},$$

\npgni the discriminant $\bigtriangleup=g_{2}^3-27g_{3}^2$, $f(u)=a_{0}u^4+4a_{1}u^3+6a_{2}u^2+4a_{3}u+a_{4}$.

$$\dot{r}=\dfrac{dr}{dt}=\dfrac{f'(u_{0})\wp'(z;g_{2},g_{3})}{4\left(\wp(z;g_{2},g_{3})-\dfrac{1}{24}f''(u_{0})\right)^2}.$$

\npgni For the corresponding result for $\theta$ and more elaborate results, see Ref.[27].

\section{Semi-Classical Wavefunctions for Stationary States Corresponding to Periodic Orbits}

\subsection{Solving the $R_{0}$ equation}\vspace{-5mm}

\npgni We consider WIMP clouds with large angular momentum in a fixed direction so motion is confined to the $z=0$ plane and a small neighbourhood. In studying the formation of planetesimals in a neighbourhood of the classical periodic planar orbit, $C_{0}$, we reiterate the successful methods we employed in studying our astronomical states for fairly general potential energies $V$ and vector potentials $\vct{A}$, restricted to 2-dimensions.

\npgni Recall that, if $\psi\sim\textrm{exp}\left(\frac{R+iS}{\hbar}\right)$, where $\hbar=\epsilon^2\sim0$, then in the limit, from the Schrödinger equation, 

$$2^{-1}((\boldsymbol{\nabla}S-\vct{A})^2-|\boldsymbol{\nabla}R|^2)+V=E,\;\;\;(\boldsymbol{\nabla}S-\vct{A}).\boldsymbol{\nabla}R=0.$$

\npgni Assuming $R$ achieves its global maximum on $C_{0}$, with polar equation $r=r_{0}(\theta)$, at which $|\boldsymbol{\nabla}R|=0$, Taylor's theorem gives, in our neighbourhood,

$$R\cong\frac{R_{0}(\theta)}{2}(r-r_{0}(\theta))^2,\;\;\;R_{0}(\theta)=R_{r}''(r_{0}(\theta),\theta).$$

\npgni The problem is to find $R_{0}(\theta)$ given our semi-classical equations above.

\npgni Since we are working in 2-dimensions, in cartesians,

$$(\boldsymbol{\nabla}S-\vct{A})=\lambda\boldsymbol{\nabla}^{\perp}R,\;\;\;\;\boldsymbol{\nabla}R=\left(\frac{\partial R}{\partial x},\frac{\partial R}{\partial y}\right),$$\vspace{-8mm}

\npgni so\vspace{-5mm}

$$\boldsymbol{\nabla}S=\vct{A}+\lambda\left(-\frac{\partial R}{\partial y},\frac{\partial R}{\partial x}\right),$$\vspace{-5mm}

\npgni where to satisfy our energy equation

$$\lambda=\sqrt{1+\frac{2(E-V)}{|\boldsymbol{\nabla}R|^2}}\sim\sqrt{\frac{2(E-V)}{|\boldsymbol{\nabla}R|^2}},$$

\npgni since $|\boldsymbol{\nabla}R|^2\sim0$ in our neighbourhood. We next recall the Nelson problem in a disguised form. (See Nelson, Ref. [19]).

\npgni In order that the $S$ term be a gradient it is necessary (and modulo some regularity conditions) that\vspace{-5mm}

$$\textrm{curl}(\boldsymbol{\nabla}S)=\vct{0}$$\vspace{-10mm}

\npgni i.e.\vspace{-5mm}

$$\textrm{curl}_{2}(\vct{A})+\lambda\bigtriangleup_{2}R+\boldsymbol{\nabla}_{2}\lambda.\boldsymbol{\nabla}_{2}R=0,$$

\npgni in polar coordinates,

$$\frac{1}{r}\frac{\partial}{\partial r}(rA_{\theta})-\frac{1}{r}\frac{\partial A_{r}}{\partial \theta}+\textrm{div}_{2}(\lambda\boldsymbol{\nabla}_{2}R)=0,\;\;\;\;\;(\ast)$$

\npgni with

$$\boldsymbol{\nabla}_{2}R=\left(\frac{\partial R}{\partial r},\frac{1}{r}\frac{\partial R}{\partial \theta}\right)=(r-r_{0}(\theta))\left(R_{0},\frac{R_{0}'}{2r}(r-r_{0})-\frac{R_{0}r_{0}'}{r}\right).$$

\npgni Observe now that,

$$\textrm{div}_{2}(\widehat{\boldsymbol{\nabla}_{2}R})=K(\theta),$$

\npgni the curvature of the level curve, $R=c$, for $c$ a small negative constant, so that $K=\dfrac{1}{\rho_{c}}$, $\rho_{c}$ the radius of curvature of the curve $R=c$. Hence, we can rewrite ($\ast$) in the form

$$\sqrt{2(E-V)}\left(\frac{1}{r}\frac{\partial}{\partial r}(rA_{\theta})-\frac{1}{r}\frac{\partial A_{r}}{\partial \theta}\right)+\sqrt{2(E-V)}\boldsymbol{\nabla}_{2}R.\boldsymbol{\nabla}_{2}\sqrt{2(E-V)}+\frac{v^2}{\rho_{c}}=0.$$

\npgni This is just Newton's $2^{\textrm{nd}}$ Law of Motion resolved in the direction normal to the curve $R=c$, for small negative $c$.

\npgni Here we recall more generally,

$$\vct{v}=\dot{\vct{X}}_{t}=(\boldsymbol{\nabla}S-\vct{A}+\boldsymbol{\nabla}R)(\vct{X}_{t}),\;\;t\ge0,\;\;\vct{X}_{t=0}=\vct{x}\in\mathbb{R}^2.$$

\npgni So that for our WIMPish particles $P$, with $\overrightarrow{OP}=\vct{X}_{t}$ at time $t$,

$$\frac{d}{dt}R(\vct{X}_{t})=|\boldsymbol{\nabla}R|^2(\vct{X}_{t})\ge0,$$

\npgni giving the desired convergence to $C_{0}$ of our semi-classical spiral.

\npgni Now $\widehat{\boldsymbol{\nabla}_{2}R}=(1,f)(1+f^2)^{-\frac{1}{2}}$, with $f=\dfrac{R_{0}'}{2R_{0}}\left(1-\dfrac{r_{0}}{r}\right)-\dfrac{r_{0}'}{r}$.

\npgni We can now expand $R(r_{0}+(r-r_{0}),\theta)$ as a Taylor series, giving from the above identity:

\begin{lemma}

\npgni Let $K_{Q}$ be the curvature of our semi-classical spiral orbit converging to the level curve $R=c$, for small negative constant $|c|\sim0$, $K_{C_{0}}$ the curvature of the classical periodic orbit at $r=r(\theta)$, $K_{c}$ the curvature of our spiral at $r=r_{0}\pm\sqrt{\dfrac{2c}{R_{\textrm{max}}}}$. Then as $|c|\sim0$,

\npgni $(r_{0}^2+r_{0}'^2)^{\frac{5}{2}}\left(\frac{K_{Q}-K_{C_{0}}}{r-r_{0}}\right)\cong(2r_{0}^2-r_{0}'^2)r_{0}''-(4r_{0}'^2+r_{0}^2)r_{0}$

\npgni \hspace{45mm} $+\dfrac{r_{0}(r_{0}^2+r_{0}'^2)}{4}\left(2\dfrac{R_{0}''}{R_{0}}-3\left(\dfrac{R_{0}'}{R_{0}}\right)^2\right)+\dfrac{r_{0}'(2r_{0}'^2-3r_{0}r_{0}''-r_{0}^2)}{2}\dfrac{R_{0}'}{R_{0}}$,

\npgni giving the quantum correction to curvature to leading order.

\end{lemma}

\npgni Another way of probing the behaviour of, $R_{0}(\theta)=R_{r}''(r_{0}(\theta),\theta)$ is to consider the stationary state $\psi$ embodying the invariant density, $\rho$, for the curve $C_{0}$,

$$\psi\sim\sqrt{\rho}\;\textrm{exp}\left(\frac{R_{0}}{2\epsilon^2}(r-r_{0})^2+i\frac{S(r,\theta)}{\epsilon^2}\right)\;\;\;\textrm{as}\;\;\epsilon^2\sim0,$$

\npgni where $r=r_{0}(\theta)$ is the polar equation of the curve of a classical periodic orbit, $R$ and $S$ satisfying our semi-classical equations in 2-dimensions.

\npgni Firstly, $\delta r=\dfrac{\epsilon}{|R_{0}|^{\frac{1}{2}}}\sin(\psi-\theta)$, is the radial width of our 2-dimensional tube centred on $C_{0}$ where collisions are most likely to occur after converging to $C_{0}$. The normal width of the tube perpendicular to the width is $\delta r_{n}$, where

$$\delta r_{n}=\delta r\sin(\psi(\theta)-\theta),$$

\npgni $\psi(\theta)$ is the angle which the tangent to $C_{0}$ at $(r_{0}(\theta),\theta)$ makes with the $x$-axis. Since for sufficiently small $\epsilon$, $|\boldsymbol{\nabla}R|\sim0$, inside the tube, the velocity, $\vct{v}$, of P is such that $|\vct{v}|\sim\sqrt{2(E-V)}$. Conservation of mass gives, as we describe the curve $C_{0}$,

$$\dfrac{\rho\sqrt{2(E-V)}\sin(\psi(\theta)-\theta)}{|R_{0}|^{\frac{1}{2}}}=D,\;\;\;\textrm{a constant}.$$

\npgni Hence for the classical Hamilton-Jacobi function, $S$ and vector potential $\vct{A}$,

\npgni $\textrm{div}((\boldsymbol{\nabla}S-\vct{A})\rho)=0$, where for KLMN problems

$$S=\int\limits^r\sqrt{f(r)}dr+C\theta,\;\;\;\vct{A}=(A_{r},A_{\theta})=\left(0,\frac{B}{r^2}\right),$$

\npgni $C$ our constant of motion, $C=\left(r^2\dot{\theta}+\dfrac{B}{r}\right)$. This gives inside our tube

$$\rho\bigtriangleup S+\boldsymbol{\nabla}\rho.(\boldsymbol{\nabla}S-\vct{A})=0,$$

\npgni giving

\begin{lemma}

\npgni $$\rho\bigtriangleup S+\sqrt{2(E-V)}\frac{\partial\rho}{\partial s}=0,\;\;s\;\textrm{being arc length}.$$

\npgni So for $C_{0}$ approximately constant,

$$\frac{\rho(\theta)}{\rho(\theta_{0})}=C_{0}\exp\left(-\int\limits^{\theta}_{\theta_{0}}\frac{\bigtriangleup S}{\sqrt{2(E-V)}}\Bigg|_{r=r_{0}(\theta)}\sqrt{r_{0}^2+r_{0}'^2}\;d\theta\right)$$

\npgni and

$$|R_{0}(\theta)|^{\frac{1}{2}}\sim\frac{1}{D}\rho(\theta)\sqrt{2(E-V)}\Big|_{r=r_{0}(\theta)}\sin(\psi(\theta)-\theta),$$

\npgni where, for periodic density, $C_{0}$ has to be chosen e.g. so that $\dfrac{\rho(2\pi)}{\rho(0)}=1$.

\npgni Comparison with classical probability yields, in 1-dimension,

$$\exp(-\ln{\sqrt{2(E-V)}})\sim\exp\left(\frac{2R}{\epsilon^2}\right),\;\;\;\textrm{so}\;\;D\sim O(\epsilon).$$

\npgni In our case $D$ determines the fluid mass in the ring associated with the periodic orbit.

\end{lemma}

\begin{theorem}

\npgni For the 2-dimensional KLMN stationary state wavefunction\vspace{-5mm}

\npgni $\psi_{E}\sim\sqrt{\rho}\;\textrm{exp}\left(\frac{R+iS}{\hbar}\right)$ as $\hbar=\epsilon^2\sim0$, for a unit mass and unit charge particle with energy, $E<0$, in a sufficiently small neighbourhood of a periodic planar classical orbit on the curve $C_{0}$, $r=r_{0}(\theta)$, in polar coordinates $(r,\theta)$, it is necessary that in our gaussian model,

$$R\cong\frac{R_{0}(\theta)}{2}(r-r_{0}(\theta))^2,\;\;\;\;R_{0}(\theta)=R_{r}''(r_{0}(\theta),\theta)<0,$$

\npgni where $R_{0}(\theta)$ and $\rho(\theta)$ are given above and for $u=\dfrac{1}{r}$,

$$S(r,\theta)=-\int\limits^{\frac{1}{r}}\frac{\sqrt{f(u)}}{u^2}du+C\theta,$$

\npgni C being the constant $C=\left(r^2\dot{\theta}+\dfrac{B}{r}\right)$, $B$ measuring the strength of the magnetic dipole.

\end{theorem}

\npgni Obviously this result generalises to any Hamiltonian system with periodic planar orbits. Here we still have to check periodicity etc. for our KLMN set-up.

\begin{theorem}

\npgni When $\bigtriangleup>0$, the condition for the KLMN equatorial orbit to be periodic reduces to

$$\frac{p}{q}\pi=(C-Bu_{0})\omega-\frac{Bf_{0}'}{4\sqrt{4a_{0}-g_{2}a_{0}-g_{3}}}\left(2\zeta(a_{0})\omega+\ln\left(\frac{\sigma(\omega-a_{0})}{\sigma(\omega+a_{0})}\right)\right)\Bigg{|}_{\frac{f_{0}''}{24}=\wp(a_{0})},$$

\npgni $p,q$ coprime $\in \mathbb{Z}$, where $f_{0}=f(a_{0})$ etc., $\wp$ is the Weierstrass elliptic function $\wp(z)=\wp(z;g_{2},g_{3})$, $g_{2}$, $g_{3}$ the quartic invariants of $f$ and

$$\omega=\int\limits_{e_{1}}^{\infty}\frac{dt}{\sqrt{4t^3-g_{2}t-g_{3}}},$$

\npgni $\omega$ given by Vieta's formula, $e_{1}=\textrm{max}\left(\cos\theta_{1},\cos\left(\frac{2\pi}{3}\pm\theta_{1}\right)\right)$, with $\theta_{1}$ the largest $\theta\in(0,\pi)$ such that $\cos(3\theta)=\dfrac{3\sqrt{3}g_{3}}{\sqrt{g_{2}^3}}$.

\npgni When $\bigtriangleup<0$, the exact same formula holds save for the fact, $e_{2}$ is the real root of $4t^3-g_{2}t-g_{3}=0$,

$$\omega=\int\limits_{e_{2}}^{\infty}\frac{dt}{\sqrt{4t^3-g_{2}t-g_{3}}},$$

$$e_{2}=\left(\frac{g_{3}}{8}\right)^{\frac{1}{3}}\left(\left(1-\sqrt{1-\frac{g_{2}^3}{27g_{3}^2}}\right)^{\frac{1}{3}}+\left(1+\sqrt{1-\frac{g_{2}^3}{27g_{3}^2}}\right)^{\frac{1}{3}}\right).$$

\end{theorem}

\begin{proof}

\npgni $r^2\dot{\theta}=(C-Bu)$ so $\dot{\theta}=u^2(C-Bu)$ and for $z_{t}=\displaystyle\int\limits_{u_{0}}^{u_{t}}\dfrac{du}{\sqrt{f(u)}}$, $z_{t}=\displaystyle\int\limits_{0}^{t}\dfrac{ds}{r^2(s)}$ is the potential well time for the orbit starting at $u_{0}$,

$$u_{t}=u_{0}+\frac{f_{0}'}{4(\wp(z_{t};g_{2},g_{3})-\frac{1}{24}f_{0}'')},$$

\npgni where $f_{0}'=f'(u_{0})$, $f_{0}''=f''(u_{0})$ and $\wp$ is the Weierstrass elliptic function $\wp(z;g_{2},g_{3})$.

\npgni Now

$$\frac{d\theta}{dz}=\frac{\dot{\theta}}{\dot{z}}=(C-Bu),$$

\npgni giving for our half cycle,

$$\bigtriangleup\theta=\int\limits_{0}^{\omega}\left(C-B\left(u_{0}+\frac{f_{0}'}{4(\wp(z)-\frac{1}{24}f_{0}'')}\right)\right)dz.$$

\npgni Setting $a_{0}=\displaystyle\int\limits_{\frac{f_{0}''}{24}}^{\infty}\dfrac{dt}{\sqrt{4t^3-g_{2}t-g_{3}}}$, the result follows from Corollary 1 or (Byrd and Friedman, Ref. [5], 1037.10), $\zeta'=-\wp$ and $\dfrac{\sigma'}{\sigma}=\zeta$.

\end{proof}

\npgni If the orbit is confined to the well, $u\in(b,a)$, our $\theta$ cycle will include a $\dot{\theta}$ sign reversal if $\dfrac{C}{B}\in(b,a)$. When this is so and $\bigtriangleup\theta$ is small compared to $\pi$, we get very loopy motions as shown in the diagrams of Ref. [27]. When $a$ or $b$ equals $\dfrac{C}{B}$, the motion is cusped and sinusoidal if $\dfrac{C}{B}\notin(b,a)$. When $|B|\sim0$, orbits will look like Newton's revolving orbits.

\npgni Needless to say the last theorem embodies an implicit equation for our KLMN orbit, measuring $\theta$ from the apse, $r_{0}=\dfrac{1}{u_{0}}$, to $u=\dfrac{1}{r}$;

\begin{theorem}

The $r=r_{0}(\theta)$ implicit equation is

$$\frac{1}{r}-\frac{1}{r_{0}}=\frac{f'(u_{0})}{4(\wp(z)-\frac{1}{24}f''(u_{0}))},$$

\npgni $z=z(\theta)$ given by

$$\theta=(C-Bu_{0})z-\frac{Bf'(u_{0})}{4\sqrt{4a_{0}^3-g_{2}a_{0}-g_{3}}}\left(2\zeta(z_{0})z+\ln\left(\frac{\sigma(z-z_{0})}{\sigma(z+z_{0})}\right)\right),$$

\npgni where $z_{0}=\displaystyle\int\limits_{\frac{f_{0}''}{24}}^{\infty}\dfrac{dt}{\sqrt{4t^3-g_{2}t-g_{3}}}$, with

$$g_{2}=ae-4bd+3c^2,\;\;\;g_{3}=ace+2bcd-c^3-ad^2-b^2e,$$

\npgni for $f(u)=au^4+4bu^3+6cu^2+4du+e=2\left(E+\mu u-\dfrac{u^2}{2}(C-Bu)^2\right)$.

\end{theorem}

\npgni This is not very transparent and we have given a different version in terms of elliptic integrals of Ref. [27]. Nevertheless, for the sake of completeness we include an implicit equation for, $z=z(t)$, $t$ being the physical time. This depends on the identity,

$$dt=\frac{dr}{\dot{r}}=-\frac{du}{u^2\dot{r}}=-\frac{du}{u^2\sqrt{f(u)}}=+\frac{dz}{u^2},\;\;\;dz>0,\;\;dt>0,$$

\npgni where

$$(u-u_{0})=\frac{f'(u_{0})}{4(\wp(z)-\frac{1}{24}f''(u_{0}))},\;\;\;u=\frac{1}{r},\;\;\dot{r}=\sqrt{f(u)},$$

\npgni $\wp(z)=\wp(z;g_{2},g_{3})$ and $f(u)=2\left(E+\mu u-\dfrac{u^2}{2}(C-Bu)^2\right)$.

\npgni Needless to say this result is important in guaranteeing the convergence in the long time limit of semi-classical orbits to classical periodic ones generalising our results for Keplerian ellipses and astronomical elliptic states. We reiterate the width of the collision\vspace{-3mm}

\npgni zone here is $\dfrac{\epsilon}{|R_{0}|^{\frac{1}{2}}}\sin(\psi-\theta)$. More generally $R_{0}$ has to satisfy the full non-linear equation,

$$\textrm{curl}_{2}(\vct{A})+\textrm{div}_{2}(\lambda\boldsymbol{\nabla}_{2}R)=0,$$

\npgni where $\lambda=\dfrac{\sqrt{2(E-V)}}{|\boldsymbol{\nabla}R|}+\dfrac{|\boldsymbol{\nabla}R|}{2\sqrt{2(E-V)}}$, $R=\dfrac{1}{2}R_{0}(\theta)(r-r_{0}(\theta))^2$, in a neighbourhood of

\npgni $C_{0}$, where $|\boldsymbol{\nabla}R|\ne0$ is small.

\npgni In generalising this result to 3-dimensions the reader must remember $\boldsymbol{\nabla}R$ must be in the osculating plane of the classical orbit $C_{0}$ and must be parallel to its normal, $\vct{n}$, to satisfy the semi-classical equations. Needless to say fixing $|\boldsymbol{\nabla}R|^2$ in our neighbourhood of $C_{0}$ in a sense is a weak violation of Heisenberg's uncertainty principle. Our astronomical states minimise these uncertainties. There is also a problem with classical constants of the motion which do not carry over to the quantum situation. Is it by chance that the result above relates so simply to near circular orbits for which $\left(r^2\dot{\theta}+\dfrac{B}{r}\right)$ is constant.

\npgni We now ask for KLMN problems "What is the corresponding momentum space wavefunction, $\tilde{\psi}_{E}(p)$, for the distribution of the momentum of WIMPish cloud particles in a neighbourhood of $C_{0}$?". After all, this will affect directly the collision processes responsible for forming planetesimals and larger heavenly bodies.

$$\tilde{\psi}_{E}(p)=N\int\limits_{0}^{2\pi}d\theta\int\limits_{0}^{\infty}rdr\;\textrm{exp}\left(\frac{1}{\epsilon^2}\left(-i\vct{p}.\vct{r}+i(S_{0}(r)+C\theta)-\frac{|R_{0}|}{2}(r-r_{0})^2\right)\right),$$

\npgni $N$ a normalising factor, $S_{0}(r)=\int\limits^r\sqrt{f(r^{-1})}dr$, $\vct{p}=(\dot{\vct{q}}+\vct{A})$, $\vct{A}=(A_{r},A_{\theta})$, $A_{r}=0$, $A_{\theta}=\dfrac{B}{r^2}$, in polar coordinates.

\npgni Let $\alpha$ be the angle between $\vct{p}$ and $\vct{r}$, $\vct{r}=\overrightarrow{OP}$, $P$ in a neighbourhood of $C_{0}$. Actually on $C_{0}$, classically, $\alpha=\theta-\psi_{0}(\theta)$, $\psi_{0}(\theta)$ being the angle between the $x-\textrm{axis}$ and the tangent to $C_{0}$ at $(r_{0}(\theta),\theta)$. For the $r$ integral we use

$$\int\limits_{-\infty}^{\infty}e^{-iyx}e^{-\frac{1}{2}x^2}dx=(2\pi)^{\frac{1}{2}}e^{-\frac{1}{2}y^2}.$$

\npgni Inside the exponential, since $\vct{p}.\vct{r}=pr\cos{\alpha}$, the interesting term is

$$-ipr_{0}\cos{\alpha}-ip(r-r_{0})\cos{\alpha}+iS_{0}(r_{0})+i(r-r_{0})\frac{\partial S_{0}}{\partial r_{0}}\Bigg|_{r_{0}(\theta)}+iC\theta-\frac{|R_{0}|}{2}(r-r_{0})^2.$$

\npgni We do the $(r-r_{0})$ integral first, giving

$$\tilde{\psi}_{E}(p)\sim N\int\limits_{0}^{2\pi}r_{0}(\theta)\textrm{exp}\left(\dfrac{-ipr_{0}\cos{\alpha}+iS_{0}(r_{0})}{\epsilon^2}\right)\textrm{exp}\left(\dfrac{iC\theta}{\epsilon^2}\right)\textrm{exp}\left(-\dfrac{\left(p\cos{\alpha}-\dfrac{\partial S_{0}}{\partial r_{0}}\right)^2}{2\epsilon^2|R_{0}|}\right)d\theta,\;\;\;$$

\npgni $p=|\vct{p}|$. So the main contribution to the integral comes from,

$$p\cos{\alpha}\cong\frac{\partial S_{0}}{\partial r}\Bigg|_{r=r_{0}(\theta)}\;\;\;\;\;\;\;\;\;\;\textrm{(a)}.$$

\npgni The main contribution to the $\theta$ integral comes from $\theta$:

$$\frac{\partial}{\partial\theta}(-pr_{0}\cos{\alpha}+S_{0}(r_{0})+C\theta)=0,\;\;\;\;\;\;\;\;\textrm{(b)}$$

\npgni as principle of stationary phase shows, see Fedoriuk, Ref. [12]; i.e.

$$\alpha'pr_{0}\sin{\alpha}-pr_{0}'\cos{\alpha}+\frac{\partial S_{0}}{\partial r_{0}}r_{0}'+C=0,\;\;\;\;\alpha'=\frac{\partial \alpha}{\partial\theta}.$$

\npgni From (a) the second and third terms cancel giving

$$p\sin{\alpha}=-\frac{C}{r_{0}\alpha'}\;\;\;\;\;\;\;\;\textrm{(c)},$$

\npgni $C$ being the constant of the motion, $C=\left(r^2\dot{\theta}+\dfrac{B}{r}\right)$.

\npgni This gives

$$p^2\cong\left(\left(\frac{\partial S_{0}}{\partial r_{0}}\right)^2+\frac{C^2}{r_{0}^2\alpha'^2}\right),\;\;\;\;\alpha'=\frac{d\alpha}{d\theta}(\theta).$$

\npgni Here it is tempting to write, $\alpha'=\dfrac{d}{d\theta}(\theta-\psi_{0}(\theta))$,

\npgni i.e.

$$\alpha'\cong1-\frac{d\psi_{0}}{d\theta}=\left(1-\frac{d\psi_{0}}{ds}\frac{ds}{d\theta}\right)$$

\npgni i.e.

$$\alpha'\cong1-\frac{\sqrt{r_{0}^2+r_{0}'^2}}{\rho_{c}},$$

\npgni where $\rho_{c}$ is the radius of curvature of $C_{0}$ at $(r_{0}(\theta_{0}),\theta_{0})$. So we see at least formally that $p=|\vct{p}|$, $\vct{p}=(\dot{\vct{q}}+\vct{A})$, $C=r^2\dot{\theta}+\dfrac{B}{r}$, a constant,

$$p^2\cong\left(\left(\frac{\partial S_{0}}{\partial r_{0}}\right)^2+\frac{C^2\rho_{c}^2}{\left(\rho_{c}-\sqrt{r_{0}^2+r_{0}'^2})\right)^2r_{0}^2}\right)\Bigg|_{\theta=\theta_{0}}$$

\npgni and

$$\tan{\alpha}\cong-\frac{C\rho_{c}}{\left(\rho_{c}-\sqrt{r_{0}^2+r_{0}'^2}\right)r_{0}\left(\frac{\partial S_{0}}{\partial r_{0}}\right)}\Bigg|_{\theta=\theta_{0}},$$

\npgni $\theta_{0}$ satisfying (b), where we assume $\dfrac{r_{0}'}{r_{0}}$ is small. We have also seen that the uncertainty in $p$ is proportional to $\epsilon|R_{0}|^{\frac{1}{2}}\sec{\alpha}$, for the above $\alpha$. So there is a residual Heisenberg Uncertainty principle at work.

\npgni It should be possible to prove an equivalent theorem in 3-dimensions using the methods of Ref. [12]. We leave this as an exercise, assuming small $z$. The relevant result from Ref. [12] is that the leading term, which is all we are concerned with here, in the asymptotic expansion of $I(\lambda)$, $I(\lambda)=\int\limits_{a}^{b}\textrm{exp}(i\lambda g(x))f(x)dx$ as $\lambda\rightarrow\infty$, is

$$I(\lambda)\sim\sqrt{\frac{2\pi}{\lambda|g''(x_{0})|}}f(x_{0})\textrm{exp}\left(i\lambda g(x_{0})+\frac{i\pi}{4}\textrm{sign}\;g''(x_{0})\right),$$

\npgni where $x_{0}$ is assumed to be the unique stationary point $x_{0}$, with $g'(x_{0})=0$, $x_{0}\in(a,b)$; $I(\lambda)=\textrm{O}(\lambda^{-\infty})$ as $\lambda\rightarrow\infty$, when there is no such $x_{0}\in(a,b)$. Here $f\in C_{0}^{\infty}(a,b)$, $g\in C^{\infty}(a,b)$, $g$ being real-valued. (See Ref. [12]).

\npgni In any case in 2-dimensions we have proved the following theorem.

\begin{theorem}

\npgni The first term in the asymptotic expansion in powers of, $\lambda=\dfrac{1}{\epsilon^2}\sim\infty$, of $\tilde{\psi}_{E}(p)$, the momentum space wavefunction in semi-classical mechanics, yields for $\vct{p}=(\dot{\vct{q}}+\vct{A}(\vct{q}))$ the distribution,

$$\mathbb{P}(p\in p\hspace{0.25mm}dp\hspace{0.25mm}d\alpha)\sim N\textrm{exp}\left(-\frac{\left(p-\sqrt{f(u_{0}(\theta))}\sec{\alpha(\theta)}\right)^2}{2\epsilon^2\sec^2\alpha(\theta)|R_{0}(\theta)|}\right)\delta(\alpha-\alpha(\theta_{0}))p\hspace{0.25mm}dp\hspace{0.25mm}d\alpha,$$

\npgni where $p=|\vct{p}+\vct{A}(\vct{q})|$, $\alpha$ being the angle between $\vct{p}$ and $\vct{r}\sim\vct{r}_{0}(\theta)$, $\alpha(\theta)=\theta-\psi(\theta)$, where $\theta_{0}$ is assumed to be unique satisfying (b), $N$ a normalisation constant. So for $C$ our constant of motion, $C=r^2\dot{\theta}+\dfrac{B}{r}$,

$$\tan{\alpha}\cong-\frac{C\rho_{c}}{\left(\rho_{c}-r_{0}(\theta)\right)r_{0}(\theta)\frac{\partial S_{0}}{\partial r_{0}}(r_{0}(\theta)}\Bigg|_{\theta=\theta_{0}},$$

$$p^2\cong\left(\left(\frac{\partial S_{0}}{\partial r_{0}}(r_{0}(\theta))\right)^2+\frac{C^2\rho_{c}^2}{\left(\rho_{c}-r_{0}(\theta)\right)^2r_{0}^2(\theta)}\right)\Bigg|_{\theta=\theta_{0}},$$

\npgni where $\rho_{c}$ is the radius of curvature of $C_{0}$ at $(r_{0}(\theta),\theta)$, $\rho_{c}\sim\dfrac{r^2}{(r-r_{0}''(\theta))}$ for $\dfrac{r_{0}'(\theta)}{r_{0}(\theta)}$ small and $r\sim r_{0}(\theta)$, $\theta\in(0,2\pi)$.

\end{theorem}

\npgni Once more this result is easy to generalise to any Hamiltonian system with periodic planar orbits.

\npgni Further quantum corrections will follow from

$$\lambda\sim\frac{\sqrt{2(E-V)}}{|\boldsymbol{\nabla}R|}\left(1+\frac{|\boldsymbol{\nabla}R|^2}{4(E-V)}\right).$$

\npgni Later we give our Kepler equation in this setting in terms of the Weierstrass elliptic function $\wp$. See Theorem 7.3 in the Appendix.

\section {Quantum Mechanics, Elementary Formula and Burgers-Zeldovich Fluids}\vspace{-5mm}

\npgni The key to the two results herein is the existence of a certain diffeomorphism, $D_{t}:\mathbb{R}^n\rightarrow\mathbb{R}^n$, $t\in(0,T)$, $D_{0}=id$, the identity, arising from the classical Newtonian mechanics in this setup and its relation to Schrödinger heat equations, for small times, $T$. The proof of the necessary and sufficiency of conditions for the existence of $T>0$ can be found in Elworthy and Truman, Ref. [11].

\subsection{Stochastic Mechanics, Nelson Newton Law and Burgers Equation}

\npgni For simplicity here we consider a putative particle of unit mass $P$, with $\overrightarrow{OP}=\vct{X}\in\mathbb{R}^n$, where the corresponding Hamiltonian,

$$H=\dfrac{\vct{p}^2}{2}+V(\vct{q}).$$

\npgni Classically from Hamilton's equations $\vct{X}(s)=\vct{x}(\vct{x}_{0},\vct{p}_{0},s)$, s being the time, $P$ being subject to the force field, $-\boldsymbol{\nabla}V(\vct{X})$, $V$ being the potential energy in the simplest case, so

$$\dfrac{d^2}{ds^2}\vct{X}(s)=-\boldsymbol{\nabla}V(\vct{X}(s)),\;\;\;s\in(0,T),$$

\npgni where $\vct{X}(0)=\vct{x}_{0}$ and $\dot{\vct{X}}(0)=\vct{p}_{0}$; $\vct{x}_{0}$, $\vct{p}_{0}$ the initial position and momentum, respectively.

\npgni We assume that for some smooth function $S_{0}:\mathbb{R}^n\rightarrow\mathbb{R}$, $\vct{p}_{0}=\boldsymbol{\nabla}S_{0}(\vct{x}_{0})$, then, as is well known, the solution of the corresponding Hamilton-Jacobi equation,

$$\dfrac{\partial S}{\partial t}+2^{-1}|\boldsymbol{\nabla}S|^2+V=0,\;\;\;\;\;\;\;\;\textrm{(H-J)}$$

\npgni $S=S(\vct{x},t)$ with $S(\vct{x},0)=S_{0}(\vct{x})$ is given by

$$S(\vct{x},t)=S_{0}(\vct{x}_{0}(\vct{x},t))+\int\limits_{0}^{t}(2^{-1}\dot{\vct{X}}_{0}^2-V(\vct{X}_{0}(s)))ds,$$

\npgni where for $s\in(0,t)$, $t\in(0,T)$, $\vct{X}_{0}(s)=\vct{x}(\vct{x}_{0},\boldsymbol{\nabla}S_{0}(\vct{x}_{0}),s)$, $\vct{x}_{0}(\vct{x},t)$ such that $\vct{x}(\vct{x}_{0},\boldsymbol{\nabla}S_{0}(\vct{x}_{0}),t)=\vct{x}$ for each $\vct{x}\in\mathbb{R}^n$, $\vct{x}_{0}(\vct{x},t)$ being assumed to be unique. The global inverse function theorem guarantees that if $S_{0}$ and $V$ are smooth and satisfy appropriate boundedness conditions together with their derivatives, the map $\vct{x}=D_{t}\vct{x}_{0}$ is such that $D_{t}$, for $t\in(0,T)$ is, for sufficiently small $T$, a diffeomorphism with $D_{0}=id$. (See e.g. Ref. [11]).

\npgni It is easy to prove that for $s\in(0,t)$, for the above $S(\vct{x},t)$,

$$\dot{X}_{0}(t)=\boldsymbol{\nabla}S(\vct{X}_{0}(t),t)$$

\npgni i.e.

$$d\vct{X}_{0}(t)=\boldsymbol{\nabla}S(\vct{X}_{0}(t),t)dt=\vct{b}(\vct{X}_{0}(t),t)dt,$$

\npgni a dynamical system with $\vct{b}=\boldsymbol{\nabla}S$. This gives a Burgers type fluid with velocity field, $\vct{b}=\vct{v}$,

$$\dfrac{\partial\vct{v}}{\partial t}+(\vct{v}.\boldsymbol{\nabla})\vct{v}=-\boldsymbol{\nabla}V,$$

\npgni by taking the gradient of equation (H-J).

\npgni We now give a brief account of Nelson's stochastic mechanics and the Nelson/Newton Law, Nelson's version of Newton's $2^{\textrm{nd}}$ Law of Motion hidden in the Schrödinger equation for our Hamiltonian. This is no substitute for Nelson's and Fenyes' original work. (See Refs. [13], [14] and[15]).

\npgni For the above set-up the quantum Hamiltonian, $H=2^{-1}\vct{p}^2+V(\vct{q})$, reads

$$H=-\dfrac{\hbar^2}{2}\bigtriangleup_{\vct{x}}+V(\vct{x})$$

\npgni for a unit mass particle $P$, subject to the force field $-\boldsymbol{\nabla}V(\vct{x})$, $\vct{x}\in\mathbb{R}^d$, and the Schrödinger equation can be written as

$$i\hbar\psi^{\ast}\dfrac{\partial\psi}{\partial t}=-\dfrac{\hbar^2}{2}\psi^{\ast}\bigtriangleup\psi+V|\psi|^2,\;\;\;\;\;\;\;\;\;\;(\textrm{SE})$$

\npgni for quantum state $\psi$. We study the real and imaginary parts of this identity.

\npgni Equating imaginary parts of this equation gives for $\rho=|\psi|^2$, the quantum particle density,

$$\dfrac{\partial\rho}{\partial t}+\textrm{div}\vct{j}=0,$$

\npgni $\vct{j}$ being the probability current,

$$\vct{j}=\dfrac{\hbar}{2i}(\psi^{\ast}\boldsymbol{\nabla}\psi-\psi\boldsymbol{\nabla}\psi^{\ast}).$$

\npgni Writing $\psi=\textrm{exp}(R_{\epsilon}+iS_{\epsilon})$, $\rho=\textrm{exp}(2R_{\epsilon})$, and, if $(R_{\epsilon}+iS_{\epsilon})=\dfrac{R+iS}{\epsilon^2}$, $\epsilon^2=\hbar$, we obtain the continuity equation,

$$\dfrac{\partial\rho}{\partial t}=\boldsymbol{\nabla}.\left(\dfrac{\epsilon^2}{2}\boldsymbol{\nabla}\rho-\vct{b}\rho\right),$$

\npgni where $\vct{b}=\boldsymbol{\nabla}(R+S)$.

\npgni Nelson regards this equation as the forward Kolmogorov equation for the diffusion process $\vct{X}$, for $\epsilon^2=\hbar$,

$$d\vct{X}(t)=\vct{b}(\vct{X}(t),t)dt+\epsilon d\vct{B}(t),\;\;\;\;\;t>0,$$

\npgni $\vct{B}(.)$ a $\textrm{BM}(\mathbb{R}^d)$ process with $\mathbb{E}(B_{i}(s)B_{j}(t))=\delta_{ij}\textrm{min}(s,t)$, for $s,t>0$,

\npgni $\vct{B}=(B_{1},B_{2},\dotsc,B_{d})$, in cartesians in $\mathbb{R}^d$. So for probability $\mathbb{P}$,

$$\mathbb{P}(\vct{X}(t)\in A)=\int_{y\in A}\rho(\vct{y},t)dy,$$

\npgni for measurable sets $A$, with

$$\dfrac{\partial\rho}{\partial t}=\mathscr{G}^{\ast}\rho=\boldsymbol{\nabla}.\left(\dfrac{\epsilon^2}{2}\boldsymbol{\nabla}\rho-\vct{b}\rho\right),$$

\npgni $\mathscr{G}^{\ast}$ the formal $L^2$-adjoint of the generator $\mathscr{G}=\left(\dfrac{\epsilon^2}{2}\bigtriangleup+\vct{b}.\boldsymbol{\nabla}\right)$.

\npgni Of course the sample paths $\vct{X}(t)$ themselves are not differentiable in $t$, but conditional expectations save the day, defining

$$D_{\pm}f(\vct{X}(t),t)=\mathbb{E}\Bigg\{\dfrac{f(\vct{X}(t\pm\delta),t\pm\delta)-f(\vct{X}(t),t)}{\pm\delta}\Big|\vct{X}(t)\Bigg\}$$

\npgni and Nelson's stochastic analogue of acceleration,

$$\vct{a}(\vct{X}(t),t)=\dfrac{1}{2}(D_{+}D_{-}+D_{-}D_{+})\vct{X}(t),$$

$$D_{\pm}\vct{X}(t)=\vct{b}_{\pm}(\vct{X}(t),t),\;\;\;\;\;D_{+}\vct{X}(t)=\vct{b}(\vct{X}(t),t)=\boldsymbol{\nabla}(R+S)(\vct{X}(t),t),$$

\npgni $\vct{b}(\vct{X}(t),t)$ the drift in our diffusion, and

$$D_{-}\vct{X}(t)=\vct{b}(\vct{X}(t),t)-\epsilon^2\boldsymbol{\nabla}\ln{\rho}(\vct{X}(t),t).$$

\npgni He defines osmotic velocity $\vct{u}$ and current velocity $\vct{v}$ by

$$\vct{u}=\dfrac{1}{2}(\vct{b}_{+}-\vct{b}_{-}),\;\;\;\vct{v}=\dfrac{1}{2}(\vct{b}_{+}+\vct{b}_{-}).$$

\npgni Elementary vector algebra then yields (what at first looks like the somewhat impenetrable result)

$$\vct{a}(\vct{X}(t),t)=-\boldsymbol{\nabla}\left(\dfrac{\partial S}{\partial t}+\dfrac{|\boldsymbol{\nabla}R|^2-|\boldsymbol{\nabla}S|^2}{2}+\dfrac{\epsilon^2}{2}\bigtriangleup R\right)(\vct{X}(t),t).$$

\npgni But from equating the real parts of (SE) above one obtains,

$$\vct{a}(\vct{X}(t),t)=-\boldsymbol{\nabla}V(\vct{X}(t))$$

\npgni i.e. for our unit mass particle the Nelson/Newton Law,

$$\textrm{Force}=\textrm{Mass}\times\textrm{Acceleration},$$

\npgni the analogue of Newton's $2^{\textrm{nd}}$ Law of Motion. We should remark that this original result is much more general than the simple version given here and is applicable to much more general quantum Hamiltonians. For details see Refs. [18] and [19].

\npgni Nelson's theory never proved popular because of certain philosophical problems but it does have one enormous advantage over Schrödinger quantum mechanics; it enables one to predict first arrival times for quantum particles or at least their probability distributions in a very simple way. Many experimentalists think these are obervable and they are impossible to predict in Schrödinger or Heisenberg theories in any straight forward way. Semi-classical mechanics and the limiting case of Nelson's stochastic mechanics open a new window on such times.

\npgni We next give a brief account of two of the main theorems supporting our ideas. These should be read alongside Freidlin and Wentzell, Ref.[13], which sheds more light on these issues.

\subsection{Classical Mechanics as a Limiting Case of Schrödinger Quantum Mechanics}

\npgni We now prove the main theorem of this section. For each fixed $t\in(0,T)$, we assume that $D_{t}:x_{0}\rightarrow x[x_{0},\nabla S_{0}(x_{0}),t]$ is a $C^1$ diffeomorphism $D_{t}:\mathbb{R}^n\rightarrow\mathbb{R}^n$, with $D_{0}=id$, and $C^1$ inverse $D_{t}^{-1}:x\rightarrow x_{0}(x,t)$, $D_{t}^{-1}:\mathbb{R}^n\rightarrow\mathbb{R}^n$. Denote by $J_{t}(x)$ or $J_{t}(x_{0})$ the Jacobian of $D_{t}^{-1}$, $J_{t}(x_{0})=J_{t}(x)=|\partial x_{0}^i/\partial x^j|$, where it is understood that, if $J_{t}(x_{0})$ is required, we put $x=x[x_{0},\nabla S_{0}(x_{0}),t]$. Then we assume that $S=\{x\in\mathbb{R}^n|J_{t}(x)=0\}$ has Lebesgue measure zero. Given sufficient regularity there is a global inversion function theorem guaranteeing that these conditions are satisfied for sufficiently small $T$. (See Ref. [11]). We then have the theorem:

\begin{theorem}

Let $\psi_{\hbar}(x,t)$ be the solution of the Schrödinger equation

$$\dfrac{\partial\psi_{\hbar}}{\partial t}=i\dfrac{\hbar}{2}\nabla_{x}^2\psi_{\hbar}+\dfrac{V(x)}{i\hbar},$$

\npgni with Cauchy data $\psi_{\hbar}(x,0)=\exp\{iS_{0}(x)/\hbar\}\phi_{0}\in L^2(\mathbb{R}^n)$, where $V\in\{L^2(\mathbb{R}^n)+L^{\infty}(\mathbb{R}^n)\}$ is real-valued and $\phi_{0}$ (independent of $\hbar$) is such that $F(\tau)=||\nabla_{x}^2J_{\tau}^{1/2}(x)\phi_{0}[x_{0}(x,\tau)]||_{L^2}$ $\in L^1(0,t)$, $||\;||_{L^2}$ being the $L^2$ norm with respect to $x$. Then, for each fixed $t\in(0,T)$,

$$\exp[-iS(x,t)/\hbar]\psi_{\hbar}(x,t)\rightarrow J_{t}^{1/2}(x)\phi_{0}[x_{0}(x,t)],$$

\npgni in the $L^2$ norm with respect to $x$, as $\hbar\rightarrow0$.

\end{theorem}

\begin{proof}

\npgni If $V\in\{L^2(\mathbb{R}^n)+L^{\infty}(\mathbb{R}^n)\}$, the Hamiltonian $H_{0}=[(-\hbar^2/2)\nabla^2+V]$ is essentially self-adjoint on some suitable domain in $L^2(\mathbb{R}^n)$ and $H$, the extension of $H_{0}$, is such that $(iH)$ generates a continuous unitary one-parameter group $U(t)$ on $L^2(\mathbb{R}^n)$. Writing $\psi_{t}(x)=\psi(x,t)$ and $U(t)=\exp(-itH/\hbar)$ we obtain $\psi_{t}=[U(t)\psi_{0}](x)$.

\npgni The diffeomorphism $D_{t}:\mathbb{R}^n\rightarrow\mathbb{R}^n$ induces a unitary map $U_{0}(t):L^2(\mathbb{R}^n)\rightarrow L^{2'}(\mathbb{R}^n)$. We define the map $U_{0}(t)$ by $\psi_{t}\in L^2(\mathbb{R}^n)$ and $U_{0}(t)\psi_{t}=\phi_{t}\in L^{2'}(\mathbb{R}^n)$, according to

$$\phi_{t}(x_{0})=J_{t}^{-1/2}(x)\exp[-iS(x,t)/\hbar]\psi_{t},\;\;\;\;\;\textrm{a.e.},$$

\npgni where on the r.h.s. $x=x[x_{0},\nabla S_{0}(x_{0}),t]$ and $||\phi_{t}||_{L_{2'}}^2=\int|\phi_{t}(x_{0})|^2d^nx_{0}$. The assumptions on $D_{t}$ ensure that $U_{0}(t)$ is an isometry,

$$||\phi_{t}||_{L_{2'}}=||U_{0}(t)\psi_{t}||_{L_{2'}}=||\psi_{t}||_{L_{2}}.$$

\npgni It is a simple matter to check that $\textrm{Ran}U_{0}(t)=L^{2'}(\mathbb{R}^n)$. Indeed, defining $U_{0}^{-1}(t)$ by $\phi\in L^{2'}(\mathbb{R}^n)$,

$$[U_{0}^{-1}(t)\phi](x)=J_{t}^{1/2}(x)\exp[iS(x,t)/\hbar]\phi[x_{0}(x,t)],\;\;\;\;\;\textrm{a.e.},$$

\npgni $U_{0}^{-1}(t)\phi\in L^2(\mathbb{R}^n).$ Hence, $U_{0}^{-1}(t)=U_{0}^{\ast}(t)$, where $U_{0}^{\ast}$ denotes the adjoint of $U_{0}$, $U_{0}^{\ast}:L^{2'}(\mathbb{R}^n)\rightarrow L^2(\mathbb{R}^n)$.

\npgni We define the putative evolution operator $\tilde{U}(t,s)$, for $\phi_{s}\in L^{2'}(\mathbb{R}^n)$, according to

$$\tilde{U}(t,s)=U_{0}(t)U(t-s)U_{0}^{\ast}(s).$$

\npgni Here, in $\tilde{U}(t,s)$, $s$ denotes the initial time and $t$ denotes the final time $(t>s)$ and $\tilde{U}(t,s)$ has the evolution group property $(u>t>s)$ $\tilde{U}(u,t)\tilde{U}(t,s)=\tilde{U}(u,s)$.

\npgni Evidently $\tilde{U}(t,s)$ is unitary on $L^{2'}(\mathbb{R}^n)$, because $U(t-s)$ is unitary on $L^2(\mathbb{R}^n)$.

\npgni The infinitesimal generator $(iA(t))$ of the evolution operator $\tilde{U}(t,s)$ is defined by

$$iA(t)\tilde{U}(t,s)=\dfrac{d\tilde{U}(t,s)}{dt}=\lim_{k\to 0^{+}}\dfrac{\tilde{U}(t+k,s)-\tilde{U}(t,s)}{k},$$

\npgni so that

$$A(t)=\bigg\{-i\dfrac{dU_{0}(t)}{dt}\tilde{U}(t,s)U_{0}^{\ast}(s)-iU_{0}(t)\dfrac{dU(t-s)}{dt}U_{0}^{\ast}(s)\bigg\}\tilde{U}^{\ast}(t,s),$$

\npgni where $d/dt=\partial/\partial t|_{x_{0}}$ is the time derivative at constant $x_{0}$,

$$\dfrac{\partial}{\partial t}\bigg|_{x_{0}}=\dfrac{\partial}{\partial t}\bigg|_{x}+\nabla S.\nabla_{x}.$$

\npgni Using the fact that $J$ satisfies the continuity equation

$$\dfrac{\partial J_{t}}{\partial t}+\nabla S.\nabla J_{t}+J_{t}\nabla^2S=0,$$

\npgni we obtain

$$-i\dfrac{dU_{0}(t)}{dt}=-\hbar^{-1}\left(\dfrac{\partial S}{\partial t}+|\nabla S|^2+\dfrac{i\hbar}{2}\nabla^2S\right)U_{0}(t).$$

\npgni Also, from the above we obtain

$$-i\dfrac{dU(t-s)}{dt}=\left(-\dfrac{H}{\hbar}-i\nabla S.\nabla\right)U(t-s).$$

\npgni Combining these results, gives the operator identity

$$A(t)=-\hbar^{-1}\bigg\{\dfrac{\partial S}{\partial t}+|\nabla S|^2+\dfrac{i\hbar}{2}\nabla^2S\bigg\}-\hbar^{-1}U_{0}(t)HU_{0}^{\ast}(t)-iU_{0}(t)\nabla S.\nabla U_{0}^{\ast}(t).$$

\npgni Putting $H=[-(\hbar^2/2)\nabla^2+V]$ and recalling that $S$ satisfies the Hamilton-Jacobi equation, we arrive at

$$A(t)=-\hbar^{-1}\left(\dfrac{|\nabla S|^2}{2}+\dfrac{i\hbar}{2}\nabla^2S\right)+\dfrac{\hbar}{2}U_{0}(t)\nabla^2U_{0}^{\ast}(t)-iU_{0}(t)\nabla S.\nabla U_{0}^{\ast}(t).\;\;\;\;(\textrm{a})$$

\npgni However,

$$U_{0}(t)\nabla^2U_{0}^{\ast}(t)=J_{t}^{-1/2}\nabla^2J_{t}^{1/2}+U_{0}(t)[\nabla^2,\exp(iS/\hbar)]J_{t}^{1/2}$$

\npgni i.e.

$$U_{0}(t)\nabla^2U_{0}^{\ast}(t)=J_{t}^{-1/2}\nabla^2J_{t}^{1/2}+J_{t}^{-1/2}i\hbar^{-1}\nabla S.\nabla J_{t}^{1/2}+U_{0}(t)\nabla.i\hbar^{-1}\nabla SU_{0}^{\ast}(t),$$

\npgni where $[\nabla^2,\exp(iS/\hbar]=\nabla^2\exp(iS/\hbar)-\exp(iS/\hbar)\nabla^2$, is the commutator of $\nabla^2$
and $\exp(iS/\hbar)$.

\npgni Hence using the operator identities

$$\nabla.\nabla S-\nabla S.\nabla=\nabla^2S\;\;\;\textrm{and}\;\;\;\nabla S.\nabla U_{0}^{\ast}(t)=\exp(iS/\hbar)[\nabla S.\nabla J_{t}^{1/2}+i\hbar^{-1}|\nabla S|^2J_{t}^{1/2}],$$

\npgni and splitting up the last term in Eq.(a), we finally obtain

$$A(t)=\dfrac{+\hbar}{2}J_{t}^{-1/2}\nabla^2J_{t}^{1/2}.$$

\npgni This expresses $A(t)$ as a differential operator on a sufficiently small domain\vspace{-5mm}

\npgni $D_{t}\subset L^{2'}(\mathbb{R}^n)$, $D_{t}=\{\phi\in L^{2'}(\mathbb{R}^n)|A(t)\phi\in L^{2'}(\mathbb{R}^n)$\}. Here $\nabla^2=\nabla_{x}^2$ must be expressed in the curvilinear coordinates $x_{0}(x,t)$, so too with $J_{t}$.

\npgni It is not difficult to show that $A(t)$ as defined above is symmetric. [We can extend the domain of definition of $A(t)$ by defining $\nabla^2$ as a pseudodifferential operator by taking Fourier transforms. In this way we can make $A(t)$ self-adjoint, but this is hardly worthwhile here.]

\npgni Putting $A(t)=-\dfrac{\hbar}{2}H(t)$, we see that

$$i\dfrac{\partial\phi_{t}}{\partial t}=\dfrac{+\hbar}{2}H(t)\phi_{t}.$$

\npgni Integrating and using the symmetry of $H(\tau)$ gives, for $\phi_{0}\in\cap_{\tau\in(0,t)}D_{\tau}$,

$$(\phi_{0},\phi_{t})_{L_{2'}}-(\phi_{0},\phi_{0})_{L_{2'}}=-\dfrac{i\hbar}{2}\int\limits_{0}^t(\phi_{0},H(\tau)\phi_{\tau})_{L_{2'}}d\tau=-\dfrac{i\hbar}{2}\int\limits_{0}^t(H(\tau)\phi_{0},\tilde{U}(\tau,0)\phi_{0})_{L_{2'}}d\tau.$$

\npgni Using the Cauchy-Schwarz inequality and the isometric property of $\tilde{U}(\tau,0)$,

$$|(\phi_{0},\phi_{t})_{L_{2'}}-(\phi_{0},\phi_{0})_{L_{2'}}|\le\dfrac{\hbar||\phi_{0}||_{L_{2'}}}{2}\int\limits_{0}^t||H(\tau)\phi_{0}||_{L_{2'}}d\tau=\dfrac{\hbar||\phi_{0}||_{L_{2'}}}{2}\int\limits_{0}^tF(\tau)d\tau,$$

\npgni where $F(\tau)=||\nabla_{x}^2J_{\tau}^{1/2}(x)\phi_{0}[x_{0}(x,\tau)]||_{L_{2}}$. Hence, if $F(\tau)\in L^1(0,t)$,

$$|(\phi_{0},\phi_{t})_{L_{2'}}-(\phi_{0},\phi_{0})_{L_{2'}}|\rightarrow0,\;\;\;\;\textrm{as}\;\;\hbar\rightarrow0.$$

\npgni Using the isometric property of $\tilde U(0,t)$, we obtain

$$||\phi_{t}-\phi_{0}||_{L_{2'}}^2=2(\phi_{0},\phi_{0})_{L_{2'}}-(\phi_{0},\phi_{t})_{L_{2'}}-(\phi_{t},\phi_{0})_{L_{2'}}\le2|(\phi_{0},\phi_{t})_{L_{2'}}-(\phi_{0},\phi_{0})_{L_{2'}}|\rightarrow0,$$

\npgni as $\hbar\rightarrow0$. Finally, changing integration variables once more, we arrive at

$$||\exp[-iS(x,t)/\hbar]\psi_{\hbar}(x,t)-J_{t}^{1/2}(x)\phi_{0}[x_{0}(x,t)]||_{L_{2}}\rightarrow0,\;\;\;\;\textrm{as}\;\;\hbar\rightarrow0.$$

\npgni This proves the theorem.

\end{proof}\vspace{-8mm}

\npgni This theorem is tantamount to

\begin{center}
    "quantum mechanics $\rightarrow$ classical mechanics as $\hbar\rightarrow0"$
\end{center}

\npgni There is a corresponding result for the Schrödinger heat equation embodying some quantum features including a very useful Feynman-Kac formuula - the elementary formula of Elworthy and Truman, Ref. [11], giving our solution of the Schrödinger heat equation as a sum over paths where the paths are the sample paths of the Nelson diffusion process for the corresponding Schrödinger equation. This leads to our Burgers-Zeldovich model in astronomy incorporating Newtonian quantum gravity's main characteristic spirals. (See Ref. [22]).

\npgni When we consider non-stationary states, as a simple example,

$$\dfrac{\partial u^{\sigma}(\vct{x},t)}{\partial t} =\dfrac{\sigma^2}{2}\bigtriangleup_{\vct{x}}u^{\sigma}(\vct{x},t)+\dfrac{V(\vct{x})}{\sigma^2}u^{\sigma}(\vct{x},t),$$

\npgni where $u^{\sigma}:\mathbb{R}^d\times\mathbb{R}_{+}\rightarrow\mathbb{R}$, with initial condition,

$$u^{\sigma}(\vct{x},0)=T_{0}(\vct{x})\exp\left(-\dfrac{\mathcal{S}_{0}(\vct{x})}{\sigma^2}\right),$$

\npgni $\mathcal{S}_{0}$ independent of $\sigma$, which it turns out is the fluid viscosity. The Burgers-Zeldovich fluid model arises from the Hopf-Cole transformation,

$$\vct{v}^{\sigma}(\vct{x},t)=-\sigma^2\boldsymbol{\nabla}\ln{(u^{\sigma}(\vct{x},t))},$$

\npgni giving

$$\dfrac{\partial\vct{v}^{\sigma}(\vct{x},t)}{\partial t}+\vct{v}^{\sigma}(\vct{x},t).\boldsymbol{\nabla}\vct{v}^{\sigma}(\vct{x},t)+\boldsymbol{\nabla}V(\vct{x})=\dfrac{\sigma^2}{2}\bigtriangleup\vct{v}(\vct{x},t),\;\;\;\;(\ast)$$

\npgni with initial condition

$$\vct{v}(\vct{x},0)=\boldsymbol{\nabla}\mathcal{S}_{0}(\vct{x})-\sigma^2\boldsymbol{\nabla}\ln{T_{0}(\vct{x})}.$$

\npgni We assume $T_{0}>0$ and $\int T_{0}^2(\vct{x})d^dx=1$, which will ensure mass conservation or conservation of probability. The link is the Hamilton-Jacobi-Bellman equation, in the limiting case as $\sigma\rightarrow0$ of

$$\dfrac{\partial\mathcal{S}^{\sigma}(\vct{x},t)}{\partial t}+\dfrac{1}{2}|\boldsymbol{\nabla}\mathcal{S}^{\sigma}(\vct{x},t)|^2+V(\vct{x})=\dfrac{\sigma^2}{2}\bigtriangleup\mathcal{S}^{\sigma}(\vct{x},t),$$

\npgni with $\mathcal{S}^{\sigma}(\vct{x},0)=\mathcal{S}_{0}(\vct{x})-\sigma^2\ln{T_{0}(\vct{x})}$.

\npgni As in our treatment of the original Schrödinger equation, we need the existence of our diffeomorphism, $D_{t}$, which, to avoid confusion with convected derivative $\dfrac{D}{dt}$, we denote by $\Phi_{t}:\mathbb{R}^d\rightarrow\mathbb{R}^d$, where in the limit as $\sigma\rightarrow0$, $(\ast)$ reads

$$\dfrac{D}{\partial t}\vct{v}(\vct{x},t)=-\boldsymbol{\nabla}V(\vct{x}),\;\;\;\;\textrm{Newton's Law},$$

\npgni $t\in(0,T)$ for sufficiently small $T$.

\npgni We need some elementary results from classical mechanics. If $\tilde{\mathcal{S}}(y)$ is defined as

$$\tilde{\mathcal{S}}(y)=\mathcal{S}_{0}(y)+\dfrac{1}{2}\int\limits_{0}^t|\dot{\Phi}_{s}(y)|^2ds-\int\limits_{0}^t V(\Phi_{s}(y))ds,$$

\npgni the Lagrangian form, and 

$$\mathcal{S}(x,t)=\tilde{\mathcal{S}}_{t}(\Phi(x)),\;\;\;\;\dot{\Phi}_{t}(\vct{x})=\boldsymbol{\nabla}\mathcal{S}(\Phi_{t}(\vct{x}),t).$$

\npgni More properly this should be taken to define $\Phi_{s}$ with $\Phi_{s=0}=\textrm{id}$, $s\in(0,t)$, $t\in(0,T)$, $T$ sufficiently small.

\npgni Of course $\Phi_{t}$ ceases to be a diffeomorphism when $\textrm{det}(D\Phi_{t}(\vct{x}))=0$, i.e.

$$\textrm{det}\left(\dfrac{\partial}{\partial\vct{x}_{0}}\vct{x}(\vct{x}_{0}(\vct{x},t),t)\right)=0,$$

\npgni i.e. when classical paths focus at a point forming caustics. We assume this does not happen for $0<t<T$, for a global $T$, called the caustic time. We then have the result:

\begin{theorem} (Elementary formula).
Given that the no-caustic condition holds with a caustic time $T>0$, then for all $t\in[0,T),$

$$u^{\sigma}(x,t)=\exp\left(-\dfrac{\mathcal{S}(x,t)}{\sigma^2}\right)\mathbb{E}\Bigg{\{}T_{0}(Y_{t}^{\sigma})\Bigg{(}-\dfrac{1}{2}\int\limits_{0}^t\bigtriangleup\mathcal{S}(Y_{s}^{\sigma},t-s)ds\Bigg{)}\Bigg{\}},$$

\npgni where

$$dY_{s}^{\sigma}=-\boldsymbol{\nabla}\mathcal{S}(Y_{s}^{\sigma},t-s)ds+\sigma dB_{s},\;\;\;Y_{0}^{\sigma}=x,\;\;s\in[0,t].$$
    
\end{theorem}

\begin{proof}
Define the local martingale,

$$M_{t}=\exp\left(-\dfrac{1}{2\sigma^2}\int\limits_{0}^t|\boldsymbol{\nabla}\mathcal{S}(Y_{s}^{\sigma},t-s)|^2ds+\dfrac{1}{\sigma}\int\limits_{0}^t\boldsymbol{\nabla}\mathcal{S}(Y_{s}^{\sigma},t-s).dB_{s}\right).$$

\npgni Assume $\mathcal{S}(x,t)$ is suitably differentiable to allow us to apply It\^{o}'s formula which gives,

\npgni \hspace{5mm} $\sigma\displaystyle{\int\limits_{0}^t}\boldsymbol{\nabla}\mathcal{S}(Y_{s}^{\sigma},t-s).dB_{s}=\mathcal{S}_{0}(Y_{t}^{\sigma})-\mathcal{S}(x,t)$

\npgni \hspace{20mm} $+\displaystyle{\int\limits_{0}^t}\left(|\boldsymbol{\nabla}\mathcal{S}(Y_{s}^{\sigma},t-s)|^2+\dfrac{\partial}{\partial t}\mathcal{S}(Y_{s}^{\sigma},t-s)-\dfrac{\sigma^2}{2}\bigtriangleup\mathcal{S}(Y_{s}^{\sigma},t-s)\right)ds$.

\npgni Thus, from the Girsanov-Cameron-Martin theorem for the change of measure applied to the Feynman-Kac formula,

\npgni \hspace{5mm} $u^{\sigma}(x,t)=\mathbb{E}\Bigg{\{}T_{0}(Y_{t}^{\sigma})\exp\Bigg{(}-\dfrac{1}{\sigma^2}\mathcal{S}(x,t)-\dfrac{1}{2}\displaystyle{\int\limits_{0}^t}\bigtriangleup\mathcal{S}(Y_{s}^{\sigma},t-s)ds$

\npgni \hspace{35mm} $+\dfrac{1}{\sigma^2}\displaystyle{\int\limits_{0}^t}\bigg{(}\dfrac{\partial}{\partial t}\mathcal{S}(Y_{s}^{\sigma},t-s)+\dfrac{1}{2}|\boldsymbol{\nabla}\mathcal{S}(Y_{s}^{\sigma},t-s)|^2+V(Y_{s}^{\sigma}\bigg{)}ds\Bigg{)}\Bigg{\}}$.

\npgni Clearly the result now follows by virtue of the Hamilton-Jacobi equation.

\end{proof}

\npgni This elementary result reveals the simple connection to Nelson's mechanical paths satisfying a version of Newton's $2^{\textrm{nd}}$ Law. As a simple application of this elegant formula we can investigate the behaviour of $u^{\sigma}(x,t)$ as $\sigma\rightarrow0$. Consider the family of uniformly bounded continuous matrix valued functions,

$$\mathcal{M}_{T}:=\bigg{\{}A:[0,T)\rightarrow\mathbb{R}^{d\times d}\bigg{|}A(t)\;\textrm{continuous},||A||:=\sup\limits_ {s\in[0,T]}||A(s)||<\infty\bigg{\}}.$$

\npgni Define the time ordered exponential $\exp_{+}:\mathcal{M}_{T}\rightarrow\mathcal{M}_{T}$,

$$\exp_{+}A(t)=I+\int\limits_{0}^tA(t_{1})dt_{1}+\int\limits_{0}^tdt_{1}\int\limits_{0}^{t_{1}}dt_{2}A(t_{1})A(t_{2})+\dotsc,$$

\npgni where the products in the multiple integrals are ordered strictly increasing in times $t_{j}$ and the infinite series converge in $\mathcal{M}_{T}$, i.e. uniformly for $t\in[0,T)$.

\begin{lemma}
Suppose that $A\in\mathcal{M}_{T}$, then $\exp_{+}A(s)$ solves the initial value problem,

$$\dfrac{d}{ds}\exp_{+}A(s)=A(s)\exp_{+}A(s),\;\;\;\;\exp_{+}A(0)=I,$$

\npgni and furthermore,

$$\det\;(\exp_{+}A(t))=\exp\left(\int\limits_{0}^t\mathrm{tr}\;A(s)ds\right).$$

\end{lemma}

\begin{proof}
Given that $A\in\mathcal{M}_{T}$, then the first part follows by a simple calculation from the definition of $\exp_{+}$ and the fact that an infinite series can be differentiated term by term if the differentiated terms are continuous and the differentiated series is uniformly convergent. Thus,

$$\exp_{+}A(s+h)=\exp_{+}A(s)+hA(s)\exp_{+}A(s)+o(h)$$

\npgni and so,

$$\det\exp_{+}A(s+h)=\det(I+hA(s))\det(\exp_{+}A(s))+o(h).$$

\npgni Therefore, we find $\det\exp_{+}A(s)$ satisfies the initial value problem,

$$\dfrac{d}{ds}\det\exp_{+}A(s)=\textrm{tr}A(s)\det\exp_{+}A(s),\;\;\;\;\det\exp_{+}A(0)=I,$$

\npgni which gives the result.

\end{proof}

\npgni Now recall from theorem 5.2 that by definition $Y_{s}^0$ satisfies the initial value problem,

$$\dot{Y}_{s}^0=-\boldsymbol{\nabla}\mathcal{S}(Y_{s}^0,t-s),\;\;\;\;Y_{0}^0=x,$$

\npgni for $s\in[0,t]$ with $t\in[0,T)$. It therefore follows that $Y_{s}^0$ actually represents the time reversed classical mechanical path $\Phi(\Phi_{t}^{-1}(x))$ which has initial momentum $\boldsymbol{\nabla}\mathcal{S}_{0}(\Phi_{t}^{-1}(x))$ and reaches the point $x$ at time $t$. That is,

$$Y_{s}^0=\Phi_{t-s}(\Phi_{t}^{-1}(x)),\;\;\;\;s\in[0,t].$$

\npgni Also define the Van Vleck matrix,

$$\mathcal{S}''(s)=D^2\mathcal{S}(y,s)|_{y=\Phi_{s}(\Phi_{t}^{-1}(x))}.$$

\begin{theorem}
Given that the no-caustic condition holds with a caustic time $T>0$, then for all $t\in[0,T)$,

$$\exp\left(\int\limits_{0}^t\mathrm{tr}\;\mathcal{S}''(s)ds\right)=\det D\Phi_{t}(y)|_{y=\Phi_{t}^{-1}(x)}.$$
\end{theorem}

\begin{proof}
From lemma 5.3 it follows that,
$$\det\left(\exp_{+}\mathcal{S}''(t)\right)=\exp\left(\int\limits_{0}^t\textrm{tr}\;\mathcal{S}''(s)ds\right).$$
\end{proof}

\npgni However

$$\dfrac{\partial^2}{\partial y_{i}\partial y_{j}}\mathcal{S}(y,s)\bigg{|}_{y=\Phi_{s}(\Phi_{t}^{-1}(x))}=\dfrac{\partial}{\partial y_{i}}\big{\{}\dot{\Phi}_{s}(\Phi_{s}^{-1}(y))\big{\}}_{j}\bigg{|}_{y=\Phi_{s}(\Phi_{t}^{-1}(x))},$$

\npgni and so,

$$\mathcal{S}''(s)=\left(\dfrac{d}{ds}D\Phi_{s}(\Phi_{t}^{-1}(x))\right)D\Phi_{s}^{-1}(\Phi_{s}(\Phi_{t}^{-1}(x))).$$

\npgni Thus the Jacobi field matrix $J(s)=D\Phi_{s}(\Phi_{t}^{-1}(x))$ satisfies,

$$\dfrac{d}{ds}J(s)=\mathcal{S}(s)J(s).$$

\npgni Therefore, by the first part of Lemma 5.3 and the uniqueness of solutions to such linear differential equations, it follows that $J(s)=\exp_{+}\mathcal{S}''(s)$.

\npgni Moreover we find that,

$$\exp\left(-\dfrac{1}{2}\int\limits_{0}^t\bigtriangleup\mathcal{S}(Y_{s}^0,t-s)ds\right)=\exp\left(-\dfrac{1}{2}\int\limits_{0}^t\textrm{tr}\;\mathcal{S}''(s)ds\right)=\sqrt{|\det D\Phi_{t}^{-1}(x)|}>0,$$

\npgni for $0<t<T$ provided $\bigtriangleup\mathcal{S}$ is bounded below. We can now recover a familiar result from the above on the asymptotic behaviour from the stochastic elementary formula and the dominated convergence theorem.

\begin{corollary}
Given that the no-caustic condition holds with a caustic time $T>0$ and that $\bigtriangleup\mathcal{S}(x,t)$ is bounded below, then for $t\in[0,T),$

$$\lim\limits_{\sigma\rightarrow0}\exp\left(\dfrac{\mathcal{S}(x,t)}{\sigma^2}\right)u^{\sigma}(x,t)=T_{0}(\Phi_{t}^{-1}(x))\sqrt{|\det D\Phi_{t}^{-1}(x)|}.$$

\end{corollary}

\npgni This gives us the behaviour of $u^{\sigma}(x,t)$ for small $\sigma$ up to first order in $\sigma^2$.

\npgni As was mentioned in the Introduction, if we define,

$$\mathcal{S}^{\sigma}(x,t)=-\sigma^2\ln{u^{\sigma}(x,t)},$$

\npgni then it is a simple exercise to prove that $\mathcal{S}(x,t)$ satisfies the Hamilton-Jacobi-Bellman equation,

$$\dfrac{\partial}{\partial t}\mathcal{S}^{\sigma}(\vct{x},t)+\dfrac{1}{2}|\boldsymbol{\nabla}\mathcal{S}^{\sigma}(\vct{x},t)|^2+V(\vct{x})=\dfrac{\sigma^2}{2}\bigtriangleup\mathcal{S}^{\sigma}(\vct{x},t),$$

\npgni with the initial condition $\mathcal{S}^{\sigma}(x,0)=S_{0}(x)-\sigma^2\ln{T_{0}(x)}$. Furthermore, if we take the gradient of the Hamilton-Jacobi-Bellman equation we find the viscous Burgers equation,

$$\dfrac{\partial}{\partial t}v^{\sigma}(x,t)+(v^{\sigma}(x,t).\nabla)v^{\sigma}(x,t)+\nabla V(x)=\dfrac{\sigma^2}{2}\bigtriangleup v^{\sigma}(x,t),$$

\npgni with the initial condition $v^{\sigma}(x,0)=\nabla S_{0}(x)-\sigma^2\nabla\ln{T_{0}(x)}$ where

$$v^{\sigma}(x,t)=-\sigma^2\nabla\ln{u^{\sigma}(x,t)}=\nabla\mathcal{S}^{\sigma}(x,t).$$

\npgni From the stochastic elementary formula we then obtain:

\begin{corollary}
Given that the no-caustic condition holds with a caustic time $T>0$, then for $t\in[0,T)$,

$$\mathcal{S}^{\sigma}(x,t)=\mathcal{S}(x,t)-\sigma^2\ln{\mathbb{E}\Bigg{\{}T_{0}(Y_{t}^{\sigma})\exp\Bigg{(}-\dfrac{1}{2}\int\limits_{0}^t\bigtriangleup\mathcal{S}(Y_{s}^{\sigma},t-s)ds\Bigg{)}\Bigg{\}}}$$\vspace{-5mm}

\npgni and,\vspace{-5mm}

$$v^{\sigma}(x,t)=v(x,t)-\sigma^2\nabla\ln{\mathbb{E}\Bigg{\{}T_{0}(Y_{t}^{\sigma})\exp\Bigg{(}-\dfrac{1}{2}\int\limits_{0}^t\nabla.v(Y_{s}^{\sigma},t-s)ds\Bigg{)}\Bigg{\}}},$$

\npgni where $\mathcal{S}$ satisfies the Hamilton-Jacobi equation and $v$ satisfies the inviscid Burgers equation.

\end{corollary}

\npgni Moreover, given appropriate conditions on $V$, $S_{0}$ and $T_{0}$, we find,

$$\mathcal{S}^{\sigma}(x,t)\rightarrow\mathcal{S}(x,t),\;\;\;\;v^{\sigma}(x,t)\rightarrow v(x,t),$$

\npgni as $\sigma\rightarrow0$ as expected.


\section{The Schrödinger-Heat Equation and a Burgers-Zeldovich Model for Spiral Galaxy Formation}\vspace{-5mm}

\npgni In this section we build on our work outlined in NQG II (Ref. [29]) on a possible application of the Schrödinger-heat equation and related Burgers-Zeldovich equation to the formation of circular spiral galaxies. We construct exact solutions and show how they can provide predictions for quantities such as the arc length of a spiral arm. We begin our analysis starting with the Schrödinger equation.

\npgni Let $\Psi(\vct{r},\epsilon^2)=\exp\left(\dfrac{R+iS}{\epsilon^2}\right)$, $\vct{r}\in\mathbb{R}^3$, be a stationary state solution of the Schrödinger equation of a unit mass particle moving in the potential, $V(r)$, where $r=|\vct{r}|$, i.e. for energy $E$,

$$-\dfrac{1}{2}\epsilon^4\bigtriangleup\Psi+V\Psi=E\Psi.$$

\npgni Writing $\epsilon^2=i\sigma^2$ and $U(\vct{r},\sigma^2)=\Psi(\vct{r},i\sigma^2)$, formally at least, $U$ is a solution of the Schrödinger-Heat equation,

$$\dfrac{1}{2}\sigma^4\bigtriangleup U+VU=EU.$$\vspace{3mm}

\begin{lemma}

If $U=\exp\left(\dfrac{S-iR}{\sigma^2}\right)$, for real $R$ and $S$, is a complex-valued solution of the Schrödinger-Heat equation then $U=\exp\left(\dfrac{S+R}{\sigma^2}\right)$ is a real solution of the modified heat equation,

$$\dfrac{1}{2}\sigma^4\bigtriangleup U+\left(V-|\boldsymbol{\nabla}R|^2\right)U=EU.$$

\end{lemma}

\begin{proof}
Equating real and imaginary parts of the Schrödinger-Heat equation in terms of $R$ and $S$ gives the result.
\end{proof}

\npgni We define $\tilde{V}=V-|\boldsymbol{\nabla}R|^2$ to be the modified potential of the system. This is analogous to including the Bohm potential in quantum mechanics.

\subsection{Circular Spiral Galaxies}\vspace{-5mm}

\npgni We now apply this result to the 3-dimensional stationary circular state solutions of the Schrödinger equation for a unit mass particle moving in the Coulomb potential $V=-\dfrac{\mu}{r}$, where $r=|\vct{r}|$. Transforming the appropriate stationary state solution, $\Psi$, (see ref.[10]), we can construct an exact solution $U$ of the modified heat equation with

$$R=-\dfrac{\mu}{\lambda}r+\dfrac{\lambda}{2}\ln{(x^2+y^2)}+\sigma^2\tan^{-1}\left(\dfrac{y}{x}\right),\;\;S=\lambda\tan^{-1}\left(\dfrac{y}{x}\right)-\dfrac{\sigma^2}{2}\ln{(x^2+y^2)},$$

\npgni where $\lambda>0$ is defined by $E=-\dfrac{\mu^2}{2\lambda^2}$ and $0<\sigma^2<\lambda$. We note that $\sigma$ need not be small and $\boldsymbol{\nabla}R.\boldsymbol{\nabla}S\ne0$.

\npgni If $\vct{v}=\boldsymbol{\nabla}R+\boldsymbol{\nabla}S$ defines the deterministic part of the particle velocity then

$$\dfrac{d\vct{v}}{dt}=-\boldsymbol{\nabla}V_{\textrm{eff}},\;\;\textrm{where}\;\;V_{\textrm{eff}}=\dfrac{\mu}{\lambda r}(\lambda-\sigma^2)-\dfrac{\lambda^2+\sigma^4}{x^2+y^2}-\dfrac{\mu^2}{\lambda^2}.$$

\npgni If $\vct{X}=(x,y,z)$ are the cartesian coordinates and $\vct{v}=\dfrac{d\vct{X}}{dt}$, for $t>0$,\vspace{5mm}

$$\dfrac{dx}{dt}=-\dfrac{\mu x}{\lambda r}+\dfrac{\lambda(x-y)}{x^2+y^2}-\dfrac{\sigma^2(x+y)}{x^2+y^2},$$

$$\dfrac{dy}{dt}=-\dfrac{\mu y}{\lambda r}+\dfrac{\lambda(x+y)}{x^2+y^2}+\dfrac{\sigma^2(x-y)}{x^2+y^2},$$\vspace{-5mm}

\npgni \hspace{45mm} $\dfrac{dz}{dt}=-\dfrac{\mu z}{\lambda r}$.

\npgni This system can be solved exactly. Firstly if $r=\sqrt{x^2+y^2+z^2}$ then

$$\dfrac{dr}{dt}=\dfrac{\lambda(\lambda-\sigma^2)-\mu r}{\lambda r},$$

\npgni giving the solution,

$$|r_{c}-r|=|r_{c}-r_{0}|\exp\left(-\dfrac{\mu}{\lambda r_{c}}t+\dfrac{r_{0}-r}{r_{c}}\right),$$

\npgni where $r_{0}=r(t=0)$ and $r_{c}=\dfrac{\lambda(\lambda-\sigma^2)}{\mu}>0$. Thus as $t\rightarrow \infty$, $r\rightarrow r_{c}$.

\npgni The solution for $z$ can be given in terms of $r$ as,

$$|z|=|z_{0}|\exp\left(-\dfrac{\mu(r-r_{0})}{\lambda(\lambda-\sigma^2)}\right)\exp\left(-\dfrac{\mu^2}{\lambda^2(\lambda-\sigma^2)}t\right),$$

\npgni where $z_{0}=z(t=0)$. From this we see that as $t\rightarrow\infty$, $z\rightarrow0$.

\npgni These results show that the particle paths spiral on to the circular orbit with radius $r_{c}=\dfrac{\lambda(\lambda-\sigma^2)}{\mu}$, in the plane $z=0$.\vspace{-3mm}

\npgni To get a sense of the spiral orbit we consider those particles for which $z_{0}=0$ and thus $z(t)=0$ for $t>0$. In this case we look at the polar distance $\rho=\sqrt{x^2+y^2}$ and polar angle $\phi=\tan^{-1}\left(\dfrac{y}{x}\right)$. From the equations defining the dynamical system it is easy to show that,

$$\dfrac{d\rho}{dt}=\dfrac{\lambda-\sigma^2}{\rho}-\dfrac{\mu}{\lambda}\;\;\;\;\;\textrm{and}\;\;\;\;\;\dfrac{d\phi}{dt}=\dfrac{\lambda+\sigma^2}{\rho^2}.$$

\npgni The equation for $\rho$ is identical to the equation for $r$ from which we deduce that as $t\rightarrow\infty$, $\rho\rightarrow r_{c}$, as expected. For the polar equation $\rho=\rho({\phi})$ we note that,

$$\dfrac{d\rho}{d\phi}=-\dfrac{\mu}{\lambda(\lambda+\sigma^2)}(\rho^2-r_{c}\rho),$$

\npgni which can be solved to give,

$$\Big{|}\dfrac{\rho-r_{c}}{\rho}\Big{|}=\Big{|}\dfrac{\rho_{0}-r_{c}}{\rho_{0}}\Big{|}\exp\left(-\dfrac{(\lambda-\sigma^2)}{(\lambda+\sigma^2)}\phi\right),$$

\npgni where $\rho_{0}=\rho(t=0)$ and $\phi=0$ when $t=0$. Again, as expected, $\rho\rightarrow r_{c}$ as $\phi\rightarrow\infty$.

\npgni If we assume $\rho>\rho_{0}>r_{c}$ i.e we are looking for particles approaching from the outer parts of the system then the above solution gives,

$$\rho=\rho(\phi)=\dfrac{r_{c}}{(1-be^{-\alpha\phi})},\;\;\;\;\phi>0,$$\vspace{-5mm}

\npgni where $b=\dfrac{\rho_{0}-r_{c}}{\rho_{0}}$ and $\alpha=\dfrac{(\lambda-\sigma^2)}{(\lambda+\sigma^2)}$.

\npgni Could this be a suitable model for the spiral arm of a spiral galaxy? To test this idea we develop a simple test involving the arc length of a spiral arm between $\phi_{0}$ and $\phi$. The arc length, $L$, of a polar curve, $\rho=\rho(\phi)$, is given by

$$L=\int\limits_{\phi_{0}}^{\phi}\sqrt{\rho^2+\left(\dfrac{d\rho}{d\phi}\right)^2}d\phi.$$\vspace{-5mm}

\npgni This integral can be performed exactly to give,

$$L=\dfrac{r_{c}}{\alpha}\Bigg{[}\sinh^{-1}\left(\dfrac{s}{\alpha(1-s)}\right)-\dfrac{\sqrt{(1+\alpha^2)s^2-2\alpha^2s+\alpha^2}}{s}\Bigg{]}_{s=1-b}^{s=1-be^{-\alpha\phi}},$$\vspace{-5mm}

\npgni where $\alpha=\dfrac{(\lambda-\sigma^2)}{(\lambda+\sigma^2)}$, $b=\dfrac{\rho_{0}-r_{c}}{\rho_{0}}$, $r_{c}=\dfrac{\lambda(\lambda-\sigma^2)}{\mu}$ and we have taken $\phi_{0}=0$.\vspace{-3mm}

\npgni To apply this result we need to approximate $r_{c}$ and thus $\alpha$.

\npgni Below is a computer simulation of the full 3-dimensional dynamical system showing clearly the spiral nature of a particle's path.\vspace{-5mm}

\begin{center}
\includegraphics[scale=0.75]{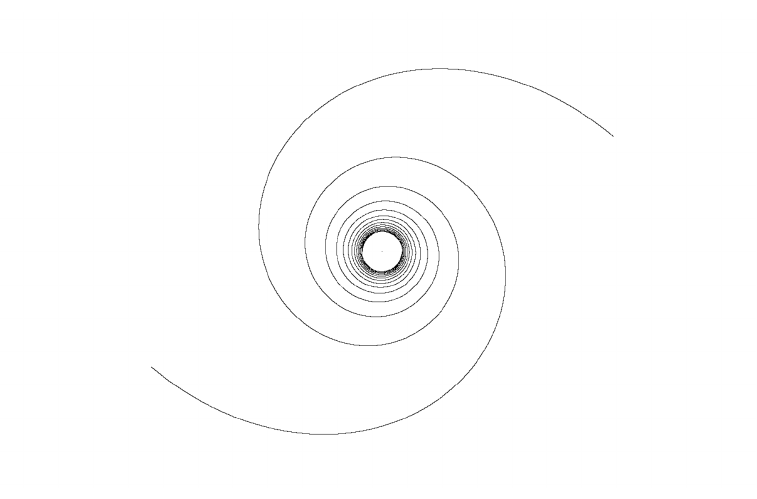}\\
Circular Spiral Path.
\end{center}\vspace{-5mm}

\npgni The image below of the Pinwheel Galaxy (M101) would appear an ideal candidate to test these results. We invite the astronomical community to make detailed observations which could help to support our ideas.

\begin{center}
\includegraphics[scale=0.5]{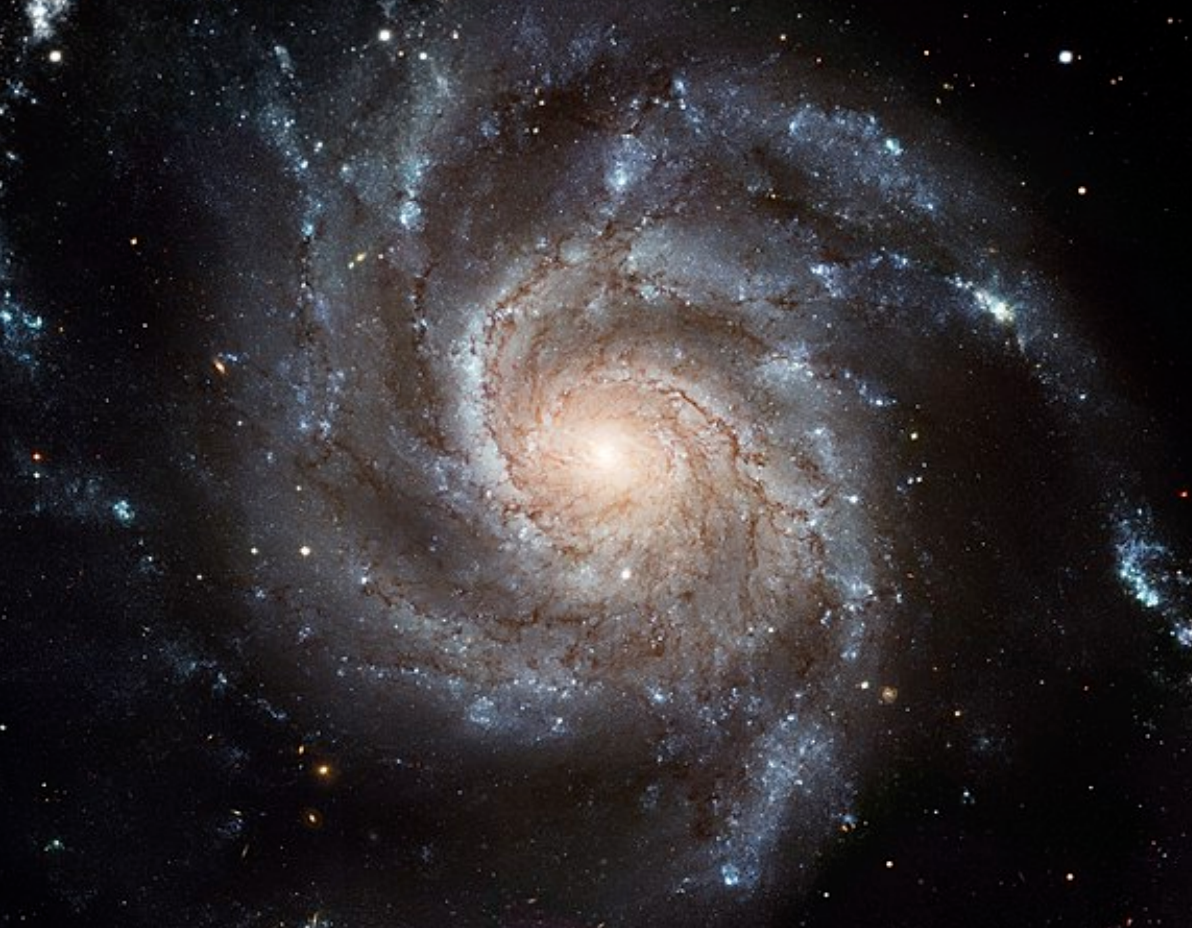}\\
Pinwheel Galaxy (M101) Credit: Hubble.
\end{center}

\npgni It is worth noting that in the case where $\sigma^2\sim0$ the dynamical system coincides with that derived from the corresponding Schrödinger wave function when $\epsilon^2\sim0$. i.e. both descriptions lead to the semi-classical realisation of the classical circular orbit, as expected, the details of which can be found in NQG I. This can be extended to the wave function corresponding to the atomic elliptic state giving rise to the semi-classical Kepler elliptical orbit (see NQG I). This naturally leads us to postulate this as a potential accretion model for the formation of circular ring galaxies such as Hoag's Object (see below) and elliptical ring galaxies such as AM 0644-741 (shown above).\vspace{-15mm}

\begin{center}
\includegraphics[scale=0.5]{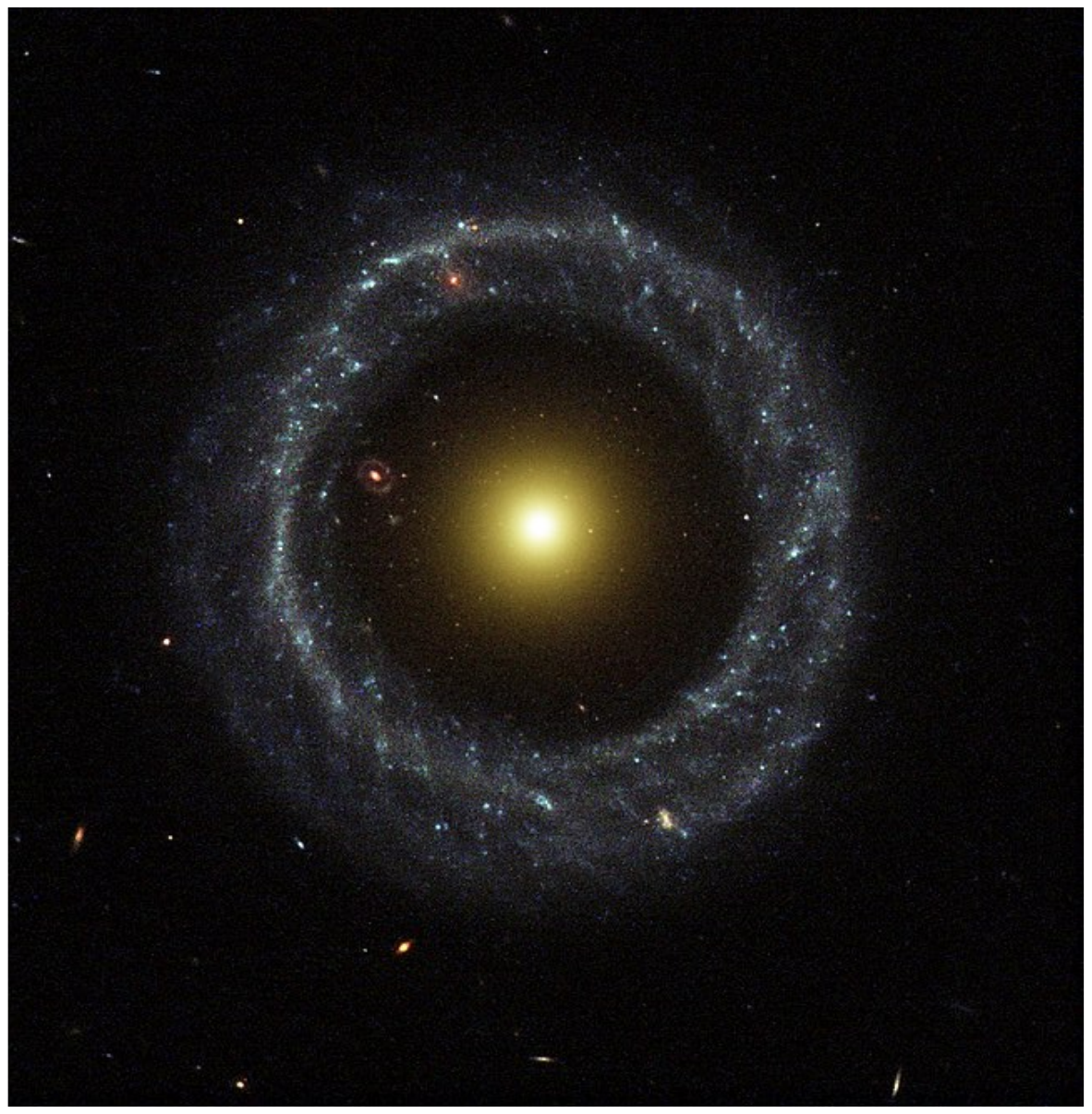}\\\vspace{-20mm}
Hoag's Object Credit: Hubble.
\end{center}

\begin{center}
"For we are most fearfully and wonderfully made"
\end{center}

\npgni In addition to the spiral nature of the solution if $\sigma^2\sim\lambda$ the $3^{\textrm{rd}}$ component of angular momentum is approximately \textbf{twice} the classically predicted value! Could this help to explain dark matter data reproducing the observed rotation curve for galaxies' gaseous parts. We intend to explore this further in future work. 

\npgni In short, we believe that taking this method further to include a vector potential (possibly derived from a magnetic effect) we can construct an alternative mechanism capable of revealing a spiral rotation curve as observed within many galaxies which does not follow the classical picture defined by the observed mass. Furthermore it is possible to construct a local solution in the semi-classical limit as $\sigma^2\sim0$ which converges to the classical orbit in a spiral fashion. We ask could a model of this type explain the observed phenomena which currently are thought to be a result of dark matter?

\section{Appendix - Some of the Wonders of the Weierstrass Elliptic Function}

\npgni In trying to make the paper self-contained we have included this appendix. We concentrate our attention on the results we have found to be most useful both here and in our work on the KLMN problem. Needless to say this is no substitute for original references e.g. Ref. [14].

\npgni We first give some historical background which we have pieced together from Greenhill (Ref. [14]). This illustrates the importance of fractional linear transformations in this context. Some credit here should go to the mysterious Mr R. Russell (Ref. [24]). The first result obviates the need to find the roots of the quartic $f$ in some circumstances. The final result only needs the initial position i.e. $r_{0}$. Needless to say $f$ is our original quartic $X$, $f(u)=2\left(E+\mu u-\dfrac{u^2}{2}(C-Bu)^2\right), u=\dfrac{1}{r}$, for applications to KLMN problems. (See Ref. [27]).\vspace{5mm}

\begin{lemma} (Euler-Lagrange)

\npgni Assume $X(x)$ and $Y(y)$ are both quartics of the same form, in binomial notation $(a,b,c,d,e)(x,1)^4$ and $(a,b,c,d,e)(y,1)^4$, where

$$(a,b,c,d,e)(x,1)^4=ax^4+4bx^3+6cx^2+4dx+e.$$\vspace{-8mm}

\npgni Then the equation,

$$\frac{dx}{\sqrt{X}}\pm\frac{dy}{\sqrt{Y}}=0,\;\;\;\;\;(\ast)$$\vspace{-8mm}

\npgni has the integral

$$\left(\frac{\sqrt{X(x)}\mp\sqrt{Y(y)}}{x-y}\right)^2=a(x+y)^2+4b(x+y)+C,$$

\npgni where $C$ is Euler's constant. (Upper signs go with upper signs and lower with lower).

\end{lemma}

\begin{proof}

\npgni Set $z={\displaystyle\int\limits^x}\dfrac{dx}{\sqrt{X}}$ so that $\dfrac{dx}{dz}=\sqrt{X}$ and $\dfrac{dy}{dz}=\mp\sqrt{Y}$.

\npgni Writing $\dfrac{d^2x}{dz^2}=\dfrac{d}{dx}\left(2^{-1}\left(\dfrac{dx}{dz}\right)^2\right)$ and $\dfrac{d^2y}{dz^2}=\dfrac{d}{dy}\left(2^{-1}\left(\dfrac{dy}{dz}\right)^2\right)$ gives

$$\dfrac{d^2x}{dz^2}=2(ax^3+3bx^2+3cx+d)=2X_{1},$$

$$\dfrac{d^2y}{dz^2}=2(ay^3+3by^2+3cy+d)=2Y_{1}.$$

\npgni Setting $x+y=p$ and $x-y=q$, gives

$$\dfrac{dp}{dz}=\sqrt{X}\mp\sqrt{Y},\;\;\;\dfrac{dq}{dz}=\sqrt{X}\pm\sqrt{Y},$$

$$\dfrac{d^2p}{dz^2}=2X_{1}+2Y_{1}=\dfrac{1}{2}a(p^3+3pq^2)+3b(p^2+q^2)+6cp+4d,$$

$$\dfrac{dp}{dz}\dfrac{dq}{dz}=(\sqrt{X}\mp\sqrt{Y})(\sqrt{X}\pm\sqrt{Y})=X-Y=\dfrac{1}{2}apq(p^2+q^2)+bq(3p^2+q^2)+6cpq+4dq,$$

\npgni giving\vspace{-5mm}

$$q\dfrac{d^2p}{dz^2}-\dfrac{dp}{dz}\dfrac{dq}{dz}=apq^3+2bq^3,$$\vspace{-10mm}

\npgni or\vspace{-5mm}

$$\dfrac{2}{q^2}\dfrac{dp}{dz}\dfrac{d^2p}{dz^2}-\dfrac{2}{q^3}\dfrac{dq}{dz}\left(\dfrac{dp}{dz}\right)^2=2ap\dfrac{dp}{dz}+4b\dfrac{dp}{dz}$$\vspace{-10mm}

\npgni i.e.\vspace{-5mm}

$$\left(\dfrac{1}{q}\dfrac{dp}{dz}\right)^2=ap^2+4bp+C,\;\;\;\textrm{as promised.}$$

\end{proof}

\npgni In our case we are interested in ${\displaystyle\int}\dfrac{dx}{\sqrt{X}}\pm{\displaystyle\int}\dfrac{dy}{\sqrt{Y}}=\epsilon{\displaystyle\int}\dfrac{ds}{\sqrt{S}}$, $S$ the Weierstrass form,

\npgni $S=4s^3-g_{2}s-g_{3}$, $g_{2}$, $g_{3}$ the quartic invariants and $\epsilon=\pm1$. In $(\ast)$ we take the lower minus sign, $x$ will be an upper and $y$ a lower limit of integration and for us, the all important final integration variable, $s$, is determined by the choice, $C=4(s+c)$, for the above $c$. Then we get, when first sign is minus,

$$s=\dfrac{F(x,y)+\sqrt{X}\sqrt{Y}}{(x-y)^2},$$

\npgni where $F(x,y)=(ax^2+2bx+c)y^2+2(bx^2+2cx+d)y+cx^2+2dx+e$, $F(x,y)=F(y,x)$, being a symmetrical quadri quadric.

\npgni A further differentiation gives

$$\sqrt{X}\dfrac{\partial s}{\partial x}=-\dfrac{(Y_{1}x+Y_{2})}{(x-y)^3}\sqrt{X}+\dfrac{(X_{1}y+X_{2})}{(x-y)^3}\sqrt{Y},$$\vspace{-10mm}

\npgni namely

\npgni \hspace{5mm} $\sqrt{X}\dfrac{\partial s}{\partial x}=-\dfrac{(ay^3+3by^2+3cy+d)x+by^3+3cy^2+3dy+e}{(x-y)^3}\sqrt{X}$

\npgni \hspace{30mm} $+\dfrac{(ax^3+3bx^2+3cx+d)y+bx^3+3cx^2+3dx+e}{(x-y)^3}\sqrt{Y}=\sqrt{Y}\dfrac{\partial s}{\partial y}$.

\npgni Now the fractional linear transformations enter, following the mysterious Mr R. Russell. We make the fractional linear transformation, $t=\dfrac{(\tau x+y)}{(\tau+1)}$, in the quartic\vspace{-5mm}

\npgni $(a,b,c,d,e)(t,1)^4$, giving\vspace{-5mm}

$$A\tau^4+4B\tau^3+6C\tau^2+4D\tau+E=X\tau^4+4(X_{1}y+X_{2})\tau^3+6F(x,y)\tau^2+4(Y_{1}x+Y_{2})+Y.$$\vspace{-10mm}

\npgni We now consider the Weierstrass form, $S=4s^3-g_{2}s-g_{3}$, more fully, $g_{2}$ and $g_{3}$ being the quartic invariants of our original quartic, so the invariants for the new quartic are $G_{2}$ and $G_{3}$, $G_{2}=(x-y)^4g_{2}$ and $G_{3}=(x-y)^6g_{3}$. So

$$S=\dfrac{(C-\sqrt{A}\sqrt{E})^3}{2(x-y)^6}-g_{2}\dfrac{(C-\sqrt{A}\sqrt{E})}{2(x-y)^2}-g_{3}.$$

$$\textrm{r.h.s.}=\dfrac{(C-\sqrt{A}\sqrt{E})^3-G_{2}(C-\sqrt{A}\sqrt{E})-2G_{3}}{2(x-y)^6}.$$

\npgni But $G_{2}=AE-4BD+3C^2$ and $G_{3}=ACE+2BCD-C^3-AD^2-B^2E$. After a simple piece of algebra we obtain

$$\textrm{r.h.s.}=\dfrac{(D\sqrt{A}-B\sqrt{E})^2}{(x-y)^6}=\dfrac{((Y_{1}x+Y_{2})\sqrt{X}-(X_{1}y+X_{2})\sqrt{Y})^2}{(x-y)^6}.$$\vspace{-8mm}

\npgni Hence, $S$ is a perfect square and

$$\sqrt{S}=\sqrt{X}\dfrac{\partial s}{\partial y}=\pm\sqrt{Y}\dfrac{\partial s}{\partial y},$$\vspace{-8mm}

\npgni where an extra $\pm$ sign could be incorporated. Keeping $y$ fixed, so $dy=0$,

$$\dfrac{ds}{\sqrt{S(s)}}=\dfrac{dx}{\sqrt{X(x)}},$$

\npgni for Weierstrass form as desired.

\npgni When $x$ and $y$ are both upper limits of integration, both variable, then we need

$${\displaystyle\int}\dfrac{dx}{\sqrt{X}}-{\displaystyle\int}\dfrac{dy}{\sqrt{Y}}=-{\displaystyle\int}\dfrac{ds}{\sqrt{S}}=-{\displaystyle\int}\dfrac{\frac{\partial s}{\partial x}dx+\frac{\partial s}{\partial y}dy}{\sqrt{S}}=-{\displaystyle\int\limits_{x}^{y}}\dfrac{dv}{\sqrt{V(v)}},$$

\npgni where $V$ is of the same form as $X$ and $Y$.

\npgni We stress here that all this follows for $X(x)=(a,b,c,d,e)(x,1)^4$, for the choice of Euler's constant $C=4(s+c)$, defining the new integration variable $s$. This $s$ is just a fractional linear function of $x$ e.g. as Greenhill [14, pp154] shows, if $X(x)$ has 4 distinct roots: $\alpha, \beta, \gamma, \delta,$ one choice could be

$$s=\dfrac{a}{12}\dfrac{(\alpha-\beta)(\alpha-\gamma)(\alpha-\delta)}{(x-\alpha)}\left(\dfrac{x-\beta}{\alpha-\beta}+\dfrac{x-\gamma}{\alpha-\gamma}+\dfrac{x-\delta}{\alpha-\delta}\right),$$

$$s-e_{1}=\dfrac{a}{4}(\alpha-\gamma)(\alpha-\delta)\dfrac{(x-\beta)}{(x-\alpha)},$$

$$s-e_{2}=\dfrac{a}{4}(\alpha-\delta)(\alpha-\beta)\dfrac{(x-\gamma)}{(x-\alpha)},$$

$$s-e_{3}=\dfrac{a}{4}(\alpha-\beta)(\alpha-\gamma)\dfrac{(x-\delta)}{(x-\alpha)},$$

\npgni $e_{i}$, $i=1,2,3$, being the roots of $4s^3-g_{2}s-g_{3}=0$.

\npgni Recalling that in our KLMN problem the Lorentz force is the crucial new element and the Lorentz group is intimately related to $\textrm{SL}(2,\mathbb{C})$, it is not surprising that the fractional linear transformation group with non-zero determinant, $\textrm{PSL}(2,\mathbb{R})$, appears here.

\npgni Euler's result foreshadowed the addition formulae for elliptic functions e.g. if $a=0$, $b=1$, $c=0$, $d=-\dfrac{g_{2}}{4}$, $e=-g_{3}$, setting $x=\wp(u)$, $y=\wp(v)$, $\sqrt{X}=-\wp'(u)$, $\sqrt{Y}=-\wp'(v)$, from the above,

$$\wp^{-1}(x)+\wp^{-1}(y)=\wp^{-1}(s)=\wp^{-1}\left(\dfrac{1}{4}\dfrac{(\sqrt{X}-\sqrt{Y})^2}{(x-y)^2}-x-y\right)$$\vspace{-8mm}

\npgni i.e.\vspace{-5mm}

$$\wp(u+v)=\dfrac{1}{4}\left(\dfrac{\wp'(u)-\wp'(v)}{\wp(u)-\wp(v)}\right)^2-\wp(u)-\wp(v).$$

\npgni This in 1761 from Eulerian, Poissonian and Newtonian mechanics! One cannot exaggerate this breakthrough in generalising all trigonometric identities.

\npgni We need:

\begin{corollary}

$$\zeta(u-v)+\zeta(u+v)-2\zeta(u)=\dfrac{\wp'(u)}{\wp(u)-\wp(v)}\;\;\textrm{for}\;\;\zeta'=-\wp$$\vspace{-10mm}

\npgni and\vspace{-5mm}

$$\wp'(u)\int\dfrac{dv}{\wp(u)-\wp(v)}=\ln\left(\dfrac{\sigma(u+v)}{\sigma(u-v)}\right)-2v\zeta(u),$$\vspace{-8mm}

\npgni where $\dfrac{\sigma'}{\sigma}=\zeta$.
    
\end{corollary}

\begin{proof}

\npgni From the above

$$\wp(u)+\wp(v)+\wp(u+v)=\dfrac{1}{4}\left(\dfrac{\wp'(u)-\wp'(v)}{\wp(u)-\wp(v)}\right)^2.$$

\npgni Replacing $v$ by $-v$ and recalling that $\wp$ is even gives

$$\wp(u)+\wp(v)+\wp(u-v)=\dfrac{1}{4}\left(\dfrac{\wp'(u)+\wp'(v)}{\wp(u)-\wp(v)}\right)^2.$$

\npgni Subtracting we obtain

$$\wp(u-v)-\wp(u+v)=\dfrac{\wp'(u)\wp'(v)}{(\wp(u)-\wp(v))^2}.$$

\npgni Integrating w.r.t $v$ gives

$$\zeta(u-v)+\zeta(u+v)+c=\dfrac{\wp'(u)}{\wp(u)-\wp(v)},$$

\npgni c a constant of integration determined by setting $v=0$ so that $\wp(v)=\infty$ thus giving $c=-2\zeta(u)$.\vspace{-5mm}

\npgni i.e.\vspace{-5mm}

$$\zeta(u-v)+\zeta(u+v)-2\zeta(u)=\dfrac{\wp'(u)}{\wp(u)-\wp(v)}.$$

\npgni Since $\dfrac{\sigma'}{\sigma}=\zeta$ a final integration w.r.t. $v$ gives the desired result.
    
\end{proof}

\npgni We now have the main results for $\wp$ we use in the main body of the work. We are indebted to Greenhill and the mysterious Mr R. Russell for these. For further investigations we believe they will be useful as are those that follow. We include them here for completeness.

\npgni The following result, originally due to Weierstrass and published by Biermann, Ref. [4], gives a general formula in terms of the Weierstrass function, $\wp$, for $u(t)$. This formula has been considered and used by several authors (e.g.  Refs. [6] and [15]) but their proofs leave out some details which we include here. The formula has also been simplified further by Mordell (See Ref. [15]). This is part of the next theorem which is proved in a straight forward way. We have made the proof self-contained for non-experts, but of course it leans heavily on the historical results above. Indeed it follows easily from them after the gap is filled.\vspace{5mm}

\begin{theorem}

If $z=z(x)=\int\limits_{a}^{x}\{f(t)\}^{-\frac{1}{2}}dt$, $x>a$, where $f(x)$ is a quartic polynomial with no repeated factors, then 

$$x=a+\frac{\{f(a)\}^{\frac{1}{2}}\wp'(z)+\frac{1}{2}\left(\wp(z)-\frac{1}{24}f''(a)\right)f'(a)+\frac{1}{24}f(a)f'''(a)}{2\left(\wp(z)-\frac{1}{24}f''(a)\right)^2-\frac{1}{48}f(a)f^{''''}(a)},$$

\npgni where $\wp(z)=\wp(z,g_{2},g_{3})$ is the Weierstrass function formed with the invariants $g_{2}$ and $g_{3}$ of the quartic $f$. The above formula can be simplified to give Mordell's formula\vspace{3mm}

$$x=a+\frac{8(12\wp(z)+f''(a))f(a)-6f'^2(a)}{48\{f(a)\}^{\frac{1}{2}}\wp'(z)-(24\wp(z)-f''(a))f'(a)-2f(a)f'''(a)}.$$

\npgni In our case the sign of $\{f(a)\}^{\frac{1}{2}}$ is determined by the physics, $x=u_{t=0}=a$.
 
\end{theorem}\vspace{5mm}

\begin{proof}

Let $f(t)=a_{0}t^4+4a_{1}t^3+6a_{2}t^2+4a_{3}t+a_{4}$ with invariants $g_{2}$ and $g_{3}$:

$$g_{2}=a_{0}a_{4}-4a_{1}a_{3}+3a_{2}^2\;\;\;;\;\;\;g_{3}=a_{0}a_{2}a_{4}+2a_{1}a_{2}a_{3}-a_{2}^3-a_{0}a_{3}^2-a_{1}^2a_{4}.$$

\npgni Following a standard approach (see e.g. Ref. [6]), define a new integration variable, $s$, by

$$s=s(t)=\frac{1}{4}\left(\frac{\sqrt{f(t)}+\sqrt{f(a)}}{t-a}\right)^2-\frac{1}{4}a_{0}(t+a)^2-a_{1}(t+a)-a_{2}.$$

\npgni This can be written as $s(t)=\dfrac{F(t)+\sqrt{f(t)}\sqrt{f(a)}}{2(t-a)^2}$, where\vspace{5mm}

$$F(t)=a_{0}a^2t^2+2a_{1}(at^2+a^2t)+a_{2}(t^2+4at+a^2)+2a_{3}(t+a)+a_{4},$$

\npgni which in turns leads to 

$$s(t)=\frac{1}{2}f(a)(t-a)^{-2}+\frac{1}{4}(t-a)^{-1}+\frac{1}{24}f''(a)+\frac{\sqrt{f(t)}\sqrt{f(a)}}{2(t-a)^2},$$

\npgni so that our integration becomes

$$z=\int\limits_{a}^{x}\{f(t)\}^{-\frac{1}{2}}dt=\int\limits_{\infty}^{s(x)}\frac{ds}{G(t)},$$

\npgni where $G(t)=\left(\dfrac{f'(t)}{4(t-a)^2}-\dfrac{f(t)}{(t-a)^3}\right)\sqrt{f(a)}-\left(\dfrac{f(a)}{(t-a)^3}+\dfrac{f'(a)}{4(t-a)^2}\right)\sqrt{f(t)}$.

\npgni If the integral on the r.h.s. is of Weierstrass form we need to show that

$$G^2(t)=4s^3(t)-g_{2}s(t)-g_{3}.$$ 

\npgni This is the key element in the proof of the Biermann-Weierstrass formula an elegant proof of which is above. (See R. Russell Ref. [24]).

\npgni More straight-forwardly, to this end we note that

$$(t-a)^6G^2(t)=g_{1}^2(t)f(a)+g_{2}^2(t)f(t)+2g_{1}(t)g_{2}(t)\sqrt{f(t)}\sqrt{f(a)}:=A(t)+B(t)\sqrt{f(t)}\sqrt{f(a)}$$\vspace{-5mm}

\npgni and

$$(t-a)^6(4s^3(t)-g_{2}s(t)-g_{3})=\frac{1}{2}(F^3(t)+3F(t)f(t)f(a)-g_{2}(t-a)^4F(t)-2g_{3}(t-a)^6)$$\vspace{-8mm}

\npgni \hspace{70mm} $+\dfrac{1}{2}(3F^2(t)+f(t)f(a)-g_{2}^4)\sqrt{f(t)}\sqrt{f(a)}$

\npgni \hspace{52mm}$:=C(t)+D(t)\sqrt{f(t)}\sqrt{f(a)}$,

\npgni where\vspace{-5mm}

$$g_{1}(t)=(a_{0}a+a_{1})t^3+3(a_{1}a+a_{2})t^2+3(a_{2}a+a_{3})t+a_{2}a+a_{4},$$\vspace{-5mm}

$$g_{2}(t)=(a_{0}a^3+3a_{1}a^2+3a_{2}a+a_{3})t+a_{1}a^3+3a_{2}a^2+3a_{3}a+a_{4},$$

$$F(t)=(a_{0}a^2+2a_{1}a+a_{2})t^2+2(a_{1}a^2+2a_{2}a+a_{3})t+a_{2}a^2+2a_{3}a+a_{4}.$$

\npgni Following significant algebraic reduction we obtain the following identities:-

\npgni $B(t)\equiv D(t)\equiv\{2a_{0}^2a^4+8a_{0}a_{1}a^3+(6a_{0}a_{2}+6a_{1}^2)a^2+(2a_{0}a_{3}+6a_{1}a_{2})a+2a_{1}a_{3}\}t^4$

\npgni \hspace{30mm} $+\{8a_{0}a_{1}a^4+(12a_{0}a_{2}+20a_{1}^2)a^3+(6a_{0}a_{3}+42a_{1}a_{2})a^2$

\npgni \hspace {62mm} $+(2a_{0}a_{4}+12a_{1}a_{3}+18a_{2}^2)a+2a_{1}a_{4}+6a_{2}a_{3}\}t^3$

\npgni \hspace{30mm} $+\{(6a_{0}a_{2}+6a_{1}^2)a^4+(6a_{0}a_{3}+42a_{1}a_{2})a^3+(36a_{1}a_{3}+36a_{2}^2)a^2$

\npgni \hspace {78mm} $+(6a_{1}a_{4}+42a_{2}a_{3})a+6a_{2}a_{4}+6a_{3}^2\}t^2$

\npgni \hspace{30mm} $+\{(2a_{0}a_{3}+6a_{1}a_{2})a^4+(2a_{0}a_{4}+12a_{1}a_{3}+18a_{2}^2)a^3$

\npgni \hspace{52mm} $+(6a_{1}a_{4}+42a_{2}a_{3})a^2+(12a_{2}a_{4}+20a_{3}^2)a+8a_{3}a_{4}\}t$

\npgni \hspace{30mm} $+2a_{1}a_{3}a^4+(2a_{1}a_{4}+6a_{2}a_{3})a^3+(6a_{2}a_{4}+6a_{3}^2)a^2+8a_{3}a_{4}a+2a_{4}^2$.

\npgni And with the aid of the last identity

\npgni $A(t)\equiv C(t)\equiv\{2a_{0}^3a^6+12a_{0}^2a_{1}a^5+(18a_{0}a_{1}^2+12a_{0}^2a_{2})a^4+(30a_{0}a_{1}a_{2}+6a_{0}^2a_{3}+4a_{1}^3)a^3$

\npgni \hspace{15mm} $+(14a_{0}a_{1}a_{3}+9a_{0}a_{2}^2+a_{0}^2a_{4}+6a_{1}^2a_{2})a^2+(2a_{0}a_{1}a_{4}+6a_{0}a_{2}a_{3}+4a_{1}^2a_{3})a$

\npgni \hspace{60mm} $+(2a_{0}a_{1}a_{4}+6a_{0}a_{2}a_{3}+4a_{1}^2a_{3})a+a_{0}a_{3}^2+a_{1}^2a_{4}\}t^6$

\npgni \hspace{15mm} $+\{12a_{0}^2a_{1}a^6+(60a_{0}a_{1}^2+12a_{0}^2a_{2})a^5+(114a_{0}a_{1}a_{2}+6a_{0}^2a_{3}+60a_{1}^3)a^4$

\npgni \hspace{60mm} $+(52a_{0}a_{1}a_{3}+54a_{0}a_{2}^2+2a_{0}^2a_{4}+132a_{1}^2a_{2})a^3$

\npgni \hspace{15mm} $+(12a_{0}a_{2}a_{4}+6a_{0}a_{3}^2+48a_{1}a_{2}a_{3}+6a_{1}^2a_{4})a+2a_{0}a_{3}a_{4}+6a_{1}a_{2}a_{4}+4a_{1}a_{3}^2\}t^5$\vspace{5mm}

\npgni \hspace{15mm} $+\{(18a_{0}a_{1}^2+12a_{0}^2a_{2})a^6+(114a_{0}a_{1}a_{2}+6a_{0}^2a_{3}+60a_{1}^3)a^5$

\npgni \hspace{15mm} $+(60a_{0}a_{1}a_{3}+90a_{0}a_{2}^2+300a_{1}^2a_{2})a^4+(10a_{0}a_{1}a_{4}+90a_{0}a_{2}a_{3}+360a_{1}a_{2}^2+140a_{1}^2a_{3})a^3$

\npgni \hspace{60mm} $+(12a_{0}a_{2}a_{4}+33a_{0}a_{3}^2+264a_{1}a_{2}a_{3}+33a_{1}^2a_{4}+108a_{2}^3)a^2$

\npgni \hspace{15mm} $+(12a_{0}a_{3}a_{4}+48a_{1}a_{2}a_{4}+48a_{1}a_{3}^2+72a_{2}^2a_{3})a+a_{0}a_{4}^2+14a_{1}a_{3}a_{4}+6a_{2}a_{3}^2+9a_{2}^2a_{4}\}t^4$

\npgni \hspace{15mm} $+\{(30a_{0}a_{1}a_{2}+6a_{0}^2a_{3}+4a_{1}^3)a^6+(52a_{0}a_{1}a_{3}+54a_{0}a_{2}^2+2a_{0}^2a_{4}+132a_{1}^2a_{2})a^5$

\npgni \hspace{60mm} $+(10a_{0}a_{1}a_{4}+90a_{0}a_{2}a_{3}+360a_{1}a_{2}^2+140a_{1}^2a_{3})a^4$

\npgni \hspace{30mm} $+(24a_{0}a_{2}a_{4}+16a_{0}a_{3}^2+528a_{1}a_{2}a_{3}+16a_{1}^2a_{4}+216a_{2}^3)a^3$

\npgni \hspace{60mm} $+(10a_{0}a_{3}a_{4}+90a_{1}a_{2}a_{4}+140a_{1}a_{3}^2+360a_{2}^2a_{3})a^2$

\npgni \hspace{15mm} $+(2a_{0}a_{4}^2+52a_{1}a_{3}a_{4}+132a_{2}a_{3}^2+54a_{2}^2a_{4})a+6a_{1}a_{4}^2+30a_{2}a_{3}a_{4}+4a_{3}^3\}t^3$

\npgni \hspace{15mm} $+\{(14a_{0}a_{1}a_{3}+9a_{0}a_{2}^2+a_{0}^2a_{4}+6a_{1}^2a_{2})a^6+(12a_{0}a_{1}a_{4}+48a_{0}a_{2}a_{3}+72a_{1}a_{2}^2+48a_{1}^2a_{3})a^5$

\npgni \hspace{60mm} $+(12a_{0}a_{2}a_{4}+33a_{0}a_{3}^2+264a_{1}a_{2}a_{3}+33a_{1}^2a_{4}+108a_{2}^3)a^4$

\npgni \hspace{15mm} $+(10a_{0}a_{3}a_{4}+90a_{1}a_{2}a_{4}+140a_{1}a_{3}^2+360a_{2}^2a_{3})a^3+(60a_{1}a_{3}a_{4}+300a_{2}a_{3}^2+90a_{2}^2a_{4})a^2$

\npgni \hspace{60mm} $+(6a_{1}a_{4}^2+114a_{2}a_{3}a_{4}+60a_{3}^3)a+12a_{2}a_{4}^2+18a_{3}^2a_{4}\}t^2$

\npgni \hspace{15mm} $+\{(2a_{0}a_{1}a_{4}+6a_{0}a_{2}a_{3}+4a_{1}^2a_{3})a^6+(12a_{0}a_{2}a_{4}+6a_{0}a_{3}^2+48a_{1}a_{2}a_{3}+6a_{1}^2a_{4})a^5$

\npgni \hspace{60mm} $+(12a_{0}a_{3}a_{4}+48a_{1}a_{2}a_{4}+48a_{1}a_{3}^2+72a_{2}^2a_{3})a^4$

\npgni \hspace{15mm} $+(2a_{0}a_{4}^2+52a_{1}a_{3}a_{4}+132a_{2}a_{3}^2+54a_{2}^2a_{4})a^3+(6a_{1}a_{4}^2+114a_{2}a_{3}a_{4}+60a_{3}^3)a^2$

\npgni \hspace{60mm} $+(12a_{2}a_{4}^2+60a_{3}^2a_{4})a+12a_{3}a_{4}^2\}t$

\npgni \hspace{15mm} $+(a_{0}a_{3}^2+a_{1}^2a_{4})a^6+(2a_{0}a_{3}a_{4}+6a_{1}a_{2}a_{4}+4a_{1}a_{3}^2)a^5$

\npgni \hspace{30mm} $+(a_{0}a_{4}^2+14a_{1}a_{3}a_{4}+6a_{2}a_{3}^2+9a_{2}^2a_{4})a^4+(6a_{1}a_{4}^2+30a_{2}a_{3}a_{4}+4a_{3}^3)a^2$

\npgni \hspace{60mm} $+(12a_{2}a_{4}^2+18a_{3}^2a_{4})a^2+12a_{3}a_{4}^2a+2a_{4}^3$.

\npgni It follows that with the appropriate choice of sign,

$$z=\int\limits_{a}^{x}\{f(t)\}^{-\frac{1}{2}}dt=\int\limits_{s(x)}^{\infty}\frac{ds}{\sqrt{4s^3-g_{2}s-g_{3}}},$$

\npgni i.e. $s(x)=\wp(z)=\wp(z;g_{2},g_{3}).$ The rest of the proof is taken from Biermann, Ref.  [4].

\npgni We know that $s=\dfrac{F(t)+\sqrt{f(t)}\sqrt{f(a)}}{2(t-a)^2}$ which can be rearranged to give

$$4(t-a)^4s^2-4(t-a)^2F(t)s+F^2(t)-f(t)f(a)=0.$$

\npgni Writing $f(t)=f(a)+f'(a)(t-a)+\frac{1}{2}f''(a)(t-a)^2+\frac{1}{6}f'''(a)(t-a)^3+\frac{1}{24}f''''(a)(t-a)^4$\vspace{-3mm}

\npgni and using the fact that $F(t)=f(a)+\frac{1}{2}f'(a)(t-a)+\frac{1}{12}f''(a)(t-a)^2$ leads to the quadratic\vspace{-3mm}

\npgni equation in $(t-a)$,

\npgni $\left(s^2-\dfrac{1}{12}f''(a)+\dfrac{1}{576}f''^2(a)-\dfrac{1}{96}f''''(a)f(a)\right)(t-a)^2$

\npgni \hspace{25mm} $+\left(-\dfrac{1}{2}f'(a)s+\dfrac{1}{48}f''(a)f'(a)-\dfrac{1}{24}f'''(a)f(a)\right)(t-a)$.

\npgni \hspace{70mm} $-f(a)s+\dfrac{1}{16}f'^2(a)-\dfrac{1}{12}f''(a)f(a)=0$.

\npgni Solving this quadratic for $t=x>a$ and $s=s(x)$ gives the Biermann-Weierstrass formula for $x$.

\npgni Rationalising the numerator in the fraction of the Biermann-Weierstrass formula by multiplying the numerator and denominator by

$$24\{f(a)\}^\frac{1}{2}\wp'(z)-(12\wp(z)-2^{-1}f''(a))f'(a)-f(a)f'''(a)$$

\npgni gives Mordell's formula for $x$.\vspace{-8mm}
   
\end{proof}

\npgni We should point out that when $u_{t=0}=u_{0}=a$ is a root of $f(u)=0$, the first formula reduces to

$$u(z)-u_{0}=\dfrac{f'_{0}}{4(\wp(z)-\frac{1}{24}f''_{0})},\;\;\;z=z(t),$$

\npgni where $f_{0}=f(u_{0})$, $f'_{0}=f'(u_{0})$, $f''_{0}=f''(u_{0})$.

\npgni This is very easy to prove (see e.g. Whittaker and Watson, [26]). The first formula attributed to Weierstrass and Biermann is given in an exercise in [26], but it is difficult to find complete proofs. This is why we have laboured the point here. We conclude with the formula for $z(t)$, e.g. for KLMN problems $u=\dfrac{1}{r}$, see Ref.[27]. We leave as an exercise the corresponding problem for two centres. (See Ref. [5]).

\begin{theorem} (Restoring the physical time)

\npgni From the above, the physical time $t$ is given in the KLMN problem in terms of $\alpha,\wp(\alpha)=-\dfrac{f_{0}''}{24}+\dfrac{f_{0}'}{4u_{0}}$ and $z={\displaystyle\int\limits_{u_{0}}^{u(t)}\dfrac{du}{\sqrt{f(u)}}}$ by

$$t=\dfrac{1}{u_{0}^2}\left(z-\dfrac{f_{0}'}{2}I(\alpha)+\dfrac{f_{0}'^{2}}{16}\dfrac{1}{\wp'(\alpha)}\dfrac{dI(\alpha)}{d\alpha}\right)\Bigg{|}_{\alpha},$$

\npgni where $I(\alpha)=-\dfrac{1}{\wp'(\alpha)}\left(2\zeta(\alpha)z+\ln\left(\dfrac{\sigma(z-\alpha)}{\sigma(z+\alpha)}\right)\right)$, $\wp(z)=\wp(z;g_{2},g_{3})$, $g_{2}$, $g_{3}$ the quartic invariants of $f(u)$, $z=z(t)>0$, for $u(t)\in(u_{0},u_{1})$, $z(t)$ the well-time for the part of the orbit between apses $u_{0}$ and $u_{1}$: $f(u_{0})=f(u_{1})=0$, $\dfrac{\sigma'}{\sigma}=\zeta$, $\zeta'=-\wp$, $f'_{0}=f'(u_{0})$ and $f''_{0}=f''(u_{0})$.

\end{theorem}\vspace{-3mm}

\begin{proof}

\npgni Setting $X=\wp(z)-\dfrac{f''_{0}}{24}$, Weierstrass gives $u(t)=u_{0}+\dfrac{f'_{0}}{4(\wp(z(t))-\frac{1}{24}f''_{0})}$.

\npgni Set $Y=u_{0}X+\dfrac{f'_{0}}{4}$ so that $X=\dfrac{1}{u_{0}}\left(Y-\dfrac{f'_{0}}{4}\right)$. Then observing that

$$\int\limits_{0}^{t}dt=\int\limits_{r_{0}}^{r(t)}\dfrac{dr}{\dot{r}}=\int\dfrac{dr}{\sqrt{f(u)}}=\int\limits_{u_{0}}^{u(t)}\dfrac{du}{u^2\sqrt{f(u)}}=\int\limits_{0}^{z(t)}\dfrac{dz}{u^2(z)},$$

\npgni we see that\vspace{-3mm}

$$t=\int\limits_{0}^{z(t)}\dfrac{X^2}{Y^2}dz=\int\limits_{0}^{z(t)}\left(1-\dfrac{f_{0}'}{2Y}+\dfrac{f_{0}'^2}{16Y^2}\right)dz.$$\vspace{-8mm}

\npgni And

$$\int\limits_{0}^{z}\dfrac{dz}{\wp(z)+d}=-\dfrac{1}{\wp'(\alpha)}\left(2\zeta(\alpha)z+\ln\left(\dfrac{\sigma(z-\alpha)}{\sigma(z+\alpha)}\right)\right),$$

\npgni where $\alpha=\wp^{-1}(d)$, $d=-\dfrac{f_{0}''}{24}+\dfrac{f_{0}'}{4u_{0}}$ in our case.

\end{proof}

\begin{corollary}

$$t=\dfrac{1}{u_{0}^2}\bigg\{z\left(1+\dfrac{f_{0}'\zeta(a_{0})}{2\wp'(a_{0})}-\left[\dfrac{2\wp(a_{0})}{(\wp'(a_{0}))^2}+\dfrac{2\wp''(a_{0})\zeta(a_{0})}{(\wp'(a_{0}))^3}\right]\dfrac{f_{0}'^2}{16}\right)+\dfrac{f_{0}'}{2\wp'(a_{0})}\ln\left(\dfrac{\sigma(z-a_{0})}{\sigma(z+a_{0})}\right)$$\vspace{-5mm}

\npgni \hspace{25mm} $+\dfrac{f_{0}'^2}{16}\left(\dfrac{\wp''(a_{0})}{(\wp'(a_{0}))^3}\ln\left(\dfrac{\sigma(z+a_{0})}{\sigma(z-a_{0})}\right)-\dfrac{(\zeta(z+a_{0})+\zeta(z-a_{0}))}{(\wp'(a_{0}))^2}\right)\bigg\}\Bigg{|}_{a_{0}=\alpha}$,

\npgni $z=z(t)$, $u(t)-u_{0}=\dfrac{f_{0}'}{4\left(\wp(z(t);g_{2},g_{3})-\frac{1}{24}f_{0}''\right)}$.
    
\end{corollary}

\npgni For the sake of completeness we include the results on the KLMN orbital elliptic curve which was our start point.

\begin{corollary}

$(\wp(z(t)),\wp'(z(t)))$ is given by

$$\wp(z)=\dfrac{(f(u)f(a))^{\frac{1}{2}}+f(a)}{2(u-a)^2}+\dfrac{f'(a)}{4(u-a)}+\dfrac{f''(a)}{24},$$

$$\wp'(z)=\left(\dfrac{f(u)}{(u-a)^3}-\dfrac{f'(u)}{4(u-a)^2}\right)(f(a))^{\frac{1}{2}}-\left(\dfrac{f(a)}{(u-a)^3}-\dfrac{f'(a)}{4(u-a)^2}\right)(f(a))^{\frac{1}{2}},$$

\npgni where $z=z(t)$ and $u=u(t)$, $u(t)|_{t=0}=a$, $u=\dfrac{1}{r}$, $\dot{r}=(f(u))^{\frac{1}{2}}$, the physics determining the sign of $(f(u))^{\frac{1}{2}}$.
    
\end{corollary}

\npgni For a beautiful account of elliptic curves see McKean and Moll. (Ref. [17]).

\npgni This approach provides a minimal tool kit for further developments. The original references contain all the details elucidated so far. The KLMN problem is just an example of how quantisation of Newtonian gravity can proceed. We notice how dependent the above results are on classical mechanics and calculus, Newton's greatest discoveries.

\end{document}